\documentclass{easychair}

\usepackage{paralist}
\usepackage{amssymb}
\usepackage{amsmath}
\usepackage{mathtools}
\usepackage{todonotes}
\usepackage{algorithm}
\usepackage{algorithmicx}
\usepackage[noend]{algpseudocode}
\usepackage{multicol}
\usepackage{bussproofs}

\newif\ifLongVersion\LongVersiontrue

\ifLongVersion 
\usepackage[createShortEnv,conf={normal}]{proof-at-the-end}
\newenvironment{myTextEnd}{}{}
\else 
\usepackage[createShortEnv,conf={end,no link to proof,restate}]{proof-at-the-end}
\newenvironment{myTextEnd}{\begin{textAtEnd}}{\end{textAtEnd}}
\fi

\newcounter{example}[section]
 \newenvironment{example}[1][]{\refstepcounter{example}\par\medskip
    \noindent \textbf{Example~\theexample. #1} \rmfamily}{\hfill$\blacksquare$\medskip}

\newtheorem{definition}{Definition}
\newtheorem{lemma}{Lemma}
\newtheorem{theorem}{Theorem}
\newtheorem{corollary}{Corollary}

\newtheorem{proposition}{Proposition}

\renewcommand{\paragraph}[1]{\noindent{\bf #1}.}

\newcommand{\nat}{\mathbb{N}}

\newcommand{\arity}{\#}
\newcommand{\arityof}[1]{{\arity{#1}}}

\newcommand{\lenof}[1]{\ell({#1})}

\newcommand{\cardof}[1]{|\!|{#1}|\!|}
\newcommand{\sizeof}[1]{\mathrm{size}({#1})}

\newcommand{\universeOf}[1]{\mathsf{#1}}

\newcommand{\isdef}{\stackrel{\scalebox{0.5}{$\mathsf{def}$}}{=}}

\newcommand{\interv}[2]{[{#1},{#2}]}
\newcommand{\tuple}[1]{\langle {#1} \rangle}
\newcommand{\Tuple}[1]{\left\langle {#1} \right\rangle}
\newcommand{\set}[1]{\{ {#1} \}}
\newcommand{\mset}[1]{\{\!\!\{ {#1} \}\!\!\}}

\newcommand{\lcm}{\mathrm{lcm}}
\newcommand{\Set}[1]{\left\{ {#1} \right\}}

\newcommand{\bigO}{\mathcal{O}}
\newcommand{\twoexptime}{$2\mathsf{EXPTIME}$}

\newcommand{\exptime}{$\mathsf{EXPTIME}$}

\newcommand{\signature}{\mathcal{F}}

\newcommand{\alphabet}{\Sigma}

\newcommand{\arrow}[2]{\xrightarrow{{\scriptscriptstyle #1}}_{{\scriptstyle #2}}}

\newcommand{\auto}[2]{\mathcal{A}_{
    {#1}
    \ifthenelse{\equal{#2}{}}{}{,{#2}}
}}

\newcommand{\autsat}[2]{\mathcal{A}^{\scriptscriptstyle\mathsf{sat}}_{
    {#1}
    \ifthenelse{\equal{#2}{}}{}{,{#2}}
}}

\newcommand{\autcut}[2]{\mathcal{A}^{\scriptscriptstyle\mathsf{cut}}_{
    {#1}
    \ifthenelse{\equal{#2}{}}{}{,{#2}}
}}

\newcommand{\autcst}[2]{\mathcal{A}^{\scriptscriptstyle\mathsf{cst}}_{
    {#1}
    \ifthenelse{\equal{#2}{}}{}{,{#2}}
  }
}

\makeatletter
\newcommand*{\da@rightarrow}{\mathchar"0\hexnumber@\symAMSa 4B }
\newcommand*{\da@leftarrow}{\mathchar"0\hexnumber@\symAMSa 4C }
\newcommand*{\xdashrightarrow}[2][]{%
  \mathrel{%
    \mathpalette{\da@xarrow{#1}{#2}{}\da@rightarrow{\,}{}}{}%
  }%
}
\newcommand{\xdashleftarrow}[2][]{%
  \mathrel{%
    \mathpalette{\da@xarrow{#1}{#2}\da@leftarrow{}{}{\,}}{}%
  }%
}
\newcommand*{\da@xarrow}[7]{%
  \sbox0{$\ifx#7\scriptstyle\scriptscriptstyle\else\scriptstyle\fi#5#1#6\m@th$}%
  \sbox2{$\ifx#7\scriptstyle\scriptscriptstyle\else\scriptstyle\fi#5#2#6\m@th$}%
  \sbox4{$#7\dabar@\m@th$}%
  \dimen@=\wd0 %
  \ifdim\wd2 >\dimen@
    \dimen@=\wd2 %
  \fi
  \count@=2 %
  \def\da@bars{\dabar@\dabar@}%
  \@whiledim\count@\wd4<\dimen@\do{%
    \advance\count@\@ne
    \expandafter\def\expandafter\da@bars\expandafter{%
      \da@bars
      \dabar@
    }%
  }%
  \mathrel{#3}%
  \mathrel{%
    \mathop{\da@bars}\limits
    \ifx\\#1\\%
    \else
      _{\copy0}%
    \fi
    \ifx\\#2\\%
    \else
      ^{\copy2}%
    \fi
  }%
  \mathrel{#4}%
}
\makeatother

\DeclareMathOperator*{\Pop}{\mathop{\scalebox{1.5}{\raisebox{-0.2ex}{$\parallel$}}}\hspace*{1pt}}%

\newcommand{\cmso}{$\mathsf{CMSO}$}

\newcommand{\mso}{$\mathsf{MSO}$}

\renewcommand{\mod}{~\mathrm{mod}~}

\newcommand{\sintfusion}[2]{\widetilde{\mathtt{IF}}({#1}\ifthenelse{\equal{#2}{}}{}{,{#2}})}

\newcommand{\step}[1]{\Rightarrow_{\scriptscriptstyle{#1}}}
\newcommand{\assoc}{\stackrel{\scriptstyle{\mathsf{a}}}{\sim}}

\newcommand{\astep}[1]{\stackrel{\raisebox{-4pt}{$\scriptstyle{\mathsf{a}}~~~~*$}}{\Rightarrow}_{\scriptscriptstyle{#1}}}

\newcommand{\langof}[2]{\mathcal{L}_{#1}({#2})}

\newcounter{index}

\newcommand{\graphof}[3]{\mathrm{subgraph}_{#1}[{#2}
\ifthenelse{\equal{#3}{}}{]}{,{#3}]}}

\newcommand{\graph}{G}

\newcommand{\graphsof}[1]{{\mathcal{G}^{\scriptscriptstyle{#1}}}}
\newcommand{\graphs}{\graphsof{}}

\newcommand{\vertices}{V}
\newcommand{\vertof}[1]{\vertices_{\scriptscriptstyle{#1}}}

\newcommand{\edgeof}[1]{\edges_{\scriptscriptstyle{#1}}}

\newcommand{\sources}{\xi}
\newcommand{\sourceof}[1]{\sources_{\scriptscriptstyle{#1}}}

\newcommand{\algof}[1]{\mathcal{#1}}

\newcommand{\grammar}{\Gamma}

\newcommand{\pump}[1]{{\mathcal{S}}_{#1}^{\pi}}
\newcommand{\pumpx}[1]{{\mathcal{S}}_{#1}^{\pi+}}
\newcommand{\period}[2]{\pi_{#2}\ifthenelse{\equal{#1}{}}{}{({#1})}}
\newcommand{\nopump}[1]{{\mathcal{S}}_{#1}^{\beta}}
\newcommand{\base}[2]{\beta_{#2}\ifthenelse{\equal{#1}{}}{}{({#1})}}

\newcommand{\node}{\mathsf{node}}

\newcommand{\rules}{\mathcal{R}}
\newcommand{\arules}{\rules^{\scriptscriptstyle\text{\ref{it1:def:sp-regular-grammar}}}}
\newcommand{\brules}{\rules^{\scriptscriptstyle\text{\ref{it2:def:sp-regular-grammar}}}}
\newcommand{\crules}{\rules^{\scriptscriptstyle\text{\ref{it3:def:sp-regular-grammar}}}}
\newcommand{\drules}{\rules^{\scriptscriptstyle\text{\ref{it4:def:sp-regular-grammar}}}}
\newcommand{\erules}{\rules^{\scriptscriptstyle\text{\ref{it5:def:sp-regular-grammar}}}}
\newcommand{\frules}{\rules^{\scriptscriptstyle\text{\ref{it6:def:sp-regular-grammar}}}}
\newcommand{\grules}{\rules^\star}
\newcommand{\nonterm}{\mathcal{N}}
\newcommand{\axioms}{\mathcal{X}}

\newcommand{\edges}{{E}}

\newcommand{\unit}[2]{\overline{1}_{
        {#1}
        \ifthenelse{\equal{#2}{}}{}{,{#2}}
}}

\newcommand{\zero}[2]{\overline{0}_{
        {#1}
        \ifthenelse{\equal{#2}{}}{}{,{#2}}
}}

\newcommand{\pop}{\parallel}
\newcommand{\pexp}[2]{{#1}^{\pop{#2}}}

\newcommand{\sop}{\circ}

\newcommand*{\langu}{L}

\let\terms\undefined
\newcommand*{\terms}[4]{\mathfrak{T}_{{#1}\ifthenelse{\equal{#2}{}}{}{,{#2}}\ifthenelse{\equal{#3}{}}{}{,{#3}}}({#4})}
\newcommand{\termsof}[2]{\mathsf{Terms}({#1},{#2})}

\newcommand{\ThresholdVars}{\SVars^{\theta}}
\newcommand{\threshold}[1]{\theta\ifthenelse{\equal{#1}{}}{}{({#1})}}
\newcommand{\SVars}{\mathcal{S}}
\newcommand{\svar}{s}
\newcommand{\svarset}{V}
\newcommand{\varsof}[1]{\mathit{vars}(#1)}
\newcommand{\Sums}[1]{\mathsf{Sums}({#1})}
\newcommand{\XSums}[2]{{#2}\text{-}\mathsf{Sums}({#1})}
\newcommand{\Prods}[1]{\mathsf{Prods}({#1})}
\newcommand{\TermsOf}[1]{\mathsf{Terms}({#1})}
\newcommand{\norm}{\mathrm{nf}}
\newcommand{\maxdeg}{\max\deg}
\newcommand{\leprod}{\preceq}
\newcommand{\leinterv}[2]{\interv{#1}{#2}_\leprod}
\newcommand{\ltprod}{\prec}
\newcommand{\lemon}{\sqsubseteq}
\newcommand{\ltmon}{\sqsubset}
\newcommand{\Func}{K}
\newcommand{\weightedcardof}[1]{|\!|{#1}|\!|_1}
\newcommand{\lengthof}[1]{\mathit{len}({#1})}
\newcommand{\nfcardof}[1]{\cardof{#1}_{\norm}}

\newcommand{\boundAx}{($\mathtt{BA}$)}
\newcommand{\periodAx}{($\mathtt{PA}$)}
\newcommand{\threshAx}{($\mathtt{TA}$)}
\newcommand{\termalg}[1]{{\algof{T}_{#1}}}
\newcommand{\termuniv}[1]{{\universeOf{T}_{#1}}}
\newcommand{\parmap}{\mathtt{par}}
\newcommand{\seqmap}{\mathtt{seq}}
\newcommand{\elem}{x}

\newcommand{\maxarityof}[1]{\alpha({#1})}
\newcommand{\maxtermsizeof}[1]{\theta({#1})}

\newcommand{\pseudo}[3]{\mathcal{E}_{{#1},{#2}}\ifthenelse{\equal{#3}{}}{}{[{#3}]}}

\title{Regular Grammars as Effective Representations of Recognizable Sets of Series-Parallel Graphs}
\author{
Marius Bozga \inst{1}
\and
Radu Iosif \inst{1}
\and
Florian Zuleger \inst{2}\inst{3}
}

\institute{
  Université Grenoble Alpes, Grenoble, France\\
  \email{\{marius.bozga,radu.iosif\}@univ-grenoble-alpes.fr}
  \and
  Technische Universtität Wien, Vienna, Austria\\
  \email{florian.zuleger@tuwien.ac.at}
  \and
  Technische Universtität München, Germany 
}

\authorrunning{Bozga, Iosif, Zuleger}
\titlerunning{Regular Grammars for Series-Parallel Graphs}

\begin{document}
\maketitle

\begin{abstract}
  Series-parallel (SP) graphs are binary edge-labeled graphs with a
  designated source and target vertex, built using serial and parallel
  composition. A set of graphs is recognizable if membership depends
  only on its image under a homomorphism into a finite algebra. For
  SP-graphs, and more generally, for graphs of bounded tree-width,
  recognizability coincides with definability in Counting Monadic
  Second-Order (CMSO) logic. Despite this strong logical
  characterization, the conciseness and algorithmic effectiveness of
  syntactic representations of recognizable sets of SP (and
  bounded-tree-width) graphs remain poorly understood.

  Building on previously introduced regular grammars for SP-graphs, we
  show that recognizable sets admit concise and effective syntactic
  representations. The main contribution is an improved construction
  of finite recognizer algebras whose size is singly-exponential in
  the size of a regular grammar, improving upon the previously known
  double-exponential bound. As a consequence, the problems of
  intersection and language inclusion for sets represented by regular
  grammars are shown to be \exptime-complete, thus improving on a
  previously known \twoexptime{} upper bound.
\end{abstract}

\section{Introduction}

A subset of an algebra is recognizable if it is the inverse image of a
homomorphism into a finite algebra over the same signature of function
symbols. For words and, more generally, for ground terms (i.e.,
node-labeled ranked trees), this definition of recognizability
coincides with the classical one, i.e., existence of finite-state word
and tree automata, respectively. A central aspect of automata-based
recognizability is the predefined traversal order of the input: words
are read either left-to-right or right-to-left and trees are traversed
bottom-up or top-down. Graphs, however, are strictly more complex
objects that, in general, do not come with a natural traversal
order. As a consequence, attempts to generalize automata-based
recognizability to graphs resulted in a variety of non-equivalent and
largely incomparable formalisms, without a broadly accepted
consensus~\cite{DBLP:conf/ifip/Thomas94,Remila1999}. This motivates
the use of the algebraic definition of recognizability instead of the
automata-based one~\cite{CourcelleI,Engelfriet1997} and the quest for
succint syntactic representations of recognizable sets, such as
regular grammars~\cite{Lics25} or regular
expressions~\cite{DBLP:conf/icalp/Doumane22}.

The series-parallel (SP) class consists of graphs having a designated
source and target vertex, that can be composed either in series (i.e.,
the target of the first is joined with source of the second) or in
parallel (i.e., both sources and targets are joined,
respectively). SP-graphs arise naturally in circuit design and
distributed data stream processing (see, e.g.~\cite{10.1145/3571230}
for a survey) and play a central role in graph theory: a classical
result states that a graph has tree-width $2$ at most if and only if
each of its 2-connected\footnote{A graph is $k$-connected if it cannot
  be disconnected by removing $k-1$ vertices, $k\geq2$.} components
(blocks) is a disoriented SP-graph (i.e., a graph obtained from an
SP-graph by reversing some edges). Hence, each graph of tree-width at
most 2 can be represented as a tree of such
blocks~\cite{TutteBook}. This structural characterization has been
used in~\cite{Lics25} to develop regular grammars for tree-width $2$
graphs, by combining regular grammars for trees and SP-graphs. While
the present paper focuses, for simplicity, on SP-graphs, our results
extend naturally to the more general class of graphs of tree-width at
most $2$. We consider this generalization for a future extended
version.

The recent development of regular grammars~\cite{Lics25} and regular
expressions~\cite{DBLP:conf/icalp/Doumane22} that describe the
recognizable sets of the tree-width $\leq2$ class of graphs has left
the following question unanswered: ``\emph{what is the complexity of
  the translation between a syntactic representation (regular
  grammar/expression) and an algebraic representation of a
  recognizable set?}''. For words, it is well-known that the minimal
algebra that recognizes the language of a nondeterministic automaton
with $n$ states, that can be obtained from a regular
grammar/expression in polynomial time, has cardinality
$2^{\bigO(n^2)}$, i.e., each element of the algebra is a set of pairs
of states. Similarly, for ground terms (ranked trees), the translation
of regular grammars/expressions into deterministic tree automata
incurs an unavoidable singly-exponential
blowup~\cite{comon:hal-03367725}.

\paragraph{Contributions}
First, we answer the above question for regular grammars~\cite{Lics25}
that represent recognizable sets of SP-graphs by showing that, for a
regular grammar of size $n$ (comprising the number of nonterminals,
rules and the size of each rule), there exists an algebra of
cardinality $2^{\bigO(n^9)}$ that recognizes its language. This result
is an improvement of the double exponential upper bound found
in~\cite{Lics25}. This upper bound is matched by an exponential lower
bound: there exist a regular grammar having $\bigO(n)$ nonterminals
and $\bigO(n \log_2 n)$ rules, whose recognizer algebra has no less
than $2^{n^2}$ elements.

Second, we apply these complexity results to find the complexity of
several natural decision problems concerning sets of SP-graphs
represented as regular grammars. In particular, the problems of joint
intersection of regular grammars and inclusion of an arbitrary grammar
into a regular one are shown to be \exptime-complete. These results
show that regular SP-grammars can be effectively manipulated within
the same complexity class as standard tree
automata~\cite{comon:hal-03367725}, which justifies future efforts for
implementing our techniques within one of the existing automata
libraries~\cite{Mona,Vata}.

\paragraph{Related work}
Automata recognizing a subclass of SP-graphs, called
\emph{synchronized SP-graphs} (SSPG), have been introduced
in~\cite{10.1145/3571230}. SSPGs have matching split and join vertex
labels, that give automata hints for parsing the input graph. SSPG
automata (SSPGA) are closed under boolean operations and their
inclusion problem is \exptime-complete. The determinization of a
nondeterministic SSPGA incurs a $2^{\Theta(n^2)}$ blowup, where $n$ is
the number of states of the nondeterministic SSPGA. In contrast, we
consider general SP-graphs and the algebraic notion of
recognizability, which naturally extends to (disoriented) SP-graphs
and graphs of tree-width at most $2$.

Nondeterministic automata for unranked trees with unboundedly many
ordered siblings have been introduced by
Thatcher~\cite{THATCHER1967317}. These automata are closed under
determinization. Several definitions of automata on unranked and
unordered trees, equivalent to various fragments of Monadic Second
Order (\mso) logic are given
in~\cite{journals/iandc/BoiretHNT17}. Algebraic recognizers for
unranked (ordered and unordered) trees are also defined
in~\cite{DBLP:conf/birthday/BojanczykW08}, in terms of forest algebras
consisting of a horizontal monoid (disjoint union) and a vertical
monoid (composition of contexts). To the best of our knowledge, the
complexity of converting a nondeterministic automaton for unranked and
unordered trees into a (deterministic) forest algebra has not been
studied. We conjecture that our encoding of parallel composition as
terms over a commutative dioid (\autoref{sec:termalg}) could as well
apply to bound the cardinality of minimal recognizer algebras of
deterministic automata for unranked and unordered trees.

Regular grammars for unranked and unordered trees originate in the
seminal work of Courcelle~\cite{CourcelleV}. A more recent definition
of regular grammars~\cite{Lics25} extends the partitioning of the
nonterminals and the syntax of rules, according to this partition,
from unranked and unordered trees to SP-graphs and graphs of
tree-width $2$ at most. Preliminary double-exponential upper bounds on
the size of the recognizer for the language of a regular grammar are
given, yielding a \twoexptime{} upper bound for the inclusion problem
between an arbitrary and a regular grammar. In contrast, in the
present paper, we give a (fairly) tight singly-exponential upper bound
for the cardinality of the recognizer and show \exptime-completeness
for inclusion and related problems.

Regular expressions that capture the recognizable sets of graphs of
tree-width at most $2$ have been introduced
in~\cite{DBLP:conf/icalp/Doumane22}. This definition relies on a
semantic \emph{guardedness} condition, which is shown to be decidable,
with unknown complexity. The recognizability of the language of a
regular expression is established by translation to Counting Monadic
Second Order logic (\cmso) in~\cite{DBLP:conf/icalp/Doumane22},
leaving the question about the cardinality of the minimal algebra
recognizing the language of a regular expression unanswered.

A natural question is if regular grammars could be extended beyond
graphs of tree-width $2$.  On the one hand, recent results on
canonical decompositions of $3$- and $4$-connected
graphs~\cite{DBLP:conf/soda/KurkofkaP26} suggest that such an
extension might be possible.  On the other hand, regular grammars
capturing the recognizable sets of graphs of bounded \emph{embeddable}
tree-width, an over-approximation of the tree-width that considers
only spanning-tree decompositions, have been proposed
in~\cite{Lpar24}. In this work, the bound on the embeddable tree-width
of graphs is not fixed. However, the recognizability of the language
of a regular grammar is established via translation to \cmso, leaving
the cardinality of the recognizer an open problem.

\paragraph{Definitions}
The set of natural numbers is denoted by $\nat$, and we use
$\nat^{\geq k} \isdef \set{n \in \nat \mid n \ge k}$.  Given numbers
$i,j\in\nat$, we write $\interv{i}{j} \isdef \set{i, i+1, \ldots, j}$,
assumed to be empty if $i>j$. The cardinality of a finite set $A$ is
denoted by $\cardof{A}$. The disjoint union $A \uplus B$ is defined as
the union of $A$ and $B$, if $A \cap B=\emptyset$, and undefined,
otherwise. The size of a syntactic object $x$ (e.g., a term,
expression or grammar) is the length of its representation as a
string, denoted $\sizeof{x}$. A function $f : \nat\rightarrow\nat$ is
exponential if there exist $c,d,k\in\nat^{\geq1}$ such that $f(n) \leq
c \cdot 2^{n^d}$, for all $n \in \nat^{\geq k}$.

\section{Series-Parallel Graphs}
\label{sec:sp}

We consider simple binary graphs over a finite alphabet $\alphabet$ of
edge labels. Formally, a graph is a tuple $\graph =
(\vertof{},\edgeof{},\sourceof{})$, where $\vertof{}$ is a finite set
of \emph{vertices}, $\edgeof{}$ is a multiset of edges from $\vertof{}
\times \alphabet \times \vertof{}$, i.e., multiple edges with the same
label are allowed, and $\sourceof{} : \set{1,2} \rightarrow \vertof{}$
is an injective function that designates the \emph{source}
$\sourceof{}(1)$ and \emph{target} $\sourceof{}(2)$ of $\graph$, i.e.,
the source and target of a graph are always distinct. The components
of a graph $\graph$ are denoted as $\vertof{\graph}$,
$\edgeof{\graph}$ and $\sourceof{\graph}$, respectively. We identify
graphs that are isomorphic, i.e., differ only by a renaming of
vertices. The $a$-\emph{bridge} (or simply bridge, when the label $a$
is not important) is the graph consisting of two vertices and a single
$a$-labeled edge, from source to target.

We consider the following operations on graphs (see Figure
\ref{fig:sp}): \begin{itemize}
\item \emph{serial composition} $\graph_1 \sop \graph_2$ joins the
  target of $\graph_1$ with the source of $\graph_2$. The source and
  target of $\graph_1 \sop \graph_2$ are the source of $\graph_1$ and
  the target of $\graph_2$, respectively.
\item \emph{parallel composition} $\graph_1 \pop \graph_2$ joins the
  source and target of $\graph_1$ and $\graph_2$ into the source and
  target of the result, respectively. $\pexp{\graph}{k}$ denotes the
  $k$-times $\pop$-composition of $\graph$.
\end{itemize}

\begin{figure}[t!]
  \centerline{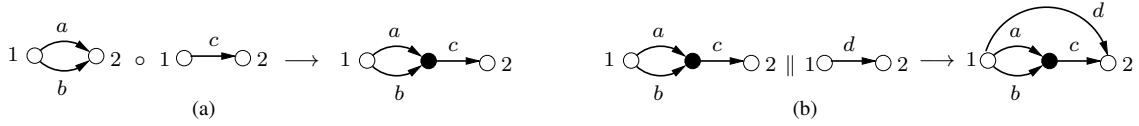}
  \caption{Series (a) and parallel (b) composition.
  }
  \label{fig:sp}
\end{figure}
\noindent
The set of \emph{series-parallel} (SP-) graphs is the closure of
bridges under the above compositions. An SP-graph $\graph$ is
\emph{atomic} if there are no graphs $\graph_1$ and $\graph_2$ such
that $\graph=\graph_1\sop\graph_2$. A \emph{P-graph} is an atomic
graph that is not a bridge and an \emph{S-graph} is a bridge or a
non-atomic graph. A P-graph is the parallel composition of at least
two S-graphs, and a non-bridge S-graph is the serial composition of at
least two SP-graphs, either bridges or P-graphs~\cite[Lemma
  6.3]{CourcelleV}.

We define recognizable sets of SP-graphs by introducing first their
underlying algebra. As usual, an algebra $\algof{A} =
(\universeOf{A},\set{f^\algof{A}}_{f\in\signature})$ consists of a
domain $\universeOf{A}$ and interpretations $f^\algof{A}$ of a
signature $\signature$ of functions symbols $f$, of arity
$\arityof{f}\in\nat$. An algebra is finite if its domain is finite. A
term $t[x_1,\ldots,x_n]$ is built as usual from function symbols of
matching arities and variables $x_i$. $\termsof{\signature}{X}$
denotes the set of terms with variables from $X$. A \emph{constant} is
a term consisting of a single function symbol of zero arity. A
\emph{ground} term $t$ has no variables and evaluates to an element
$t^\algof{A} \in \universeOf{A}$, obtained by interpreting the
function symbols from $t$ according to $\algof{A}$. A
\emph{homomorphism} between algebras $\algof{A}$ and $\algof{B}$ is a
mapping $h : \universeOf{A} \rightarrow \universeOf{B}$ such that
$h(f^\algof{A}(a_1,\ldots,a_n)) = f^\algof{B}(h(a_1),\ldots,h(a_n))$,
for all $f \in \signature$, $\arityof{f}=n$ and $a_1,\ldots,a_n \in
\universeOf{A}$.

\begin{definition}\label{def:rec}
  Let $\algof{A}$ be an algebra. A set $\langu \subseteq
  \universeOf{A}$ is \emph{recognizable} in $\algof{A}$ if there
  exists a finite algebra $\algof{B}$ and a homomorphism $h$ between
  $\algof{A}$ and $\algof{B}$ such that $\langu=h^{-1}(F)$, where $F
  \subseteq \universeOf{B}$ is called the \emph{accepting set} of
  $\langu$.
\end{definition}
We say that $\langu$ is \emph{recognized} by $\algof{B}$, or that
$\algof{B}$ is a \emph{recognizer} for $\langu$, when $\algof{A}$ is
understood. It is known that, in general, recognizable sets enjoy the
following boolean closure properties:

\begin{proposition}[Propositions 2.9 and 2.10 in~\cite{JEPBook}]\label{prop:rec-bool-closure}
  Let $\algof{A} =
  (\universeOf{A},\set{f^\algof{A}}_{f\in\signature})$ be an algebra
  and $\mathcal{L}, \mathcal{K} \subseteq \universeOf{A}$ be sets
  recognized by algebras $\algof{B}$ and $\algof{C}$, with accepting
  sets $F\subseteq\universeOf{B}$ and $G\subseteq\universeOf{C}$,
  respectively. Then, $\universeOf{A} \setminus \mathcal{L}$ is
  recognized by $\algof{B}$ with accepting set $\universeOf{B}
  \setminus F$ and $\mathcal{L} \cap \mathcal{K}$, $\mathcal{L} \cup
  \mathcal{K}$ are recognized by $\algof{B} \times \algof{C}$ with
  accepting sets $F \times G$ and $(F \times \universeOf{C}) \cup
  (\universeOf{B} \times G)$, respectively.
\end{proposition}

In particular, the signature of the algebra $\algof{SP}$ of SP-graphs
is the set $\signature_{\algof{SP}} \isdef \set{\sop,\pop} \uplus
\alphabet$, where $\sop$ and $\pop$ are binary function symbols,
interpreted as the serial and parallel composition, and each $a \in
\alphabet$ is a constant, interpreted as the $a$-bridge, the domain of
the algebra $\algof{SP}$ being the set of SP-graphs (up to
isomorphism). For instance, the terms $(a \pop b) \sop c$ and $((a
\pop b) \sop c) \pop d$ are interpreted in $\algof{SP}$ as the
rightmost graphs from Figure \ref{fig:sp} (a) and (b), respectively. A
set of SP-graphs is recognizable if it is recognizable in the
$\algof{SP}$ algebra (Definition \ref{def:rec}).

\section{Recognizability of Regular Languages}
\label{sec:rec}

This section recalls regular grammars for SP-graphs~\cite{Lics25}.  In
particular, we establish that languages defined by regular SP-grammars
are recognizable, providing the explicit construction of a recognizer
algebra.  We note that the existence of such an algebra in itself is
not a new result (see Theorem~13 in the extended
version~\cite{Lics25Arxiv} of~\cite{Lics25}).  However, we present a
new construction, establishing an exponential bound on its cardinality
(\autoref{sec:complexity}), improving the earlier doubly-exponential
construction.  We further note that the converse direction (i.e., that
each recognizable set of SP-graphs is the language of some regular
grammar) has been proved in \cite[Theorem 3]{Lics25}.

\subsection{Regular Grammars}

We introduce regular grammars, as representations of the recognizable
sets of SP-graphs. A \emph{grammar}
$\grammar=(\nonterm,\rules,\axioms)$ consists of a finite set
$\nonterm=\set{x_1,\ldots,x_n}$ of \emph{nonterminals}, a set $\rules
\subseteq \nonterm \times \termsof{\signature_{\algof{SP}}}{\nonterm}$
of rules, written $x \rightarrow t[x_1,\ldots,x_k]$, where
$x,x_1,\ldots,x_k \in \nonterm$ and a set of \emph{axioms} $\axioms
\subseteq \nonterm$. A derivation step is a pair of terms, written $u
\step{\grammar} v$, such that $v$ is obtained by replacing a single
occurrence of a nonterminal $x$ in $u$ by a term $t$, for some rule $x
\rightarrow t$ from $\grammar$. A derivation $x \step{\grammar}^* t$
is a possibly empty sequence of derivation steps starting with a
nonterminal. A derivation is \emph{complete} if $t$ is a ground
term. The \emph{language} of $\grammar$ is the set of evaluations of
ground terms resulting from the complete derivations of $\grammar$
that start with an axiom, i.e., $\langof{}{\grammar} \isdef
\set{t^\algof{A} \mid x \step{\grammar}^* t,~ t \in
  \termsof{\signature_\algof{SP}}{\emptyset},~x \in \axioms}$.  A
\emph{context-free} set is the language of a grammar. The following
definition was introduced in~\cite{Lics25}:

\begin{definition}\label{def:sp-regular-grammar}
  A grammar $\grammar = (\mathcal{P} \uplus
  \mathcal{S},\rules,\axioms)$ is said to be \emph{regular} if its
  nonterminals are partitioned into $\mathcal{S}$ and $\mathcal{P}$
  and its rules are one of the following forms, for some nonterminals
  $p,p_1,p_2 \in \mathcal{P}$, $s,s_1,\ldots,s_k \in \mathcal{S}$ and
  labels $a \in \alphabet$: \begin{enumerate}[(A)]
  \item\label{it1:def:sp-regular-grammar} $p \rightarrow p \pop
    \pexp{s}{\ell}$, for $\ell\in\nat^{\geq1}$,
  \item\label{it2:def:sp-regular-grammar} $p \rightarrow
    \pexp{s_1}{\ell_1} \pop \ldots \pop \pexp{s_k}{\ell_k}$, for
    $k,\ell_1,\ldots,\ell_k \in \nat^{\geq1}$, $\sum_{i=1}^k
    \ell_i \geq 2$,
  \item\label{it3:def:sp-regular-grammar} $s \rightarrow p \sop s_1$,
  \item\label{it4:def:sp-regular-grammar} $s \rightarrow p_1 \sop p_2$,
  \item\label{it5:def:sp-regular-grammar} $p \rightarrow a$,
  \item\label{it6:def:sp-regular-grammar} $s \rightarrow a$.
  \end{enumerate}
  We denote by $\arules$, $\brules$, $\crules$, $\drules$, $\erules$
  and $\frules$ the subsets of $\rules$ consisting rules of the form
  (\ref{it1:def:sp-regular-grammar}),
  (\ref{it2:def:sp-regular-grammar}),
  (\ref{it3:def:sp-regular-grammar}),
  (\ref{it4:def:sp-regular-grammar}),
  (\ref{it5:def:sp-regular-grammar}) and
  (\ref{it6:def:sp-regular-grammar}), respectively. A set $\langu$ of
  SP-graphs is \emph{regular} if there exists a regular grammar
  $\grammar$ such that $\langu=\langof{}{\grammar}$.
\end{definition}

\begin{example}\label{ex:universal-sp-grammar}
  Consider the regular grammar:
  \[\grammar_{\mathit{univ}} \isdef
  (\underbrace{\set{p}}_{\mathcal{P}}\uplus\underbrace{\set{s}}_{\mathcal{S}},\rules,\set{p,s})\]
  having the following rules:
  \[\begin{array}{rlcrl}
    p \rightarrow & p \pop s & \hspace*{5mm} & s \rightarrow & p \sop s \\
    p \rightarrow & s \pop s && s \rightarrow & p \sop p \\
    p \rightarrow & a && s \rightarrow & a \text{, for each } a \in \alphabet
  \end{array}\]
  The proof of $\langof{}{\grammar_{\mathit{univ}}}=\universeOf{SP}$
  can be found in~\cite[Lemma 24]{Lics25Arxiv}.  In contrast, a
  grammar having rules of the form $p \rightarrow p' \pop s$ or $s
  \rightarrow s_1 \sop s_2$, for some $p,p' \in \mathcal{P}$ and
  $s,s_1,s_2 \in \mathcal{S}$ is not regular.
\end{example}

It is already known that, from a regular grammar $\grammar$, one can
build a finite algebra $\algof{A}$, of doubly-exponential cardinality
in the size of $\grammar$, that recognizes
$\langof{}{\grammar}$~\cite[Theorem 13]{Lics25Arxiv}. The remainder of
this section gives a different construction, whose purpose is to show
a singly exponential upper bound on the cardinality of the minimal
recognizer of $\langof{}{\grammar}$. The proof that
$\cardof{\universeOf{A}}$ is exponential in $\sizeof{\grammar}$ will
be given in \autoref{sec:complexity}.

First, we introduce a few technical notions. Given a regular grammar
$\grammar=(\mathcal{P}\uplus\mathcal{S},\rules,\axioms)$ and a
nonterminal $p \in \mathcal{P}$, a nonterminal $s \in \mathcal{S}$ is
said to be: \begin{itemize}
\item \emph{periodic} for $p$ if $\grammar$ contains a rule $p
  \rightarrow p \pop \pexp{s}{k} \in \arules$. If there is only one
  such rule in $\grammar$, we say that $k$ is the \emph{period} of $s$
  for $p$, denoted $\period{s}{p}$.
\item \emph{bounded} for $p$ if $s=s_i$ for some rule $p \rightarrow
  \pexp{s_1}{\ell_1} \pop \ldots \pop \pexp{s_k}{\ell_k} \in \brules$
  from $\grammar$ and some $1 \leq i \leq k$, where $s_1, \ldots, s_k$
  are pairwise distinct variables. The \emph{bound} of $s$ for $p$,
  denoted $\base{s}{p}$, is the maximum integer $\ell$ among the
  occurrences of $\pexp{s}{\ell}$ in a rule from $\brules$ with
  left-hand side $p$, plus one.
\end{itemize}
We denote by $\pump{p}$ (resp. $\nopump{p}$) the subset of
$\mathcal{S}$ consisting of periodic (resp. bounded) nonterminals for
$p$, whenever the grammar in question is understood.  Note that
$\pump{p}$ and $\nopump{p}$ are not disjoint and do not cover
$\mathcal{S}$, in general.

\begin{example}\label{ex:base-period}
  In the grammar $\grammar_{\mathit{univ}}$ (Example
  \ref{ex:universal-sp-grammar}), $\pump{p}=\nopump{p}=\set{s}$,
  $\period{s}{p}=1$ and $\base{s}{p}=2$.
\end{example}

We can however assume w.l.o.g. that $\pump{p} \cap \nopump{p} =
\emptyset$, by putting each regular grammar in the following normal
form:

\begin{definition}\label{def:sp-normal-form}
  A regular grammar $\grammar =
  (\mathcal{P}\uplus\mathcal{S},\rules,\axioms)$ is in \emph{normal
    form (normalized)} if and only if the following
  hold: \begin{enumerate}
  \item\label{it1:def:sp-normal-form} for each pair
    $(p,s)\in\mathcal{P}\times\mathcal{S}$ there is at most one rule
    $p \rightarrow p \pop \pexp{s}{\ell} \in \arules$,
  \item\label{it2:def:sp-normal-form} $s_1, \ldots, s_k \in
    \nopump{p}\setminus\pump{p}$, for each rule $p \rightarrow
    \pexp{s_1}{\ell_1} \pop \ldots \pop \pexp{s_k}{\ell_k} \in
    \brules$, $k\geq1$.
  \end{enumerate}
\end{definition}

We show that each grammar can be normalized at the cost of a
polynomial increase in the size of the grammar:

\begin{propositionE}\label{prop:sp-normal-form}
  For each regular grammar $\grammar =
  (\mathcal{P}\uplus\mathcal{S},\rules,\axioms)$, one can build, in
  polynomial time, a normalized regular grammar $\grammar' =
  (\mathcal{P}\uplus\mathcal{S}',\rules',\axioms)$, such that
  $\langof{}{\grammar} = \langof{}{\grammar'}$. 
\end{propositionE}
\begin{proofE}
  We transform the grammar $\grammar$ as follows to take care of the
  two conditions of Definition \ref{def:sp-normal-form}. Let $p
  \rightarrow p \pop \pexp{s}{\ell_1}$, $\ldots$, $p \rightarrow p
  \pop \pexp{s}{\ell_n} \in \arules$, for some $n\geq1$, be the rules
  that involve the same pair $(p,s) \in
  \mathcal{P}\times\mathcal{S}$. We replace these rules by $p
  \rightarrow p \pop \pexp{s_1}{\ell_1}$, $\ldots$, $p \rightarrow p
  \pop \pexp{s_n}{\ell_n}$ and add new rules $\set{s_i \rightarrow t
    \mid s \rightarrow t \in \rules}$, for each $1 \leq i \leq n$,
  where $s_1,\ldots,s_n \not\in \mathcal{S}$ are fresh
  nonterminals. This transformation takes polynomial time, because
  there are at most $\cardof{\arules}$ rules of type
  (\ref{it1:def:sp-regular-grammar}) and at most $\cardof{\rules}$ new
  rules to be introduced, for each pair. Moreover, the transformation
  takes care of condition (\ref{it1:def:sp-normal-form}), because in
  the output grammar there is at most one rule of type
  (\ref{it1:def:sp-regular-grammar}), for each pair $(p,s) \in
  \mathcal{P} \times \mathcal{S}$. For the condition
  (\ref{it2:def:sp-normal-form}), let $p \rightarrow
  \pexp{s_1}{\ell_1} \pop \ldots \pop \pexp{s_k}{\ell_k} \in \brules$
  be a rule and note that $s_1, \ldots, s_k \not\in \pump{p}$, because
  the above rule must be present in the original grammar, i.e., the
  transformation does not introduce rules of the form
  (\ref{it2:def:sp-regular-grammar}).
\end{proofE}

\begin{example}\label{ex:normalized-grammar}
  Consider the grammar $\grammar_{\mathit{univ}}$ (Example
  \ref{ex:universal-sp-grammar}). The normalized grammar:
  \[\grammar'_{\mathit{univ}} = (\underbrace{\set{p}}_{\mathcal{P}} \uplus \underbrace{\set{s,s'}}_{\mathcal{S}},\rules',\set{p,s})\]
  is obtained from $\grammar_{\mathit{univ}}$ by replacing the rule $p
  \rightarrow s \pop s$ by the rules $p \rightarrow s' \pop s'$, $s'
  \rightarrow p \sop p$, $s' \rightarrow p \sop s$ and $s' \rightarrow
  a$, for all $a \in \alphabet$. This construction, used in the proof
  of Proposition \ref{prop:sp-normal-form}, implies that
  $\langof{}{\grammar'_{\mathit{univ}}} =
  \langof{}{\grammar_{\mathit{univ}}}$.
\end{example}

\subsection{A Finite Algebra of Terms}
\label{sec:termalg}

In this section, we assume w.l.o.g. that
$\grammar=(\mathcal{P}\uplus\mathcal{S},\rules,\axioms)$ is a fixed
regular grammar in normal form. The definition of a recognizer for the
language of a regular grammar given in~\cite{Lics25Arxiv,Lics25}
represents graphs obtained by parallel composition as sets of
multisets of nonterminals from $\mathcal{S}$, where the inner level
(multisets) counts the number of parallel components up to a certain
period and the outer level (sets) acounts for the nondeterminism in
the grammar. In the following, we use an equivalent encoding of sets
of multisets using terms built from addition, multiplication and
nonterminals from $\mathcal{S}$. For instance, the set of multisets
$\set{\mset{s_1,s_2,s_2},\mset{s_1,s_1,s_3}}$ is unambiguously encoded
by the term $s_1 \cdot s_2^2 + s_1^2 \cdot s_3$.

Formally, the definition of a recognizer for $\langof{}{\grammar}$
uses finite term algebras $\termalg{p}$, one for each nonterminal $p
\in \mathcal{P}$. These algebras have the same signature
$\signature_\termalg{} = \set{+,\cdot,0,1}$, where $+$, $\cdot$ are
binary operations and $0$, $1$ are constants. The domain of the
algebra $\termalg{p} = (\termuniv{p}, +, \cdot, 0, 1)$ is the set of
terms built using the function symbols $+$ and $\cdot$ using variables
$\nopump{p} \uplus \pump{p} \uplus \ThresholdVars$, where $\nopump{p},
\pump{p} \subseteq \mathcal{S}$ are the sets of bounded and periodic
nonterminals for the given $p \in \mathcal{P}$, respectively, and
$\ThresholdVars$ is a set of \emph{threshold variables}, distinct from
$\mathcal{S}$. Note that $\nopump{p}$ and $\pump{p}$ are disjoint, by
the assumption that $\grammar$ is in normal form (Definition
\ref{def:sp-normal-form}). The bound $\base{}{p} : \nopump{p}
\rightarrow \nat^{\geq2}$ and period $\period{}{p} : \pump{p}
\rightarrow \nat^{\geq1}$ mappings have been defined previously. The
threshold variables are artifacts introduced for the purposes of the
upcoming complexity proof (\autoref{sec:complexity}) and do not belong
to $\grammar$. We consider a threshold mapping $\threshold{} :
\ThresholdVars \rightarrow \nat^{\geq2}$.

Let $p \in \mathcal{P}$ be a fixed nonterminal in the following.  The
interpretation of $\signature_\termalg{}$ in $\termalg{p}$ forms an
idempotent commutative semiring\footnote{Often also called a
  commutative dioid.}, where $(\termuniv{p},+,0)$ and
$(\termuniv{p},\cdot,1)$ are commutative monoids, $+$ is idempotent,
$0$ is cancelling for $\cdot$, and $\cdot$ distributes over $+$. The
finiteness of $\termalg{p}$ is ensured by the following additional
axioms:
\begin{align*}
  \text{\boundAx}~ \svar^{\base{\svar}{p}} = 0,~\svar \in \nopump{p} \hspace*{5mm}
  \text{\periodAx}~ \svar^{\period{\svar}{p}} = 1,~\svar \in \pump{p} \hspace*{5mm}
  \text{\threshAx}~ \svar^{\threshold{\svar}} = \svar^{\threshold{\svar}-1},~\svar \in \ThresholdVars
\end{align*}
where by $s^n$ we denote the $n$-times product of $s$ with itself, as
usual.

A \emph{monomial} is a term of the form $s_1 \cdot \ldots \cdot s_n$,
where $s_1,\ldots,s_n \in \nopump{p} \uplus \pump{p} \uplus
\ThresholdVars$. When no confusion arises, we write $s_1s_2$ instead
of $s_1 \cdot s_2$. A monomial is \emph{reduced} if none of the
\boundAx, \periodAx{} and \threshAx{} axioms can be applied to reduce
the number of occurrences of variables in it. For a monomial $m$, we
denote by $\varsof{m}$ the \emph{support} of $m$, i.e., the set of
variables occurring in $m$. For a variable $\svar \in \varsof{m}$, we
denote by $\deg(m,\svar)$ the degree of $\svar$ in $m$ and by
$\deg(m)$ the degree of $m$, defined as $\deg(m) \isdef \sum_{\svar
  \in \varsof{m}} \deg(m,\svar)$. A term is \emph{linear} if it is a
sum $s_1 + \ldots + s_n$ of variables. A \emph{linear product} is a
term of the form $\ell_1 \cdot \ldots \cdot \ell_n$, where $\ell_i$
are linear terms.

Due to the axioms of distributivity, each term $t \in \termuniv{p}$ is
equivalent to a sum of distinct reduced monomials, called \emph{normal
  form}, denoted $\norm_p(t)$. 
When no confusion arises, by writing $m \in \norm_p(t)$ we mean that
the monomial $m$ occurs as a summand in the normal form of $t$. Note
that the normal form of a term is unique modulo associativity and
commutativity of the binary operations and the set $\set{\norm_p(t)
  \mid t \in \termuniv{p}}$ is finite. In the following, we shall need an
upper bound on the number of monomials from a term in normal form:

\begin{lemmaE}\label{lemma:match}
  Let $t \in \termsof{\signature_{\termalg{}}}{\mathcal{S}}$ be a term
  in normal form. Then each monomial $m \in t$ is of size at most
  $\maxtermsizeof{\grammar}\cdot\cardof{\mathcal{S}}$ and there are at
  most $\maxtermsizeof{\grammar}^{\cardof{\mathcal{S}}}$ such
  monomials.
\end{lemmaE}
\begin{proofE}
  Each monomial $m \in t$ is a product of at most $\sum_{s \in
    \nopump{p}} \base{s}{p} + \sum_{s \in \pump{s}} \period{s}{p}$
  variables and $t$ has at most $\prod_{s\in\nopump{p}} \base{s}{p}
  \cdot \prod_{s\in\pump{p}} \period{s}{p}$ monomials.
\end{proofE}

\begin{myTextEnd}
\begin{lemma}\label{lemma:nf-nesting}
  For each $p \in \mathcal{P}$ and monomials $m_1,m_2
  \in\termuniv{p}$, we have $\norm_p(m_1\cdot m_2) =
  \norm_p(\norm_p(m_1)\cdot \norm_p(m_2))$.
\end{lemma}
\begin{proof}
  Let $\svar \in \varsof{m_1 \cdot m_2} = \varsof{m_1} \cup
  \varsof{m_2}$ be a nonterminal. We consider two cases. First, assume
  that $\svar\in\nopump{p}$. If $\deg(m_1\cdot
  m_2,\svar)<\base{\svar}{p}$ then $\deg(m_1\cdot m_2,\svar) =
  \deg(\norm_p(m_1 \cdot m_2),\svar) = \deg(\norm_p(\norm_p(m_1) \cdot
  \norm_p(m_2)))$. Else, if $\deg(m_1\cdot
  m_2,\svar)\geq\base{\svar}{p}$ then $\norm_p(m_1\cdot m_2)=0$,
  because $\norm_p(m_1 \cdot m_2)$ is obtained from $m_1 \cdot m_2$ by
  applying the \boundAx{} axiom. Again, we consider two cases. If
  $\deg(m_1,\svar) \geq \base{\svar}{p}$ or $\deg(m_2,\svar) \geq
  \base{\svar}{p}$, then $\norm_p(m_1) = 0$, hence $\norm_p(m_1) \cdot
  \norm_p(m_2) = \norm_p(\norm_p(m_1) \cdot \norm_p(m_2)) = 0$. Else,
  $\deg(m_1,\svar) < \base{\svar}{p}$, $\deg(m_2,\svar) <
  \base{\svar}{p}$ and $\deg(m_1,\svar)+\deg(m_2,\svar) \geq
  \base{\svar}{p}$, then $\norm_p(\norm_p(m_1) \cdot \norm_p(m_2)) =
  0$.

  Second, assume that $\svar\in\pump{p}$. By \periodAx, we obtain, for each monomial $m$:
  \begin{align}\label{eq:nf-mod-base}
    \deg(m, \svar) \equiv & ~\deg(\norm_p(m), \svar) \mod \base{\svar}{p}
  \end{align}
  Using the properties of modular addition, we compute:
  \begin{align*}
    \deg(m_1 \cdot m_2, s) \\
    \equiv \deg(\norm_p(m_1 \cdot m_2), s) \mod \base{\svar}{p}
    \\\\
    \deg(m_1,s) + \deg(m_2,s) \\
    \equiv \deg(\norm_p(m_1),s) + \deg(\norm_p(m_2),s) \mod \base{\svar}{p} \\
    \equiv \deg(\norm_p(m_1) \cdot \norm_p(m_2), s) \mod \base{\svar}{p} \\
    \equiv \deg(\norm_p(\norm_p(m_1) \cdot \norm_p(m_2)), s) \mod \base{\svar}{p}
  \end{align*}
  Hence, $\deg(\norm_p(m_1 \cdot m_2), s) = \deg(\norm_p(\norm_p(m_1)
  \cdot \norm_p(m_2)), s)$ because $\deg(m_1 \cdot m_2, s) =
  \deg(m_1,s) + \deg(m_2,s)$, $\deg(\norm_p(m_1 \cdot m_2), s) <
  \base{\svar}{p}$ and $\deg(\norm_p(\norm_p(m_1) \cdot \norm_p(m_2)),
  s) < \base{\svar}{p}$. Since the choice of $\svar\in\varsof{m_1
    \cdot m_2}$ was arbitrary, we obtain $\norm_p(m_1\cdot m_2) =
  \norm_p(\norm_p(m_1)\cdot \norm_p(m_2))$.
\end{proof}
\end{myTextEnd}

\subsection{The Recognizer of a Regular Set}
\label{subsec:regular-languages-are-recognizable}

We define a finite algebra $\algof{A} = (\universeOf{A},
\sop^\algof{A},\pop^\algof{A},\set{a^{\algof{A}}}_{a \in \alphabet})$
and a homomorphism $h$ between $\algof{SP}$ and $\algof{A}$, such that
$\langof{}{\grammar} = h^{-1}(F)$, where $F\isdef
h(\langof{}{\grammar})$, thus establishing that $\langof{}{\grammar}$
is recognizable. Even though the definitions of $\algof{A}$ and $h$
depend on $\grammar$, we omit mentioning it, to avoid clutter.

Intuitively, a P-graph (i.e., a parallel composition of at least two
graphs) is represented by a tuple of terms $\tuple{t_p}_{p \in
  \mathcal{P}}$, where $t_p \in \termuniv{p}$, for each $p \in
\mathcal{P}$. Each term $\norm_p(t_p)$ is a sum of monomials, that
represent the possible ways in which the P-graph is parsed by
$\grammar$ starting from $p$, using only rules from
$\arules\cup\brules$ (modulo repetitions of rules from
$\arules$). This motivates the following notions of \emph{view}, i.e., a
monomial that represents a set of derivations having a common prefix
of the form $p \step{\grammar}^* s_1 \pop \ldots \pop s_n$, and
\emph{profile}, i.e., a sum of monomials that accounts for the
nondeterminism in $\grammar$:

\begin{definition}\label{def:p-view}
  Let $\graph$ be a P-graph. A \emph{view} of $\graph$ is a monomial
  $m = s_1 \cdot \ldots \cdot s_n$, for which there exist complete
  derivations $s_i \step{\grammar}^* t_i$, for all $i \in
  \interv{1}{n}$, such that $\graph = (t_1 \pop \ldots \pop
  t_n)^\algof{SP}$. For a nonterminal $p \in \mathcal{P}$, the
  $p$-profile of $\graph$ is the term $h_p(\graph) \isdef
  \norm_p\left(\sum_{m \text{ view of } \graph} m\right)$. The
  \emph{profile} of $\graph$ is the tuple $h(\graph) \isdef
  \tuple{h_p(\graph)}_{p \in \mathcal{P}}$.
\end{definition}

A view of an S-graph (i.e., either a bridge or a sequential
composition of least two graphs) is a pair $(s,q)$, where $s \in
\mathcal{S}$ represents the beginning and $q \in
\mathcal{P}\uplus\mathcal{S}\uplus\set{\bot}$ the end of a derivation
$s \step{\grammar}^* p_1 \sop \ldots p_{n} \sop q$ (resp. $s
\step{\grammar}^* p_1 \sop \ldots p_{n}$ if $q = \bot$) using only
rules from $\crules\cup\drules$, for $p_1,\ldots,p_{n} \in
\mathcal{P}$. This representation of S-graphs is similar to the
classical representation of words by pairs of states in a finite
word automaton.

For two terms $t$ and $u$, over the $\signature_\algof{SP}$ signature,
we denote by $t \assoc u$ the fact that $t$ can be rewritten into $u$
by applying zero or more times the associativity law for $\sop$. For a
grammar $\grammar$, we write $t \astep{\grammar} u$ if $t
\step{\grammar}^* w$ and $w \assoc u$. The profile of an S-graph is a
set of views representing the nondeterminism in $\grammar$:

\begin{definition}\label{def:s-view}
  Let $\graph$ be an S-graph. A \emph{view} of $\graph$ is a pair
  $(s,q) \in \mathcal{S} \times (\mathcal{S} \uplus \mathcal{P} \uplus
  \set{\bot})$, for which there exists a derivation, either $s
  \astep{\grammar} t \sop q$, if $q \neq \bot$, or $s \astep{\grammar}
  t$, otherwise, and $t$ is a ground term such that $\graph =
  t^\algof{SP}$. The \emph{profile} of $\graph$ is the relation
  $h(\graph) \isdef \set{(s,q) \in \mathcal{S} \times (\mathcal{S}
    \uplus \mathcal{P} \uplus \set{\bot}) \mid (s,q) \text{ is a view
      of } \graph}$.
\end{definition}

The definition of $\algof{A}$ benefits from the notion of
\emph{alternative} grammar, that differs from regular grammar only in
that we allow rules of the form $p \rightarrow s$ (i.e., in violation
of condition (\ref{it2:def:sp-regular-grammar}) of Definition
\ref{def:sp-regular-grammar}), where $s \in \mathcal{S}$ occurs in a
rule $s \rightarrow a$, $a \in \alphabet$ and nowhere else.  This
allows us to represent the bridges that may occur in the parallel
composition of a P-graph using nonterminals from $\mathcal{S}$, which
provides a uniform definition of views for P-graphs.  Each regular
grammar can be transformed into a language-equivalent alternative
grammar by adding a rule $s_a \rightarrow a$ and replacing each rule
$p \rightarrow a$ with $p \rightarrow s_a$, for each $a \in
\alphabet$. In the following, we denote by $\grules$ the set of newly
introduced rules of this form, in an alternative grammar. If,
moreover, $p \in \axioms \cap \mathcal{P}$ is an axiom of $\grammar$,
we promote $s_a$ to be axiom of the alternative grammar.  In the
following, we consider w.l.o.g. $\grammar$ to be an alternative
grammar in normal form.

\begin{myTextEnd}
\begin{lemma}[Lemma 25 in ~\cite{Lics25Arxiv}]\label{lemma:assoc-sp-start}
  Let $s \in \mathcal{S}$ and $q \in \mathcal{P} \uplus \mathcal{S}$
  be nonterminals.  Then, for every derivation $s \astep{\grammar} u
  \sop q$, there are non-terminals $p_1, \ldots, p_n$, derivations
  $p_i \step{\grammar}^* v_i$ and rules $s_i \rightarrow p_i \sop
  s_{i+1}$, for $1 \le i < n-1$, with $s=s_1$, and a rule $s_{n-1}
  \rightarrow p_n \sop q$ such that $v_1 \sop (v_2 \sop \ldots (v_n)
  \ldots) \assoc u$.
\end{lemma}

\begin{lemma}[Lemma 26 in ~\cite{Lics25Arxiv}]\label{lemma:assoc-sp}
  Let $\grammar = (\mathcal{P}\uplus\mathcal{S},\rules,\axioms)$ be a
  regular grammar, $\graph = \graph_1 \sop \graph_2$ be an S-graph, $s
  \in \mathcal{S}$ and $q \in \mathcal{P}\uplus\mathcal{S}$ be
  nonterminals. The following are equivalent: \begin{enumerate}
  \item $s \astep{\grammar} t \sop q$,
    where $t$ is a ground term such that $t^\algof{SP} = \graph$,
  \item there exists $s'\in\mathcal{S}$ and derivations $s
    \astep{\grammar} t_1 \sop s'$ and $s' \astep{\grammar} t_2
    \sop q$, where $t_i$ is a ground term such that $t_i^\algof{SP} =
    \graph_i$, for $i=1,2$.
  \end{enumerate}
\end{lemma}
\end{myTextEnd}

The domain of the recognizer algebra $\algof{A}$ is $\universeOf{A}
\isdef \set{h(\graph) \mid \graph \in \universeOf{SP}}$. We denote by
$\universeOf{A}^P$ and $\universeOf{A}^S$ the sets of P-profiles
(i.e., profiles of P-graphs) and S-profiles (i.e., profiles of
S-graphs), respectively. Note that $\universeOf{A}^P \uplus
\universeOf{A}^S = \universeOf{A}$. For technical convenience, in the
definition of $\algof{A}$, we consider the mappings $\parmap$, that
converts an arbitrary profile into a tuple of terms in normal form,
and $\seqmap$, that converts an arbitrary profile into a relation
between $\mathcal{P} \uplus \mathcal{S}$ and $\mathcal{P} \uplus \mathcal{S} \uplus \set{\bot}$:
\begin{align*}
  \parmap(\elem) \isdef & ~\begin{cases}
    \elem & \text{if } \elem \in \universeOf{A}^P \\
    \Tuple{\norm_p\left(\sum\set{s \mid
        (s,\bot) \in \elem,~
        s \in \nopump{p} \uplus \pump{p}
        }\right)}_{p \in \mathcal{P}} & \text{if } \elem \in \universeOf{A}^S
  \end{cases}
  \\
  \seqmap(\elem) \isdef & ~\begin{cases}
    \elem & \text{if } \elem \in \universeOf{A}^S \\
    \set{(p,\bot) \mid p \leadsto m \in t_p,~ p \in \mathcal{P}} \cup \set{(s,q) \mid s \rightarrow p \sop q \in \rules}
    & \text{if } \elem = \tuple{t_p}_{p \in \mathcal{P}} \in \universeOf{A}^P
  \end{cases}
\end{align*}
where $\tuple{\elem}_p$ denotes the $p$-th component of a tuple $\elem \in \universeOf{A}^P$ and
$\norm(\tuple{t_p}_{p \in \mathcal{P}}) \isdef \tuple{\norm_p(t_p)}_{p \in \mathcal{P}}$ denotes the component-wise normal form.
Further, by writing $p \leadsto m$, we mean the existence of a rule $p
\rightarrow \pexp{s_1}{\ell_1} \pop \ldots \pop \pexp{s_n}{\ell_n} \in
\brules$ such that $m = s_1^{\ell_1} \cdot \ldots \cdot
s_n^{\ell_n}$, and by $p \leadsto^* m$ we mean the existence of a
derivation $p \step{\grammar}^* s_1 \pop \ldots \pop s_n$, using only
rules from $\arules\cup\brules$, such that $m = s_1 \cdot \ldots \cdot
s_n$. The $\leadsto$ and $\leadsto^*$ relations enjoy the following
property:

\begin{lemmaE}\label{lemma:leadsto}
  For each nonterminal $p \in \mathcal{P}$ and monomial
  $m\in\termuniv{p}$, we have $p \leadsto^* m$ if and only if $p
  \leadsto \norm_p(m)$.
\end{lemmaE}
\begin{proofE}
  ``$\Rightarrow$'' Let $\svar \in \varsof{m}$ be a nonterminal. We
  consider two cases. First, assume that $\svar\in\nopump{p}$. If
  $\deg(m,\svar)\geq\base{\svar}{p}$ then $\norm_p(m)=0$, because
  $\norm_p(m)$ is obtained from $m$ by applying the \boundAx{}
  axiom. Then, we have $p \not\leadsto^* m$ and $p \not\leadsto
  \norm_p(m)$. Else, $\deg(m,\svar) = \deg(\norm_p(m),\svar) <
  \base{\svar}{p}$ and there is nothing to prove. Second, assume that
  $\svar\in\pump{p}$. If $\deg(m,\svar)=\deg(\norm_p(m),\svar)$ there
  is nothing to prove, so we assume that $\deg(m,\svar) >
  \deg(\norm_p(m),\svar)$. Because $\norm_p(m)$ is obtained from $m$
  by applying the \periodAx{} axiom, it must be the case that
  $\deg(m,\svar) = \deg(\norm_p(m),\svar) + k \cdot
  \period{\svar}{p}$, for some $k \in \nat^{\geq 1}$. Assume
  w.l.o.g. that $m = \svar^{\deg(m,\svar)} \cdot m'$. Since $p
  \leadsto m$, there exists a derivation:
  \begin{align*}
    p \step{\grammar}^* \pexp{\svar}{\deg(m,\svar)} \pop p
  \end{align*}
  such that, moreover, $p \leadsto^* m'$. This derivation must apply $k$
  times the rule $p \rightarrow p \pop
  \pexp{\svar}{\period{\svar}{p}}$ (note that this rule is unique,
  because $\grammar$ is assumed to be in normal form). By
  removing the $k$ steps corresponding to the application of this
  rule, we obtain a derivation:
  \begin{align*}
    p \step{\grammar}^* \pexp{\svar}{\deg(\norm_p(m),\svar)} \pop p
  \end{align*}
  Since $p \leadsto^* m'$, we can apply inductively this
  transformation to all variables from $\pump{p}\setminus\set{\svar}$,
  in order to obtain a single-step derivation (i.e., a rule from
  $\brules$) that witnesses $p \leadsto \norm_p(m)$. ``$\Leftarrow$''
  This direction uses a symmetric argument.
\end{proofE}

The interpretation of the signature $\signature_\algof{SP}$ in the
algebra $\algof{A}$ is defined below, for all $a \in \alphabet$ and
$\elem_1,\elem_2 \in \universeOf{A}$, where $y_1;y_2$ denotes the relational
composition, i.e., $y_1;y_2 = \{(a,c) \mid (a,b) \in y_1, (b,c) \in y_2\}$, and $\tuple{t_p}_{p \in \mathcal{P}} \cdot \tuple{u_p}_{p \in \mathcal{P}} \isdef \tuple{t_p \cdot u_p}_{p \in \mathcal{P}}$ denotes component-wise multiplication:
\begin{align*}
  a ^\algof{A} \isdef & ~\set{(s,\bot) \mid s \rightarrow a \in \frules} \cup \set{(s,q) \mid s \rightarrow p \sop q \in \crules,~ p \rightarrow s_a \in \grules} \\
  \elem_1 \pop^\algof{A} \elem_2 \isdef & ~\norm(\parmap(\elem_1) \cdot \parmap(\elem_2)) \quad \quad \quad
  \elem_1 \sop^\algof{A} \elem_2 \isdef ~\seqmap(\elem_1);\seqmap(\elem_2) 
\end{align*}
Note that, since bridges are considered to be S-graphs, $a^\algof{A}$
is an S-profile. Since the definition of $a^\algof{A}$ assumes the
grammar to be alternative, hence there two ways in which a bridge can
be parsed, either $s \step{\grammar} a$, thus yielding the view
$(s,\bot)$, or $s \step{\grammar} p \circ q \step{\grammar} s_a \circ
q \step{\grammar} a \circ q$, thus yielding the view $(s,q)$.

The parallel composition is interpreted in the recognizer algebra as
the multiplication in the term algebra, and the serial composition is
interpreted as relational composition. Intuitively, the interpretation
of $\signature_\algof{SP}$ in the algebra $\algof{A}$ mimicks the
parallel and serial composition operations at the level of profiles,
as shown by the following example:

\begin{example}
  Consider the regular grammar $\grammar_{\mathit{univ}}$ from Example
  \ref{ex:universal-sp-grammar}.  The corresponding alternative
  grammar is obtained by first normalizing the grammar
  $\grammar_{\mathit{univ}}$ (see Example \ref{ex:normalized-grammar})
  and then replacing the rule $p \rightarrow a$ with the rules $p
  \rightarrow s_a$ and $s_a \rightarrow a$, where $s_a$ is a fresh
  nonterminal, for each $a \in \alphabet$. The resulting
  grammar:
  \[\grammar_{\mathit{univ}} \isdef
  (\underbrace{\set{p}}_{\mathcal{P}}\uplus\underbrace{\set{s,s'} \cup \set{s_a \mid a \in \alphabet}}_{\mathcal{S}},\rules,\set{p,s})\]
  has the following rules:
  \[\begin{array}{rlcrlcrlcrl}
  p \rightarrow & p \pop s & \hspace*{5mm} & s \rightarrow & p \sop s & \hspace*{5mm} & p \rightarrow & s_a & \hspace*{5mm} & s_a \rightarrow & a  \\
  p \rightarrow & s' \pop s' && s\rightarrow & p \sop p & \hspace*{5mm} & s \rightarrow & a & \hspace*{5mm} & s' \rightarrow & a \text{, for each } a \in \alphabet \\
  s' \rightarrow & p \sop s && s' \rightarrow & p \sop p 
  \end{array}\]
  For some $a \in \alphabet$, we compute, successively:
  \begin{align*}
    a^\algof{A} = & ~\set{(s_a,\bot),(s,\bot),(s,p),(s,s),(s',\bot),(s',p),(s',s)} \hspace*{5mm}\\
    \parmap(a^\algof{A}) = & ~\tuple{s + s' + s_a} \hspace*{5mm}\\
    a^\algof{A} \pop^\algof{A} a^\algof{A} = & ~s + s'^2 + s \cdot s' + s \cdot s_a + s_a \cdot s' \hspace*{5mm}\\
    \seqmap(a^\algof{A} \pop^\algof{A} a^\algof{A}) = & ~\set{(s,p),(s,s),(s',p),(s',s),(p,\bot)} \hspace*{5mm}\\
    (a^\algof{A} \pop^\algof{A} a^\algof{A}) \sop^\algof{A} a^\algof{A} = & ~\set{(s,\bot),(s,p),(s,s),(s',\bot),(s',p),(s',s)}
  \end{align*}
  It is easy to check that $h((a \pop a) \sop a) =
  \set{(s,\bot),(s,p),(s,s),(s',\bot),(s',p),(s',s)}$, because the right-hand side contains
  all the views of the S-graph $(a \pop a) \sop a$ for
  $\grammar_{\mathit{univ}}$. Consequently, we have $h((a \pop a) \sop
  a) = (a^\algof{A} \pop^\algof{A} a^\algof{A}) \sop^\algof{A}
  a^\algof{A}$.
\end{example}


The following lemma proves the correctness of the above definition:

\begin{lemmaE}\label{lemma:homomorphism}
  $h$ is a homomorphism between $\algof{SP}$ and $\algof{A}$.
\end{lemmaE}
\begin{proofE}
  We prove several points, for all $a \in \alphabet$ and
  $\graph_1, \graph_2 \in \universeOf{SP}$:

  \vspace*{.5\baselineskip}
  \noindent\underline{$h(a^\algof{SP}) = a^{\algof{A}}$}:
  ``$\subseteq$''
  Let $(s,q) \in h(a^\algof{SP})$ be a pair.
  since $a^\algof{SP}$ is an S-graph, by Definition \ref{def:s-view}, we
  have $(s,q) \in h(a^\algof{SP})$ for the derivation $s \astep{\grammar} a$ (in case of $q = \bot$) or the derivation with steps $s \astep{\grammar} p \sop q$, $p \astep{\grammar} s_a$ and $s \astep{\grammar} a$ (in case of $q \neq \bot$).
  In both cases, we have $(s,q)\in a^\algof{A}$, by the definition of $a^\algof{A}$.
  ``$\supseteq$'' Let $(s,q) \in a^\algof{SP}$ be a
  pair.
  Then, there is either a rule $s \rightarrow a \in \frules$ (in case of $q=\bot$), or there are rules  $s \rightarrow p \sop q, p \rightarrow s_a, s \rightarrow a \in \frules$ (in case of $q \neq \bot$).
  In the first case we obtain a derivation $s \astep{\grammar} a$ and hence $(s,\bot) \in
  h(a^\algof{SP})$; in the second case we obtain a derivation $s \astep{\grammar} a$ with steps $s \astep{\grammar} p \sop q$, $p \astep{\grammar} s_a$ and $s \astep{\grammar} a$, and hence $(s,q) \in h(a^\algof{SP})$, by Definition \ref{def:s-view}.

  \vspace*{.5\baselineskip}
  \noindent\underline{$h(\graph_1
    \pop^{\algof{SP}} \graph_2) = h(\graph_1)
    \pop^{\algof{A}} h(\graph_2)$}:
  We distinguish the following cases, according to the types of $G_1$ and $G_2$:
  \begin{itemize}
  \item $\graph_1$ and $\graph_2$ are P-graphs: ``$\subseteq$'' Let $m
    \in h_p(\graph_1 \pop^{\algof{SP}} \graph_2)$, for some $p \in
    \mathcal{P}$, and $m' = \prod_{s \in \mathcal{S}} s^{k_s}$, for
    some $k_s \ge 0$, be the view of $\graph_1 \pop^{\algof{SP}}
    \graph_2$ such that $m = \norm_p(m')$.  Then, for each $s \in
    \mathcal{S}$ and $i \in \interv{1}{k_s}$, there exists a complete
    derivation $s \step{{\grammar}}^* v(s,i)$, where $v(s,i)$ is a
    ground term such that $\graph_1 \pop^{\algof{SP}} \graph_2 =
    (\pop_{s \in \mathcal{S}, i \in \interv{1}{k_s}}
    v(s,i))^\algof{SP}$.  We split the set of ground terms $v(s,i)$
    according to whether $v(s,i)^\algof{SP}$ is a subgraph of
    $\graph_1$ or of $\graph_2$, thus obtaining views $m'_i$ of
    $\graph_i$ such that $m' = m'_1 \cdot m'_2$. Note that, since
    $v(s,i)^\algof{SP}$ is a S-graph, it must be either a subgraph of
    $\graph_1$ or of $\graph_2$. Since $\norm_p(m'_i) \in
    h_p(\graph_i)$ and $m = \norm_p(m'_1 \cdot m'_2) =
    \norm_p(\norm_p(m'_1) \cdot \norm_p(m'_2))$, by
    Lemma~\ref{lemma:nf-nesting}, we obtain $m \in (h(\graph_1)
    \pop^{\algof{A}} h(\graph_2))_p$.

    \noindent''$\supseteq$'' Let $m \in(h(\graph_1)
    \pop^{\algof{A}}h(\graph_2))_p$ for some $p \in \mathcal{P}$.
    Then, there exist $m_i \in h_p(\graph_i)$, for $i=1,2$, such that
    $m = \norm_p(m_1 \cdot m_2)$.  Then, there exist views $m'_i$ of
    $\graph_i$ such that $m_i = \norm_p(m'_i)$.  By Lemma
    \ref{lemma:nf-nesting}, we obtain $m = \norm_p(m'_1\cdot m'_2)$.
    Since $m'_1 \cdot m'_2$ is a view of $\graph_1 \pop^\algof{SP}
    \graph_2$, we obtain $m \in h_p(\graph_1 \pop^\algof{SP}
    \graph_2)$.
  \item $\graph_i$ is a P-graph and $\graph_{3-i}$ is an S-graph, for
    $i = 1,2$: We assume that $i=1$, the proof being the same for the
    case $i=2$, because of the commutativity of $\pop^\algof{SP}$ and
    $\pop^{\algof{A}}$. ``$\subseteq$'' Let $m \in
    h_p(\graph_1 \pop^{\algof{SP}} \graph_2)$, for some $p \in \mathcal{P}$, and
    $m'  = \prod_{s \in \mathcal{S}} s^{k_s}$, for some $k_s \ge 0$, be the view of $\graph_1 \pop^{\algof{SP}} \graph_2$ such
    that $m = \norm_p(m')$. Then, for each $s \in
    \mathcal{S}$ and $i \in \interv{1}{k_s}$, there exists
    a complete derivation $s \step{{\grammar}}^* v(s,i)$,
    where $v(s,i)$ is a ground term such that $\graph_1
    \pop^{\algof{SP}} \graph_2 = (\pop_{s \in \mathcal{S}, i
      \in \interv{1}{k_s}} v(s,i))^\algof{SP}$.
    Since $\graph_2$ is
    an S-graph, the only possibility is that $\graph_2 =
    v(s,i)^\algof{SP}$, for a single $s \in \mathcal{S}$ and $i \in
    \interv{1}{k_s}$, such that $s \step{{\grammar}}^*
    v(s,i)$ is a complete derivation.  Then, $(s,\bot) \in
    h(\graph_2)$.
    Moreover, $m'_1 \isdef m' / s$ (where $/$ here denotes the polynomial division) is a view of $\graph_1$ for ${\grammar}$ and $m_1 \isdef
    \norm_p(m'_1) \in
    h_p(\graph_1)$.
    By~Lemma \ref{lemma:nf-nesting}, we obtain $m = \norm_p(m_1 \cdot s)$, thus
    $m \in (h(\graph_1) \pop^{\algof{A}}
    h(\graph_2))_p$.

    \noindent``$\supseteq$''
    Let $m \in (h(\graph_1) \pop^{\algof{A}} h(\graph_2))_p$ for some $p \in \mathcal{P}$.
    Then , there exists $m_1 \in h_p(\graph_1)$ and $(s,\bot)
    \in h(\graph_2)$ such that $m = \norm_p(m_1 \cdot s)$.
    Then, there exists a view $m'_1$ of $\graph_1$ such that $m_1 =     \norm_p(m'_1)$.
    Moreover, there exists a complete derivation $s \step{{\grammar}}^* v$ such that $v^\algof{SP} = \graph_2$.
    Hence, $m'_1 \cdot s$ is a view of
    $\graph_1 \pop^\algof{SP} \graph_2$ and
    $m = \norm_p(m_1 \cdot s) =
    \norm_p(\norm_p(m'_1) \cdot s)$, by Lemma \ref{lemma:nf-nesting}.
    Thus, $m \in h_p(\graph_1 \pop^\algof{SP} \graph_2)$.
  \item $\graph_1$ and $\graph_2$ are S-graphs:
  ``$\subseteq$'' Let $m \in h_p(\graph_1 \pop^{\algof{SP}} \graph_2)$, for some $p \in \mathcal{P}$, and $m'$ be the view of $\graph_1 \pop^{\algof{SP}} \graph_2$ such that $m = \norm_p(m')$. Then, $m' =
    s_1 \cdot s_2$ and there are complete derivations $s_i
    \step{{\grammar}}^* v_i$, for $i=1,2$, where $v_i$ is a
    ground term such that $\graph_i = v_i^\algof{SP}$.
    Then, we obtain $(s_i,\bot) \in h(\graph_i)$, for
    $i=1,2$.
    Moreover, $m'=s_1 \cdot s_2$.
    Thus, $m \in h(\graph_1) \pop^{\algof{A}}
    h(\graph_2)$.

    \noindent``$\supseteq$''
    Let $m \in (h(\graph_1) \pop^{\algof{A}}
    h(\graph_2)_p$ for some $p \in \mathcal{P}$.
    Then, there exist views $(s_i,\bot) \in
    h(\graph_i)$ such that $m = \norm_p(s_1 \cdot s_2)$.
    Hence, there are complete derivations $S_i
    \step{{\grammar}} v_i$ such that $\graph_i =
    v_i^\algof{SP}$, for $i=1,2$.
    Then, $s_1 \cdot s_2$ is a view of $\graph_1 \pop^\algof{SP} \graph_2$.
    Thus, $m \in h_p(\graph_1 \pop^\algof{SP}
    \graph_2)$.
  \end{itemize}

  \vspace*{.5\baselineskip}
  \noindent\underline{$h(\graph_1
    \sop^{\algof{SP}} \graph_2) = h(\graph_1)
    \sop^{\algof{A}} h(\graph_2)$}:
  Since $\graph_1 \sop^{\algof{SP}} \graph_2$ is an S-graph, any view of it is a pair $(s,q)$, such that $s \in
  \mathcal{S}$ and $q \in \mathcal{S} \uplus \mathcal{P} \uplus \set{\bot}$.
  We distinguish the following cases, according to
  the types of $\graph_1$ and $\graph_2$: \begin{itemize}
  \item $\graph_1$ and $\graph_2$ are P-graphs: ``$\subseteq$'' Let
    $(s,q) \in h(\graph_1 \sop^{\algof{SP}} \graph_2)$ and assume
    $q\neq\bot$ (the case $q = \bot$ is similar).  Then, there exists
    a derivation $s \astep{{\grammar}} v \sop q$, for a ground term
    $v$ such that $\graph_1 \sop^{\algof{SP}} \graph_2 =
    v^\algof{SP}$.  By Lemma~\ref{lemma:assoc-sp}, there exist
    derivations $s \astep{{\grammar}} v_1 \sop s'$ and $s'
    \astep{{\grammar}} v_2 \sop q$, for some nonterminal $s' \in
    \mathcal{S}$ and ground terms $v_i$ such that $\graph_i =
    v_i^\algof{SP}$, for $i=1,2$.  Then, by
    Lemma~\ref{lemma:assoc-sp-start}, and because $\graph_1$ and
    $\graph_2$ are $\sop$-atomic, there are rules $s \rightarrow p_1
    \sop s'$ and $s' \rightarrow p_2 \sop q$, and complete derivations
    $p_i \step{{\grammar}}^* v_i$ for some non-terminals $p_1,p_2 \in
    \mathcal{P}$.  By a reordering, we can assume w.l.o.g. that these
    complete derivations are of the form:
    \begin{align*}
      p_i \step{{\grammar}}^* \pop_{s \in \mathcal{S}} \pexp{s}{k^i_s}   \step{{\grammar}}^* v_i
    \end{align*}
    Then, $m_i = \prod_{s \in \mathcal{S}} \pexp{s}{k^i_s}$ is a view of $\graph_i$.
    Thus, $p_i \leadsto^* m_i$, and we
    obtain $p_i \leadsto \norm_{p_i}(m_i)$, by Lemma \ref{lemma:leadsto}, for
    $i=1,2$.
    Moreover, $\norm_{p_i}(m_i) \in
    h_{p_i}(\graph_i)$, for $i=1,2$.
    Thus,$(s,q) \in h(\graph_1) \sop^{\algof{A}} h(\graph_2)$.

    \noindent ``$\supseteq$'' Let $(s,q) \in h(\graph_1)
    \sop^{\algof{A}} h(\graph_2)$ and assume $q \neq \bot$ (the case
    $q = \bot$ is similar).  Then, there exists rules $s \rightarrow
    p_1 \sop s_1$ and $s_1 \rightarrow p_2 \sop q$ in ${\grammar}$,
    with $p_i \leadsto m_i \in h_{p_i}(\graph_i)$, for $i=1,2$.  Then,
    there are views $m_i' = \prod_{s \in \mathcal{S}} \pexp{s}{k^i_s}$
    of $\graph_i$, with $\norm_{p_i}(m_i') = m_i$, for $i=1,2$.  Thus,
    there are derivations:
    \begin{align*}
      p_i \step{{\grammar}}^* \pop_{s \in \mathcal{S}} \pexp{s}{k^i_s} \step{{\grammar}}^* v_i
    \end{align*}
    for some ground terms $v_i$, such that $\graph_i = v_i^\algof{SP}$.
    Thus, we can build a
    derivation:
    \begin{align*}
      s \step{{\grammar}} p_1 \sop s_1 \step{{\grammar}} p_1 \sop p_2
      \sop q \astep{{\grammar}}^* v_1 \sop v_2 \sop q
    \end{align*}
    such that $\graph_1 \sop^\algof{SP} \graph_2 = (v_1 \sop v_2)^\algof{SP}$ and $(s,q) \in h(\graph_1\sop^\algof{SP} \graph_2)$.
  \item $\graph_1$ is a P-graph and $\graph_2$ is an S-graph:
    ``$\subseteq$'' Let $(s,q) \in h(\graph_1 \sop^{\algof{SP}}
    \graph_2)$ be a view of $\graph_1\sop^{\algof{SP}} \graph_2$ and
    assume $q \neq \bot$ (the case $q = \bot$ is similar).  Then,
    there exists a derivation $s \astep{{\grammar}} v \sop q$, for a
    ground term $v$ such that $\graph_1 \sop^{\algof{SP}} \graph_2 =
    v^\algof{SP}$.  By Lemma \ref{lemma:assoc-sp}, there exist
    derivations $S \astep{{\grammar}} v_1 \sop s_1$ and $s_1
    \astep{{\grammar}} v_2 \sop q$, for some $s_1 \in \mathcal{S}$,
    such that $v_i^{\algof{SP}} = \graph_i$, for $i=1,2$.  Then, by
    Lemma~\ref{lemma:assoc-sp-start}, and because $\graph_1$ is
    $\sop$-atomic, there is a rule $s \rightarrow p \sop s_1$ and a
    complete derivation $p \step{{\grammar}}^* v_1$ for some
    non-terminal $p \in \mathcal{P}$.  By a reordering, if necessary,
    we assume the latter derivation to have the form:
    \begin{align*} P
      \step{{\grammar}}^* ~\pop_{s \in
        \mathcal{S}} \pexp{s}{k_s}
      \step{{\grammar}}^* ~\pop_{s \in
        \mathcal{S},i \in \interv{1}{k_s}} v_1(s,i)
      \end{align*}
    for some $k_s \ge 0$, $s \in \mathcal{S}$, where $v_1(s,i)$ are
    ground terms such that $\graph_1 = (\pop_{s \in \mathcal{S},i \in
      \interv{1}{k_s}} v_1(s,i))^\algof{SP}$.  Then, $m = \prod_{s \in
      \mathcal{S}} \pexp{s}{k_s}$ is a view of $\graph_1$ such that $p
    \leadsto^* m$. By Lemma~\ref{lemma:leadsto}, we obtain $p
    \leadsto \norm_p(m) \in h_p(\graph_1)$.  Moreover, since
    $\graph_2$ is an S-graph, we obtain $(s,q) \in h(\graph_2)$. Thus
    we obtain $(s,q) \in h(\graph_1) \sop^{\algof{A}} h(\graph_2)$.

    \noindent ``$\supseteq$'' Let $(s,q) \in h(\graph_1)
    \sop^{\algof{A}} h(\graph_2)$ and assume that $q \neq \bot$ (the
    case $q = \bot$ is similar).  Then, there exists a rule $s
    \rightarrow p\sop s_1$ such that $p \leadsto m$, for some $m \in
    h_p(\graph_1)$, and $(s_1,q) \in h(\graph_2)$.  Then, there exists
    a view $m' = \prod_{s \in \mathcal{S}} \pexp{s}{k_s}$ of
    $\graph_1$, such that $\norm_p(m') = m$.  Hence, there exists a
    derivation:
      \begin{align*}
      p \step{{\grammar}}^* \pop_{s \in
          \mathcal{S}} \pexp{s}{k_s} \step{{\grammar}}^* v_1,
      \end{align*}
      for some ground term $v_1$, such that $\graph_1 =v_1^\algof{SP}$. Moreover, $(s_1,q) \in h(\graph_2)$, hence there exists a derivation $s_1
      \astep{{\grammar}} v_2 \sop w$, where $v_2$ is a ground term such
      that $\graph_2 = v_2^\algof{SP}$.
      We can build a derivation:
      \begin{align*}
        s \step{{\grammar}} p \sop s_1 \step{{\grammar}}^* (\pop_{s \in
          \mathcal{S}} \pexp{s}{k_s}) \sop s_1 \astep{{\grammar}} v_1 \sop v_2 \sop q.
      \end{align*}
      Then, $(v_1 \sop v_2)^\algof{SP} = \graph_1 \sop^\algof{SP}
      \graph_2$ and $(s,q) \in h(\graph_1 \sop^\algof{SP} \graph_2)$.
    \item $\graph_1$ is an S-graph and $\graph_2$ is a P-graph:
      ``$\subseteq$'' Let $(s,q) \in h(\graph_1 \sop^{\algof{SP}}
      \graph_2)$ and assume that $q \neq \bot$ (the case $q = \bot$ is
      similar).  Then, there exists a derivation $s \astep{\grammar}
      v \sop q$, where $v$ is a ground term such that $\graph_1
      \sop^{\algof{SP}} \graph_2 = v^\algof{SP}$. By
      Lemma~\ref{lemma:assoc-sp}, there exist derivations $s
      \astep{{\grammar}}^* v_1 \sop s_1$ and $s_1 \astep{{\grammar}}
      v_2 \sop q$, for some $s_1 \in \mathcal{S}$, such that
      $v_i^{\algof{SP}} = \graph_i$, for $i=1,2$.  Since $\graph_1$ is
      an S-graph, we obtain $(s,s_1) \in h(\graph_1)$.  Then, by
      Lemma~\ref{lemma:assoc-sp-start}, and because $\graph_2$ is
      $\sop$-atomic, there is a rule $s_1 \rightarrow p \sop q$ and a
      complete derivation $p \step{{\grammar}}^* v_2$ for some
      non-terminal $p \in \mathcal{P}$.  By a reordering, if
      necessary, we assume the latter derivation to have the
      form: \begin{align*} P \step{{\grammar}}^* ~\pop_{s \in
          \mathcal{S}} \pexp{s}{k_s} \step{{\grammar}}^* ~\pop_{s \in
          \mathcal{S},i \in \interv{1}{k_s}} v_2(s,i)
      \end{align*}
      for some $k_s \ge 0$, $s \in \mathcal{S}$, where $v_2(s,i)$ are
      ground terms and $\graph_2 = (\pop_{s \in \mathcal{S},i \in
        \interv{1}{k_s}} v_2(s,i))^\algof{SP}$. Then, $m = \prod_{s
        \in \mathcal{S}} \pexp{s}{k_s}$ is a view of $\graph_2$ such
      that $p \leadsto^* m$. By Lemma~\ref{lemma:leadsto}, we obtain
      $p \leadsto \norm_p(m) \in h_p(\graph_2)$. Thus, we obtain
      $(s,q) \in h(\graph_1) \sop^{\algof{A}} h(\graph_2)$.
      \noindent ``$\supseteq$'' Let $(s,q) \in h(\graph_1)
      \sop^{\algof{A}} h(\graph_2)$ and assume that $q \neq \bot$ (the
      case $q = \bot$ is similar).  Then, there exist $(s,s_1)\in
      h(\graph_1)$, a rule $s_1 \rightarrow p \sop q$ in ${\grammar}$
      and some $m \in h_p(\graph_2)$ such that $p \leadsto m$.  Then,
      there exists a view $m' = \prod_{p \in \mathcal{S}}
      \pexp{s}{k_s}$ of $\graph_2$, such that $\norm_p(m') = m$, and a
      derivation
      \begin{align*}
      p \step{{\grammar}}^* \pop_{p \in          \mathcal{S}} \pexp{s}{k_s} \step{{\grammar}}^* v_2,
      \end{align*}
      for some ground term $v_2$, such that $\graph_2 = v_2^\algof{SP}$.
      Moreover, there exists a derivation $s
      \astep{\grammar} v_1 \sop s_1$, where $v_1$ is a ground term such that $\graph_1 = v_1^\algof{SP}$.
      By Lemma
      \ref{lemma:assoc-sp}, there exists a derivation $s \astep{{\grammar}} v \sop q$ such that $v^{\algof{SP}} = \graph_1 \sop^{\algof{SP}} \graph_2$.
      Thus, we obtain $(s,q) \in h(\graph_1 \sop^\algof{SP} \graph_2)$.
    \item $\graph_1$ and $\graph_2$ are S-graphs: ``$\subseteq$'' Let
      $(s,q) \in h(\graph_1 \sop^{\algof{SP}}
      \graph_2)$ and assume that $q \neq \bot$ (the case $q = \bot$ is
      similar).
      Then, there exists a derivation $s \astep{\grammar} v \sop q$, where      $v$ is a ground term such that $\graph_1 \sop^{\algof{SP}} \graph_2 = v^\algof{SP}$.
      By Lemma~\ref{lemma:assoc-sp}, there
      exist derivations $s \astep{{\grammar}} v_1 \sop s_1$
      and $s_1 \astep{{\grammar}} v_2 \sop q$, where $s_1 \in \mathcal{S}$ and $v_i$ are ground terms such that $\graph_i = v_i^{\algof{SP}}$, for $i=1,2$.
      Hence, we obtain $(s,s_1) \in
      h(\graph_1)$ and $(s_1,q) \in
      h(\graph_2)$.
      Thus, $(s,q) \in h(\graph_1) \sop^{\algof{A}}
      h(\graph_2)$.

      \noindent''$\supseteq$''
      Let $(s,q) \in h(\graph_1) \sop^{\algof{A}} h(\graph_2)$ and assume that $q \neq \bot$
      (the case $q = \bot$ is similar).
      Then, there exist $(s,s_1) \in h(\graph_1)$ and $(s_1,q) \in h(\graph_2)$, and there are derivations
      $s \astep{\grammar} v_1 \sop s_1$ and $s_1 \astep{\grammar} v_2 \sop q$, where $v_i$ are ground terms such that $\graph_i = v_i^\algof{SP}$, for $i=1,2$.
      By Lemma~\ref{lemma:assoc-sp}, we  obtain a derivation $s \astep{\grammar} v \sop q$ such that
      $v^\algof{SP} = \graph_1 \sop^\algof{SP} \graph_2$, hence $(s,q)
      \in h(\graph_1 \sop^\algof{SP} \graph_2)$.
  \end{itemize}
\end{proofE}
Moreover, $\langof{}{\grammar}$ does not distinguish graphs with the
same image in $\algof{A}$:

\begin{lemmaE}\label{lemma:membership}
  SP-graphs having the same image via $h$ are indistinguishable by
  $\langof{}{\grammar}$.
\end{lemmaE}
\begin{proofE}
  We assume $h(\graph_1)=h(\graph_2)$ for some $\graph_1, \graph_2 \in
  \universeOf{SP}$ and prove that $\graph_1\in\langof{}{\grammar} \iff
  \graph_2\in\langof{}{\grammar}$. Because $h(\graph_1)=h(\graph_2)$,
  we have that $\graph_1$ and $\graph_2$ are either both P-graphs or
  both S-graphs. Assume that $\graph_1\in\langof{}{\grammar}$ (the
  other direction is symmetric). We prove that
  $\graph_2\in\langof{}{\grammar}$ by considering the following cases:

  \vspace*{.5\baselineskip}\noindent \underline{$\graph_1,\graph_2$ are P-graphs}:
  Since $\graph_1 \in \langof{}{\grammar}$, there exists a complete
  derivation $p \step{{\grammar}}^* v$, for some axiom $p \in \axioms \cap \mathcal{P}$
  such that $v^\algof{SP} = \graph_1$. The derivation can be reorganized w.l.o.g. as:
    \begin{align*}
      p \step{{\grammar}}^* ~\pop_{s \in \mathcal{S}} ~\pexp{s}{k_s}
      \step{{\grammar}}^* ~\pop_{s \in \mathcal{S}, i \in
        \interv{1}{k_s}} v_1(s,i)
    \end{align*}
    for some coefficients $k_s \ge 0$ and ground
    terms $v_1(s,i)$ such that $\graph_1 = (\pop_{s \in
      \mathcal{S},~ i \in \interv{1}{k_s}}
    v_1(s,i))^\algof{SP}$.
    Thus, $m_1 = \prod_{s \in \mathcal{S}} s^{k_s}$ is a view of  $\graph_1$, and $\norm_p(m_1) \in h_p(\graph_1)$.
    Because $h(\graph_1)=h(\graph_2)$, we have, in particular, that $h_p(\graph_1)=h_p(\graph_2)$.
    Thus, there exists a view $m_2 = \prod_{s \in \mathcal{S}} s^{l_s}$ of $\graph_2$ such that $\norm_p(m_1) =
    \norm_p(m_2)$.
    Moreover, $p \leadsto^*
    m_1$, hence $p \leadsto
    \norm_p(m_1) =
    \norm_p(m_2)$ and $p
    \leadsto^* m_2$, by Lemma~\ref{lemma:leadsto}. We obtain a derivation:
    \begin{align*}
      p \step{{\grammar}}^* ~\pop_{s \in \mathcal{S}} ~\pexp{s}{l_s}
      \step{{\grammar}}^* ~\pop_{s \in \mathcal{S}, i \in
      \interv{1}{l_s}} v_2(s,i)
    \end{align*}
    where $v_2(s,i)$ are ground terms, such that $\graph_2 = (\pop_{s \in \mathcal{S},~ i \in \interv{1}{l_s}}
    v_2(s,i))^\algof{SP}$.
    Since $\rightarrow p$ is an axiom of
    $\grammar$, we obtain that $\graph_2 \in
    \langof{}{\grammar}$.

    \vspace*{.5\baselineskip}\noindent \underline{$\graph_1,\graph_2$
      are S-graphs}:
    We distinguish two cases: \begin{itemize}
    \item $\graph_1$ is not atomic. Since $\graph_1
      \in\langof{}{\grammar}$, there exists a complete derivation $s
      \step{{\grammar}}^* v_1$, such that $\graph_1 = v_1^\algof{SP}$,
      for some axiom $s \in \axioms \cap \mathcal{S}$.  Then,
      $(s,\bot) \in h(\graph_1) = h(\graph_2)$, hence there exists a
      complete derivation $s \step{{\grammar}}^* v_2$ such that
      $v_2^\algof{SP}=\graph_2$, leading to $\graph_2 \in
      \langof{}{\grammar}$.
    \item If $\graph_1 = a^\algof{SP}$ is a bridge, for some $a \in
      \alphabet$ then, because $\graph_1 \in\langof{}{\grammar}$ and
      $\grammar$ is an alternative grammar, there exists a derivation
      $s \step{\grammar} a$, for some axiom $s \in \axioms \cap
      \mathcal{S}$. The proof proceeds as in the previous case.
    \end{itemize}
\end{proofE}

The main result of this section is a consequence of Lemmas
\ref{lemma:membership} and \ref{lemma:homomorphism} above:

\begin{theorem}\label{thm:reg-rec}
  The language of each regular grammar $\grammar$ is recognized by the
  finite algebra $\algof{A}$ with accepting set $F \isdef
  h(\langof{}{\grammar})$.
\end{theorem}
\begin{proof}
  Since $h$ is a homomorphism between $\algof{SP}$ and $\algof{A}$, by
  Lemma \ref{lemma:homomorphism}, it remains to prove that
  $\langof{}{\grammar} = h^{-1}(F)$. ``$\subseteq$'' This direction is
  trivial, by the definition of $F$. ``$\supseteq$'' For each SP-graph
  $\graph \in h^{-1}(F)$ there exists $\graph' \in
  \langof{}{\grammar}$ such that $h(\graph) = h(\graph')$. By Lemma
  \ref{lemma:membership}, we obtain $\graph \in \langof{}{\grammar}$.
\end{proof}

We also need the following constructive definition of the accepting
set of the $\algof{A}$ recognizer built in the previous:

\begin{lemmaE}\label{lemma:accepting}
  The accepting set $F$ consists of the elements $a \in
  \universeOf{A}$, for which either one of the following
  holds: \begin{enumerate}
  \item\label{it1:lemma:accepting} if $a =
    \tuple{t_p}_{p\in\mathcal{P}} \in F \cap \universeOf{A}^P$ then
    there exists an axiom $p \in \mathcal{P} \cap \axioms$ and a
    nonempty monomial $m \in t_p$ such that $p \leadsto m$.
  \item\label{it2:lemma:accepting} if $a \in F \cap \universeOf{A}^S$
    then $(s,\bot) \in a$, for an axiom $s \in \mathcal{S} \cap
    \axioms$.
  \end{enumerate}
\end{lemmaE}
\begin{proofE}
  (\ref{it1:lemma:accepting}) Let $a =
  \tuple{t_p}_{p\in\mathcal{P}}\in \universeOf{A}^P \cap F$ and
  $\graph$ be a P-graph such that $h(\graph) = a$. Since $F =
  h(\langof{}{\grammar})$, by definition, we obtain $\graph \in
  \langof{}{\grammar}$.  Since, moreover, $\graph$ is a P-graph, there
  exists an axiom $p \in \mathcal{P} \cap \axioms$ such that $p
  \step{\grammar}^* v$ is a complete derivation and $v^\algof{SP} =
  \graph$. This derivation can be w.l.o.g. considered to be of the
  form:
  \begin{align*}
    p \step{\grammar}^* \Pop_{s \in \mathcal{S}} \pexp{s}{k_s}
    \step{\grammar}^* \Pop_{s \in \mathcal{S}} v(s,i) = t
  \end{align*}
  where each $v(s,i)$ is a ground term, for $s \in \mathcal{S}$ and $1
  \leq i \leq k_s$. Then, the nonempty monomial $m \isdef
  \prod_{s\in\mathcal{S}} s^{k_s}$ is view of $\graph$ and $\norm_p(m)
  \in t_p$ (Definition \ref{def:p-view}). Moreover, $p \leadsto^* m$
  hence, by Lemma \ref{lemma:leadsto}, we obtain $p \leadsto
  \norm_p(m)$.

  \vspace*{.5\baselineskip}
  \noindent(\ref{it2:lemma:accepting}) Let $a \in \universeOf{A}^S
  \cap F$ and $\graph$ be a S-graph such that $h(\graph) = a$. Since
  $F = h(\langof{}{\grammar})$, we obtain $\graph \in
  \langof{}{\grammar}$. Since, moreover, $S$ is an S-graph, there
  exists an axiom $s \in \mathcal{S} \cap \axioms$ such that $s
  \step{\grammar}^* v$ is a complete derivation and $v^\algof{SP} =
  \graph$. Then $(s,\bot) \in a$ is a view of $\graph$ (Definition
  \ref{def:s-view}).
\end{proofE}

The next proposition gives a lower bound on the cardinality of the
minimal recognizer for a regular grammar, which improves on the
standard $2^n$ lower bound of the recognizers for the language of a
nondeterministic word or tree automaton~\cite[Theorem
  1.3.1]{comon:hal-03367725}. The worst-case example has been adapted
to our setting from~\cite[Theorem 3.4]{10.1145/1516512.1516518}, see
also~\cite[Theorem 5.2]{10.1145/3571230}:

\begin{propositionE}\label{prop:blowup}
  There exists a regular grammar $\grammar$ having $\bigO(n)$
  nonterminals and $\bigO(n \log_2 n)$ rules, such that each
  recognizer for $\langof{}{\grammar}$ requires at least $2^{n^2}$
  elements.
\end{propositionE}
\begin{proofE}
  Consider the alphabet $\alphabet=\set{a,b,c,\$,\#}$. Each word $w
  \in \alphabet^*$ is encoded by the edge labels of a path
  $\graph_w$. Given an integer $k\geq2$ and two words
  $u,v\in\set{a,b}^k$, we consider the set $\langu^k_{u,v}$ of words
  $w = w_1 \ldots w_n$ such that $w_i \in \set{a,b}^* \$ \set{a,b}^*
  \#$, for each $i \in \interv{1}{n}$, and $u \$ v \# = w_j$, for some
  $j \in \interv{1}{n}$. Let $\graphs^k_{u,v} \isdef \set{\graph_w
    \mid w \in \langu^k_{u,v}}$ be the set of corresponding line
  graphs and define the set of SP-graphs $\graphs^k \isdef
  \bigcup_{u,v\in\set{a,b}^k} \left(c \pop (\graphs^k_{u,v} \sop
  v)\right) \sop u$. The structure of the graphs from $\graphs^k$ is
  illustrated in Figure \ref{fig:cex} (a).

  First, one checks that $\graphs^k$ is the language of the regular
  grammar $\grammar^k = (\mathcal{P} \uplus \mathcal{S}, \rules,
  \axioms)$, where:
  \begin{align*}
    \mathcal{P} \isdef & ~\set{p_w \mid w \in \set{a,b}^k} \cup \set{\overline{p}_\alpha \mid \alpha \in \alphabet} \\
    \mathcal{S} \isdef & ~\set{s^i_w \mid w \in \set{a,b}^\ell,~ 0 \leq \ell \leq k,~ 0 \leq i \leq 5} \\
    \axioms \isdef & ~\set{s_0}
  \end{align*}
  and $\rules$ consists of the rules described in Figure \ref{fig:cex}
  (b) and (c). Assuming that $w \isdef w_1 \ldots w_m \in \set{a,b,\$,\#}^m$, for some $m\geq1$, the notation $s^0_{w} \arrow{w}{}^*
  s^2_\epsilon$ stands for the rules: \begin{align*}
    s^0_{w_\ell \ldots w_m} \rightarrow & ~\overline{p}_{w_\ell} \sop s^0_{w_{\ell+1} \ldots w_m} \text{, for all } 1 \leq \ell < m \\
    s^0_{w_m} \rightarrow & ~\overline{p}_{w_m} \sop s^2_\epsilon \\
    \overline{p}_\alpha \rightarrow & ~\alpha \text{, for all } \alpha \in \set{a,b,\$,\#}
  \end{align*}
  Similarly, $s^2_\epsilon \arrow{w}{}^* s^3_{w}$ stands for the rules: \begin{align*}
    s^2_\epsilon \rightarrow & ~\overline{p}_{w_1} \sop s^3_{w_1} \\
    s^3_{w_1 \ldots w_\ell} \rightarrow & ~\overline{p}_{w_{\ell+1}} \sop s^3_{w_1 \ldots w_{\ell+1}} \text{, for all } 1 \leq \ell < m
  \end{align*}
  Finally, $s^3_{w} \arrow{w}{}^*$ means: \begin{align*}
    s^3_{w_1 \ldots w_m} \rightarrow & ~\overline{p}_{w_1} \sop s^5_{w_2 \ldots w_m} \\
    s^5_{w_\ell \ldots w_m} \rightarrow & ~\overline{p}_{w_\ell} \sop s^5_{w_{\ell+1} \ldots w_m} \text{, for all } 2 \leq \ell < m-1 \\
    s^5_{w_{m-1}w_m} \rightarrow & ~\overline{p}_{w_{m-1}} \sop \overline{p}_{w_m}
  \end{align*}
  The edges $s \arrow{\alpha}{} s'$ stand for the rules $s \rightarrow
  \overline{p}_\alpha \sop s'$, respectively. We note that
  $\grammar^k$ has $\bigO(2^k)$ nonterminals, because there are $n
  \isdef 2^{k+1}-1$ words $u,v \in \bigcup_{0 \leq \ell \leq k}
  \set{a,b}^\ell$. Since there are $\bigO(n)$ meta-rules of the form
  $s \rightarrow^* u$ or $s \rightarrow^* u \sop s'$, that generate at
  most $k$ rules each, the grammar $\grammar^k$ has $\bigO(k \cdot
  2^k) = \bigO(n \cdot \log_2 n)$ rules.

\begin{figure}[t!]
  \centerline{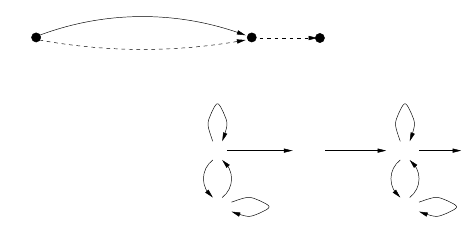}
  \caption{The structure of rules from $\grammar^k$}
  \label{fig:cex}
\end{figure}

  Second, suppose, for a contradiction, that the minimal recognizer
  for $\langof{}{\grammar^k}$ has less than $2^{{(2^k)}^2}$ elements. By
  the pigeonhole principle, there exist two relations $R_i \isdef
  \set{(u^i_j,v^i_j)}_{1 \leq j \leq \ell_i} \subseteq \set{a,b}^k
  \times \set{a,b}^k$, for $i=1,2$, such that the line graphs $\graph_{w_1}$ and
  $\graph_{w_2}$ are mapped, by the homomorphism from Definition
  \ref{def:rec}, into the same element of the recognizer, where $w_i
  \isdef u^i_1 \$ v^i_1 \# \ldots u^i_{k_i} \$ v^i_{\ell_i} \#$, for
  $i=1,2$. This is because there are $2^{{(2^k)}^2} = 2^{n^2}$
  relations on words from $\set{a,b}^k$. Let $1
  \leq i \leq \min(\ell_1,\ell_2)$ be the first index such that
  $(u^1_i,v^1_i) \neq (u^2_i,v^2_i)$. But then, the graphs $\left[c
    \pop (\graph_{w_1} \sop \graph_{v^1_{i}})\right] \sop \graph_{u^1_{i}}
  \in \langof{}{\grammar^k}$ and $\left[c \pop (\graph_{w_2} \sop
    \graph_{v^1_i})\right] \sop \graph_{u^1_i} \not\in
  \langof{}{\grammar^k}$ are mapped into the same element of the
  recognizer, contradiction.
\end{proofE}

\section{An Exponential Bound on the Cardinality of the Recognizer}
\label{sec:complexity}

We establish a $2^{\bigO(n^9)}$ upper bound on the cardinality of the
recognizer $\algof{A}$ built from the regular grammar $\grammar$,
where $n$ is the size of $\grammar$ (Corollary \ref{cor:main}). This
upper bound is (almost) optimal, since a $2^{\Omega(n^2)}$ blowup is
unavoidable (Proposition \ref{prop:blowup}).

Because the domain of $\algof{A}$ is partitioned into S-profiles
($\universeOf{A}^S$) and P-profiles ($\universeOf{A}^P$), we analyse
the two cases separately. Note first that there are at most
$\cardof{\mathcal{S}} \cdot
(\cardof{\mathcal{P}}+\cardof{\mathcal{S}}+1)$ pairs of the form
$(s,q)\in\mathcal{S}\times(\mathcal{P}\uplus\mathcal{S}\uplus\set{\bot})$
in $\universeOf{A}^S$, hence $\cardof{\universeOf{A}^S}$ is
exponential in the size of $\grammar$ (Definition
\ref{def:s-view}). In the rest of this section, we prove that
$\cardof{\universeOf{A}^P}$ is also exponential in the size of
$\grammar$, which provides an exponential bound for
$\cardof{\universeOf{A}} = \cardof{\universeOf{A}^S} +
\cardof{\universeOf{A}^P}$.

For the purpose of proving that $\cardof{\universeOf{A}^P}$ is
exponentially bounded, a key observation is that $\universeOf{A}^P$
consists of tuples of normal forms of linear products\footnote{A
  linear product is a term of the form $(s_{11} + \ldots + s_{1n_1})
  \cdot \ldots \cdot (s_{k1} + \ldots + s_{kn_k})$, where $s_{ij} \in
  \mathcal{S}$.} over variables from $\mathcal{S}$. Indeed, this is
the case because each $\tuple{t_p}_{p \in \mathcal{P}} \in
\universeOf{A}^P$ is the profile $h(\graph) =
\tuple{h_p(\graph)}_{p\in\mathcal{P}}$ of a P-graph $\graph$, which is
the parallel composition of at least two S-graphs, i.e., $\graph =
\graph_1 \pop \ldots \pop \graph_n$, $n\geq2$. This means $h_p(\graph)
= \norm_p(\parmap(h(\graph_1))_p \cdot \ldots \cdot
\parmap(h(\graph_n))_p)$, for each $p \in \mathcal{P}$. Because each
$\graph_i$ is an S-graph, we have $h(\graph_i) \in \universeOf{A}^S$
hence each term $\parmap(h(\graph_i))_p$ is a sum of variables from
$\mathcal{S}$, by the definition of $\parmap$. Then, $\parmap(h(\graph_1))_p \cdot \ldots \cdot
\parmap(h(\graph_n))_p$ is a linear product and $t_p=h_p(\graph)$ is its
normal form, for each $p\in\mathcal{P}$. Consequently,
$\cardof{\universeOf{A}^P}$ is at most the number of linear products
over $\mathcal{S}$ that have distinct normal forms, at power $\cardof{P}$.


\subsection{Counting Products of Linear Terms}

Let us fix a non-terminal $p\in\mathcal{P}$ and let $\nopump{}$ and
$\pump{}$ denote the sets of bounded and periodic variables associated
to $p$, respectively. Since $p$ is fixed, we omit mentioning it in the
following. We also consider a set $\ThresholdVars$ of threshold
variables. For simplicity, in this context, we assume that
$\mathcal{P}=\set{p}$ and $\SVars = \nopump{} \uplus \pump{} \uplus
\ThresholdVars$.

The term algebra corresponding to $p$ is denoted $\termalg{} =
(\termuniv{},+,\cdot,0,1)$. We recall that the domain $\termuniv{}
\isdef \termsof{\signature_\algof{T}}{\SVars}$ of the term algebra is
the set of terms built from $\SVars$ using the operations $+$ and
$\cdot$ (\autoref{sec:termalg}). We assume w.l.o.g. that each variable
$\svar \in \pump{}$ has period $\period{s}{} \ge 2$: if
$\period{\svar}{} = 1$ for some variable $\svar$, then $\norm(\svar) =
1$, by the \text{\periodAx}~ axiom, hence $\svar$ either does not
occur in the normal form of a linear term, or yields the empty
monomial $1$.

The goal of this subsection is to provide an exponential bound on the
number of terms in normal form $\norm(t)$\footnote{We also write
  $\norm(t)$ instead of $\norm_p(t)$ because the nonterminal $p \in
  \mathcal{P}$ is fixed.}, where $t$ is a linear product of bounded
and periodic variables. This is achieved by providing a polynomial
cut-off bound on the length (i.e., number of factors) of the linear
products having the same normal form. The existence of a cut-off $c$
means that each product longer than $c$ has the same normal form as a
product shorter than $c$. The cut-off allows to bound the number of
distinct linear product terms by an exponential, since the number of
linear products of length at most $c$ is exponential in $c$.

The polynomial cut-off is obtained by considering progressively the
cases of linear factors consisting of (i) bounded variables only
(Lemma \ref{lemma:complexity:bounded-only}), (ii) threshold variables
only (Lemmas \ref{lemma:complexity:2-threshold-only} and
\ref{lemma:complexity:threshold-only}), (iii) bounded and periodic
variables containing the $1$ monomial (Lemma
\ref{lemma:complexity:1-bounded-periodic}), (iv) bounded and periodic
variables containing a distinguished periodic variable (Lemmas
\ref{lemma:complexity:s0-bounded-periodic-basic} and
\ref{lemma:complexity:s0-bounded-periodic}) and (v) the general case
of arbitrary bounded and periodic variables (Lemma
\ref{lemma:complexity:summary}). We recall that the threshold
variables do not occur in the input grammar $\grammar$, being
introduced to support the intermediate steps of the proof.

%
%

For each set $\svarset \subseteq \SVars$ of nonterminals, we define
the following measure: \begin{align}\label{eq:measure-one}
  \weightedcardof{\svarset} \isdef \hspace*{-2mm} \sum_{\svar \in
    \svarset\cap\nopump{}} (\base{\svar}{}-1) + \hspace*{-2mm}
  \sum_{\svar\in\svarset \cap \pump{}} \hspace*{-2mm}
  (\period{\svar}{}-1) + \hspace*{-2mm} \sum_{\svar\in \svarset \cap
    \ThresholdVars{}} \hspace*{-2mm} (\threshold{\svar}-1)
\end{align}
We consider the following sets of linear terms and products, for $V
\subseteq \SVars$ and $U \subseteq \termuniv{}$:
\begin{align*}
\hspace*{-5mm}\Sums{V} \isdef & ~\Set{\sum\nolimits_{\svar \in W} \svar ~\mid~
  \emptyset \not= W \subseteq V} \hspace*{5mm}
\Prods{U} \isdef \set{u_1 \ldots u_k ~\mid~ u_i \in U,~ i \in
  \interv{1}{k},~k\geq0}
\end{align*}
The length of a product $t=u_1 \ldots u_n$ of linear terms $u_1,
\ldots, u_n$ is denoted $\lengthof{t} \isdef n$. The below lemma shows
that linear products over bounded variables only, that are longer than
$\weightedcardof{\nopump{}}$ have normal form $0$, which provides a
trivial cut-off on their length:

\begin{lemmaE}\label{lemma:complexity:bounded-only}
  For each linear product $t \in \Prods{\Sums{\nopump{}}}$ of length
  $\lengthof{t} > \weightedcardof{\nopump{}}$, we have $\norm(t)=0$.
\end{lemmaE}
\begin{proofE}
  Let $m$ be some unreduced monomial occurring in the product $t$.
  Observe that $m$ contains as many variables as $\lengthof{t}$.  If
  $\lengthof{t} > \weightedcardof{\nopump{}}$ then, by using the
  pigeonhole principle, at least one of the variables $\svar$ occurs
  at least as many times as his bound $\base{\svar}{}$.  Hence,
  $\norm(m) = 0$ and as $m$ was chosen arbitrarily, we conclude
  $\norm(t) = 0$.
\end{proofE}

%
%

Before considering the other cases, we need further notations. Given
reduced monomials $m_1$, $m_2$ we define $m_1 \lemon m_2$ iff
$\varsof{m_1} \subseteq \varsof{m_2}$ and $\deg(m_1,\svar) \le
deg(m_2,\svar)$ for all $\svar \in \varsof{m_1}$.  We also write $m_1
\ltmon m_2$ iff $m_1 \lemon m_2$ and $m_1 \not= m_2$.  For a term
$\norm(t)$ in normal form we denote by $\sup( \norm(t))$ the subterm
of $\norm(t)$ containing only the maximal monomials with respect to
$\lemon$. We also define $\maxdeg{\norm(t}) \isdef \max_{m \in
  \norm(t)} \deg(m)$, i.e., the maximal degree of monomials $m$
occurring in $\norm(t)$.

The following two lemmas give the cut-off on the length of linear
products over threshold variables
only. Lemma~\ref{lemma:complexity:2-threshold-only} gives an auxiliary
result concerning variables of threshold $2$, called $2$-threshold in
the following, namely that the maximal (i.e., w.r.t. $\lemon$)
monomials of a linear product using only $2$-threshold variables are
exactly its monomials of maximal degree. Most of its proof shows that
a linear product $t$ of $2$-threshold variables enjoys a similar
exchange property\footnote{For each $m,n \in \norm(t)$ such that
  $\deg(m) < \deg(n)$ there exists $z \in \varsof{n} \setminus
  \varsof{m}$ such that $m \cdot z \in \norm(t)$.} as matroids,
guaranteeing that maximal bases have equal
cardinalities~\cite{matroids}.

\begin{lemmaE}\label{lemma:complexity:2-threshold-only}
  Assume that $\threshold{\svar}=2$ for each $\svar\in
  \ThresholdVars$. Then, for each linear product
  $t\in\Prods{\Sums{\ThresholdVars}}$, we have \(\sup(\norm(t)) = \sum
  \set{ m \in \norm(t) ~\mid~ \deg(m) = \maxdeg(\norm(t))}\).
\end{lemmaE}
\begin{proofE}
  As all threshold variables have threshold 2, note that every
  variable occur at most once (that is, at power 1) in a reduced
  monomial in $\norm(t)$.  Moreover, we prove that the monomials in
  $\norm(t)$ satisfy the following \emph{exchange property}:

  \begin{quote}
    For any two monomials $m,n \in \norm(t)$ such that $\deg(m) <
    \deg(n)$ there exists a variable $z \in \varsof{n} \setminus
    \varsof{m}$ such that $m \cdot z \in \norm(t)$.
  \end{quote}
  Note that this is sufficient to establish the conclusion of the lemma, as
  follows. ``$\subseteq$'' Let $m \in \sup(\norm(t))$ be a monomial
  and suppose, for a contradiction, that $\deg(m) <
  \maxdeg{\norm(t)}$. Then, there exists a monomial $n \in \norm(t)$
  such that $\deg(m) < \deg(n)$. By the exchange property above,
  there exists a variable $z \in \varsof{n} \setminus \varsof{m}$ such
  that $m \cdot z \in \norm(t)$. We obtain $m \ltmon m\cdot z$, which contradicts
  the choice of $m \in \sup(\norm(t))$. ``$\supseteq$'' Let $m \in
  \norm(t)$ such that $\deg(m) = \maxdeg(\norm(t))$ and suppose,
  for a contradiction, that there exists $n \in \norm(t)$ such that $m
  \ltmon n$. Then, we obtain $\deg(m) < \deg(n)$, which contradicts
  $\deg(m) = \maxdeg(\norm(t))$.

  To prove the above exchange property, we consider arbitrary
  monomials $m,n \in \norm(t)$, where $\deg(m) < \deg(n)$, together
  with their respective injective witness functions $f_m : \varsof{m}
  \rightarrow \interv{1}{k}$ and $f_n : \varsof{n} \rightarrow
  \interv{1}{k}$, such that $y \in \ell_{f_m(y)}$ for all $y \in
  \varsof{m}$ and $x \in \ell_{f_n(x)}$ for all $x \in \varsof{n}$. In
  other words, the witness function of a monomial $m \in \norm(t) =
  \norm(\ell_1 \cdot ... \cdot \ell_k)$ describes from which (unique)
  linear term of the product each variable of $m$ has been taken. Note
  that the witness function is an injective choice function (i.e., the
  same variable can appear in several linear terms), where distinct
  variables must be associated to distinct linear terms. Consider now
  the partial function $\pi : \varsof{n} \rightharpoonup \varsof{m}$
  defined as follows:

  \begin{figure}[ht]
    \begin{center}
      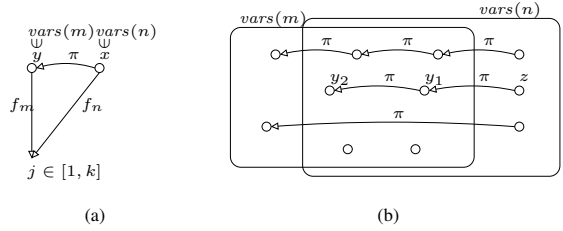
      \caption{\label{fig:exchange} The definition (a) and
        behavior (b) of the $\pi$ mapping.}
    \end{center}
  \end{figure}

  \[\pi(x) \isdef \left\{
  \begin{array}{rl}
    y & \text{if } f_n(x) = f_m(y), \mbox{ for some } y \in \varsof{m}
    \\ \bot & \text{otherwise} \end{array}
  \right.\]
  The definition of $\pi$ is illustrated in Figure~\ref{fig:exchange}
  (a). First of all, observe that $\pi(x)$ is consistently defined for
  all $x \in \varsof{n}$, because for distinct variables $y_1$, $y_2$
  in $\varsof{m}$ we have $f_m(y_1) \not= f_m(y_2)$ due to injectivity
  of $f_m$.  Second, observe that $\pi$ is injective (when restricted
  to its domain) because $\pi(x_1) = \pi(x_2) = y$ implies $f_n(x_1) =
  f_n(x_2) = f_m(y)$ and hence $x_1 = x_2$ due to the injectivity of
  $f_n$. Consider now the variables in $\varsof{n} \setminus
  \varsof{m}$, as in Figure~\ref{fig:exchange} (b). We distinguish two
  cases:
  \begin{itemize}
  \item There exists $z \in \varsof{n} \setminus \varsof{m}$ such that
    $\pi(z) = \bot$.  That means, $z$ is mapped by $f_n$ to some index
    which does not occur in the image of $f_m$. In this case, we
    obtain $m \cdot z\in\norm(t)$, witnessed by the injective function
    $f_m \cup \set{z \mapsto f_n(z)}$.
  \item For all $z \in \varsof{n} \setminus \varsof{m}$ we have
    $\pi(z) \not= \bot$.  Then, for each such $z$ there exists a
    non-empty \emph{maximal chain} of variables $y_0, y_1, ..., y_p$
    following the application of $\pi$, that is such that (i) $y_0 = z
    \in \varsof{n} \setminus \varsof{m}$, $y_i \in \varsof{n} \cap
    \varsof{m}$ for all $i \in \interv{1}{p-1}$, $y_p \in \varsof{m}$,
    and (ii) $\pi(y_{i-1}) = y_i$ for all $i\in\interv{1}{p}$,
    $\pi(y_p) = \bot$. Note that due to the injectivity of $\pi$, all
    these maximal chains constructed from distinct variables in
    $\varsof{n} \setminus \varsof{m}$ go through pairwise disjoint
    sets of variables. We distinguish now two
    sub-cases: \begin{itemize}
    \item there exists a maximal chain $z, y_1, ..., y_p$ ending
      in $\varsof{m} \cap \varsof{n}$.  Then, $m\cdot z \in \norm(t)$,
      witnessed by the injective function $f_m[ y_i \mapsto f_n(y_i)]_{i=1,p}
      \cup \set{z \mapsto f_n(z)}$.
    \item all maximal $\pi$-chains end at variables $\varsof{m} \setminus
      \varsof{n}$. This means that all the variables from $\varsof{n}
      \setminus \varsof{m}$ are uniquely mapped by these maximal chains
      to variables from $\varsof{m} \setminus \varsof{n}$.
      Consequently, we must have $\cardof{\varsof{m} \setminus
        \varsof{n}} \ge \cardof{\varsof{n} \setminus \varsof{m}}$
      which entails $\cardof{\varsof{m}} \ge \cardof{\varsof{n}}$,
      thus contradicting $\lenof{m} < \lenof{n}$.
    \end{itemize}
  \end{itemize}
\end{proofE}

\begin{example}\label{ex:complexity:2-threshold-only}
  Consider threshold variables $\svar_1,\ldots,\svar_4 \in
  \ThresholdVars$ such that $\threshold{\svar_i}=2$, for all
  $i\in\interv{1}{4}$.  Let $t \isdef (\svar_1 + \svar_2) (\svar_2 +
  \svar_3) (\svar_3 + \svar_4)$. We compute: \begin{align*} \norm(t) =
    & ~\svar_1 \svar_2 \svar_3 + \svar_1 \svar_2 \svar_4 + \svar_1
    \svar_3 + \svar_1 \svar_3 \svar_4 + \svar_2 \svar_3 + \svar_2
    \svar_3 \svar_4 \\ \sup(\norm(t)) = & ~\svar_1 \svar_2 \svar_3 +
    \svar_1 \svar_2 \svar_4 + \svar_1 \svar_3 \svar_4 + \svar_2
    \svar_3 \svar_4
  \end{align*}
  and observe that $\maxdeg(\norm(t)) = 3$ and $\sup(\norm(t)) = \sum
  \{m \in \norm(t) ~|~ \deg(m) = 3\}$.
\end{example}

%
%

For two products $t_1,t_2 \in \Prods{U}$, we write $t_1 \leprod t_2$
if $t_1$ is a subproduct of $t_2$, i.e., $t_1$ is the product of a
subset of the factors of $t_2$. Note that $\leprod$ agrees with
$\lemon$ over reduced monomials. We write $t_1 \ltprod t_2$ whenever
$t_1 \leprod t_2$ and $t_1 \not= t_2$ and define $\leinterv{t_1}{t_2}
\isdef \set{t ~\mid~ t_1 \leprod t \leprod t_2}$. With these
notations, Lemma~\ref{lemma:complexity:threshold-only} generalizes
Lemma~\ref{lemma:complexity:2-threshold-only} from $2$ to arbitrary
thresholds and provides a cut-off bound on the linear products
consisting of threshold variables only:

\begin{lemmaE}\label{lemma:complexity:threshold-only}
  For each linear product $t\in\Prods{\Sums{\ThresholdVars}}$, the
  following hold:
  \begin{enumerate}[(i)]
  \item\label{it1:lemma:complexity:threshold-only} $\sup(\norm(t)) =
    \sum \set{ m \in \norm(t) ~\mid~ \deg(m) = \maxdeg(\norm(t))}$.
  \item\label{it2:lemma:complexity:threshold-only} $\sup(\norm(t)) =
    \sup(\norm(t \cdot \ell))$, for each linear term $\ell \in
    \Sums{\ThresholdVars}$ such that $ \maxdeg(\norm(t)) = \maxdeg(\norm(t
    \cdot \ell))$.
  \item\label{it3:lemma:complexity:threshold-only} if $\lengthof{t} >
    \weightedcardof{\ThresholdVars}$ then there exists $t' \ltprod t$,
    such that $\sup(\norm(t)) = \sup(\norm(t''))$, for all $t'' \in
    \leinterv{t'}{t}$.
  \end{enumerate}
\end{lemmaE}
\begin{proofE}
  \noindent(\ref{it1:lemma:complexity:threshold-only}) This is a
  consequence of Lemma~\ref{lemma:complexity:2-threshold-only}. We
  define a set $\SVars^\sharp$ of fresh \emph{2-threshold}
  variables: \[\SVars^\sharp \isdef \set{ \svar^\sharp_i ~\mid~
    \svar\in\ThresholdVars{}, i\in \interv{1}{\threshold{s} - 1}}\]
  satisfying $\threshold{\svar^\sharp_i} = 2$, for all $\svar^\sharp_i
  \in \SVars^\sharp$. Note that $\weightedcardof{\ThresholdVars} =
  \weightedcardof{\SVars^\sharp}$.  We define substitutions $\sigma:
  \ThresholdVars \rightarrow \Sums{\SVars^\sharp}$ and $\sigma^{(-1)}:
  \SVars^\sharp \rightarrow \ThresholdVars$ by taking respectively
  $\sigma(\svar) \isdef \svar^\sharp_1 + ... +
  \svar^\sharp_{\threshold{s} - 1}$, $\sigma^{(-1)}(\svar^\sharp_i)
  \isdef \svar$, for all $\svar\in \ThresholdVars$, for all $i \in
  \interv{1}{\threshold{\svar}}$. Using these notations, we will show
  that computations illustrated in Fig.~\ref{fig:threshold} commute,
  that is, formally:
  \begin{equation}\label{eqn:complexity:threshold-only}
    \sup(\norm(t)) = \norm(
      \sigma^{(-1)}(\sup(\norm(\sigma(t)))))
  \end{equation}
  for each $t\in\Prods{\Sums{\ThresholdVars}}$.

  \begin{figure}[ht]
    \begin{center}
      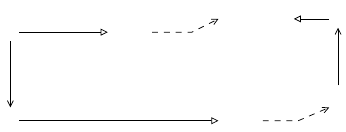
      \caption{\label{fig:threshold} Linear products of $2$-threshold variables}
    \end{center}
  \end{figure}

  Using property~(\ref{eqn:complexity:threshold-only}) we complete the
  proof of (\ref{it1:lemma:complexity:threshold-only}) as follows.
  First, observe that Lemma~\ref{lemma:complexity:2-threshold-only}
  applies for $t^\sharp \isdef \sigma(t)$, hence,
  $\sup(\norm(t^\sharp)) = \sum_{m^\sharp \in \norm(t^\sharp)} \set{
    m^\sharp ~\mid~ \deg(m^\sharp) = \maxdeg(\norm(t^\sharp))}$.
  Then, we get from (\ref{eqn:complexity:threshold-only}) that:
  \begin{align*}
    & \sup(\norm(t)) ~= \\
    & \norm(\sigma^{(-1)}(\sup(\norm(t^\sharp)))) ~= \\
    & \norm(\sigma^{(-1)}(\sum_{m^\sharp \in \norm(t^\sharp)} \set{ m^\sharp ~\mid~ \deg(m^\sharp) = \maxdeg(\norm(t^\sharp))})) ~= \\
    & \norm\left(\sum_{m^\sharp \in \norm(t^\sharp)} \set {\sigma^{(-1)}(m^\sharp) ~\mid~ \deg(m^\sharp) = \maxdeg(\norm(t^\sharp))}\right)
  \end{align*}
  Note that $\sigma^{(-1)}$ is a substitution ensuring that
  $\sigma^{(-1)}(m^\sharp)$ is a reduced monomial on $\ThresholdVars$
  variables of same degree as $m^\sharp$, for every reduced monomial
  $m^\sharp$ in $\norm(t^{\sharp})$.  Henceforth, from above we obtain
  that all monomials in $\sup(\norm(t))$ have the same degree.  We can
  now easily establish the equality of
  (\ref{it1:lemma:complexity:threshold-only}) by checking the double
  inclusion:

  \noindent ''$\supseteq$'': Let $m \in \norm(t)$ such that $\deg(m) =
  \maxdeg(\norm(t))$.  Assume by contradiction that $m \not\in
  \sup(\norm(t))$.  Then, there exists $m' \in \norm(t)$ such that $m
  \ltmon m'$ and hence $\deg(m) < \deg(m')$.  But, this leads to
  $\maxdeg(\norm(t)) = \deg(m) < \deg(m') \le \maxdeg(\norm(t))$ which
  is impossible.

  \noindent ''$\subseteq$'': Assume by contradiction there exists $m
  \in \sup(\norm(t))$ and $\deg(m) < \maxdeg(\norm(t))$.  From
  previous inclusion, $\sup(\norm(t))$ contains all monomials of
  $\norm(t)$ of maximal degree $\maxdeg(\norm(t))$. But then, this
  contradicts that all monomials from $\sup(\norm(t))$ have the same
  degree.

  We are left to prove property~(\ref{eqn:complexity:threshold-only}).
  The next two assertions show that the ordering between reduced
  monomials are preserved by $\sigma$ and $\sigma^{(-1)}$:
  \begin{itemize}
    \item $\forall m_1, m_2 \in \norm(t).~ m_1 \ltmon m_2 \implies \exists
      m_1^\sharp, m_2^\sharp \in \norm(t^\sharp).~ m_1^\sharp \in
      \norm(\sigma(m_1)), m_2^\sharp \in \norm(\sigma(m_2)),
      m_1^\sharp \ltmon m_2^\sharp$: we associate to every monomial $m
      \isdef \prod_{i=1}^n \svar_i^{p_i} \in \norm(t)$ the unique
      representative $m^\star \isdef \prod_{i=1}^n \prod_{j=1}^{p_i}
      \svar^\sharp_{i,j}$ i.e., by expanding powers to products of
      distinct replicated variables.  It is an easy check that
      $m^\star \in \norm(\sigma(m))$ for any reduced monomial, and
      that the assertion holds by choosing $m_1^\sharp \isdef
      m_1^\star$, $m_2^\sharp \isdef m_2^\star$.

    \item $\forall m_1^\sharp, m_2^\sharp \in \norm(t^\sharp).~
      m_1^\sharp \ltmon m_2^\sharp \implies \sigma^{(-1)}(m_1^\sharp)
      \ltmon \sigma^{(-1)}(m_2^\sharp)$: It is an easy check that the
      assertion holds as $\sigma^{(-1)}$ is fusioning back the
      replicated variables.
  \end{itemize}
  The property~(\ref{eqn:complexity:threshold-only}) follows
  immediately from the two assertions above.

  \noindent(\ref{it2:lemma:complexity:threshold-only}) Let $r \isdef
  \maxdeg(\norm(t)) = \maxdeg(\norm(t \cdot \ell))$. We will show equality
  $\sup(\norm(t)) = \sup(\norm(t \cdot \ell))$ by double inclusion of maximal
  monomials.

  \noindent ''$\subseteq$'' Let $m \in \sup(\norm(t))$. Then, we obtain
  $\deg(m) = r$ from (\ref{it1:lemma:complexity:threshold-only}). Let $\svar
  \in \ell$ and consider $m' \isdef \norm(m \cdot \svar) \in
  \norm(t \cdot \ell)$. Then, either one of the following cases
  holds: \begin{itemize}
  \item $\deg(m') = \deg(m)$: in this case, $m' = m$ and since
    $\deg(m) = \maxdeg(\norm(t \cdot \ell))$, we obtain $m \in
    \sup(\norm(t \cdot \ell))$, by
    (\ref{it1:lemma:complexity:threshold-only}).
  \item $\deg(m') = \deg(m) + 1$: in this case, $\deg(m') = r+1 > r =
    \maxdeg(\norm(t \cdot \ell))$, contradiction.
  \end{itemize}

  \noindent ''$\supseteq$'' Let $m'\in \sup(\norm(t \cdot \ell))$ be a
  monomial. Then, we obtain $\deg(m') = r$ from
  (\ref{it1:lemma:complexity:threshold-only}). Let $m_1 \in \norm(t)$ and $\svar
  \in \ell$ such that $m' = \norm(m_1 \cdot \svar)$.  We
  distinguish two cases: \begin{itemize}
  \item $\deg(m_1) = \deg(m')$: in this case, $m' = m_1$ and then
    since $\deg(m') = r$, we obtain $m' \in \sup(\norm(t))$, by
    (\ref{it1:lemma:complexity:threshold-only}).
  \item $\deg(m_1) = \deg(m') - 1$: in this case, $\deg(m_1) = r -
    1$ and there must exists $m \in \sup(\norm(t))$ such that
    $\deg(m)=r$ and $m_1 \ltmon m$.  We distinguish two subcases:
    \begin{itemize}
    \item $\deg(\norm(m \cdot s)) = deg(\norm(m_1 \cdot s))$: in this
      case, we have $m = \norm(m_1 \cdot s) = m'$ and so $m' \in
      \sup(\norm(t))$,
    \item $\deg(\norm(m \cdot s)) = deg(\norm(m_1 \cdot s)) + 1$: in
      this case, $\deg( \norm(m\cdot s) ) = r + 1 > \maxdeg(\norm(t
        \cdot \ell))$, contradiction as $\norm(m \cdot s) \in \norm(t
        \cdot \ell)$.
    \end{itemize}
  \end{itemize}

  \noindent(\ref{it3:lemma:complexity:threshold-only}) Consider $t
  \isdef \ell_1 \cdot ... \cdot \ell_k$, with $\ell_i \in
  \Sums{\ThresholdVars}$ for all $i\in\interv{1}{k}$, for some $k \ge 1$.
  Let $r \isdef \maxdeg(\norm(t))$.  We choose $r$ pairwise distinct
  factors $\ell_{i_1}$, ..., $\ell_{i_r}$ amongst the $k$ factors of
  $t$ such that for $t' \isdef \ell_{i_1} \cdot ... \cdot \ell_{i_r}$
  we have $\maxdeg(\norm(t')) = r$.  Let $J \isdef \interv{1}{k}
  \setminus \set{i_1, ..., i_r}$.  We can check that $t'$ satisfies
  our requirements:
  \begin{itemize}
  \item if $\lengthof{t} > \weightedcardof{\ThresholdVars}$ then $t' \ltprod
    t$: actually, the maximal degree of reduced monomials constructed
    with variables from $\ThresholdVars$ is $\weightedcardof{\ThresholdVars}$
    hence, $r \le \weightedcardof{\ThresholdVars} < \lengthof{t} = k$ and
    so $J \not=\emptyset$ and hence $t' \ltprod t$,
  \item for all $t'' \in \interv{t'}{t}$ we have $\sup(\norm(t'')) =
    \sup(\norm(t))$: that is, because $t = t' \cdot \prod_{j \in J}
    \ell_j$ and since $\maxdeg(\norm(t)) = \maxdeg(\norm(t')) = r$ all
    the intermediate products leading from $t'$ to $t$ by adding the
    extra factors (with indices from $J$) will preserve the maximal
    degree of monomials and hence, the set of maximal monomials
    according to (\ref{it2:lemma:complexity:threshold-only}).
  \end{itemize}
\end{proofE}

%
%

For each $V \subseteq \SVars$, we denote by $\XSums{V}{1} \isdef
\set{1 + \sum\nolimits_{\svar\in W} \svar ~\mid~ \emptyset \not= W
  \subseteq V}$ the set of linear products whose factors all contain
the empty monomial
$1$. Lemma~\ref{lemma:complexity:1-bounded-periodic} gives a cut-off
on the length of products of such linear terms, relying on simple
observation: a linear term $1+\svar$, for $\svar \in \nopump{} \uplus \pump{}$, behaves as a threshold variable of threshold equal to $\base{\svar}{}$ resp.
$\period{\svar}{}$. Intuitively, this is the case because the normal
forms of such linear products are downward closed w.r.t. the $\lemon$
order on monomials. Hence, the bounded and periodic variables can be
syntactically substituted with threshold variables and
Lemma~\ref{lemma:complexity:threshold-only} can be applied.

\begin{lemmaE}\label{lemma:complexity:1-bounded-periodic}
  For each linear product $t\in \Prods{\XSums{\nopump{} \uplus
      \pump{}}{1}}$ of length $\lengthof{t} >
  \weightedcardof{\nopump{} \uplus \pump{}}$, there exists $t' \ltprod
  t$ such that $\norm(t) = \norm(t'')$, for all $t'' \in
  \leinterv{t'}{t}$.
\end{lemmaE}
\begin{proofE}
We define a set of fresh \emph{threshold} variables $\SVars^\sharp
\isdef \set{ \svar^\sharp ~\mid~ \svar \in \SVars}$ where
$\threshold{\svar^\sharp} \isdef \base{\svar}{}$ if
$\svar\in\nopump{}$ and $\threshold{ \svar^\sharp}
\isdef\period{\svar}{}$ if $\svar\in\pump{}$.  Note that
$\weightedcardof{\SVars^\sharp} = \weightedcardof{\nopump{} \uplus
  \pump{}}$.  We moreover define the substitutions $\sigma : \SVars
\rightarrow \SVars^\sharp$ and $\sigma^{(-1)} : \SVars^\sharp
\rightarrow \TermsOf{\SVars}$ defined respectively by $\sigma(\svar)
\isdef \svar^\sharp$ and $\sigma^{(-1)}(\svar^\sharp) \isdef 1 + \svar$,
for every $\svar \in \SVars$.  Consider arbitrary $t \isdef (1 +
\ell_1) \cdot ... \cdot (1 + \ell_k)$ where $\ell_i \in
\Sums{\nopump{} \uplus \pump{}}$, for all $i \in \interv{1}{k}$, for
some $k \ge 1$.  We denote by $t^\sharp \isdef \sigma(\ell_1) \cdot
... \cdot \sigma(\ell_k) \in \Prods{\Sums{\SVars^\sharp}}$. Using
these notations, we will show that computations depicted in
Fig.~\ref{fig:1-sums} commute, that is, formally:
\begin{equation}\label{eqn:complexity:1-bounded-periodic}
\norm(t) = \norm(
  \sigma^{(-1)}(\norm(t^\sharp))) = \norm(
  \sigma^{(-1)}(\sup(\norm(t^\sharp))))
\end{equation}
for each $t\in\Prods{\Sums{\ThresholdVars}}$.

\begin{figure}[ht]
  \begin{center}
    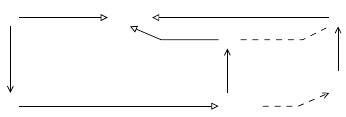
    \caption{\label{fig:1-sums} Linear products of threshold variables}
  \end{center}
\end{figure}

The property~(\ref{eqn:complexity:1-bounded-periodic}) is then
sufficient to establish the conclusion of the lemma, as follows.
Consider $\lengthof{t} > \weightedcardof{\nopump{} \uplus \pump{}}$.
Then, $t^\sharp \in \Prods{\Sums{\SVars^\sharp}}$ and
$\lengthof{t^\sharp} = \lengthof{t} > \weightedcardof{\SVars^\sharp}$.
Using Lemma~\ref{lemma:complexity:threshold-only}
(\ref{it3:lemma:complexity:threshold-only}) there exists $t'^{\sharp}
\in \Prods{\Sums{\SVars^\sharp}}$ such that $t'^\sharp \ltprod t^\sharp$
and forall $u \in \interv{t'^\sharp}{t^\sharp}$ it holds
$\sup(\norm(u)) = \sup(\norm(t^\sharp))$.  Then, we construct $t'
\isdef \sigma^{(-1)}(t'^\sharp) \in \Prods{\XSums{\SVars}{1}}$ and we
check that $t' \ltprod t$ (because $t'^\sharp \ltprod t^\sharp$) and, for
all $t'' \in \interv{t'}{t}$ we have $\norm(t'') =
\norm(\sigma^{(-1)}(\sup(\norm(t''^\sharp)))) =
\norm(\sigma^{(-1)}(\sup(\norm(t^\sharp)))) = \norm(t)$ using that
$t''^\sharp \in \interv{t'^\sharp}{t^\sharp}$.

We are left to prove property~(\ref{eqn:complexity:1-bounded-periodic}).
We proceed by showing the two equalities:

\noindent ''$\norm(t) = \norm(\sigma^{(-1)}(\norm(t^\sharp)))$'': Note
that $(1 + \svar)^m = 1 + \svar + ... + \svar^m$ for every $\svar \in
\nopump{} \uplus \pump{}$ and positive integer $m$.  Hence, $(1 +
\svar)$ \emph{behaves} like the threshold variable $\svar^\sharp$ with
respect to multiplication, that is, either $(1 +
\svar)^{\base{\svar{}}{}} = (1 + \svar)^{\base{\svar}{} - 1}$ if
$\svar \in \nopump{}$ or $(1 + \svar)^{\period{\svar}{}} = (1 +
\svar)^{\period{\svar}{} - 1}$ if $\svar \in \pump{}$. These
identities allow us to carry the computation of $\norm(t)$ by
rewriting factors $1 + \svar_1 + ... + \svar_n$ into $(1 + \svar_1) +
... + (1 + \svar_n)$ and leveraging from the fact that any sum $1 +
\svar_j$ behaves like the threshold variable $\svar_j^\sharp$.
Henceforth, $\norm(t)$ can be equivalently computed by leveraging its
$1 + \ell_i$ factors in $\XSums{\SVars}{1}$ to $\sigma(\ell_i)$
factors in $\Sums{\SVars^\sharp}$, performing the product with these
new factors, and then, getting back by the reverse substitution
$\sigma^{(-1)}$, as follows:
\begin{align*}
  t &= (1 + \ell_1)\cdot ... \cdot (1 + \ell_k) \\
  &= \sigma^{(-1)} (\sigma(\ell_1) \cdot ... \cdot \sigma(\ell_k)) \\
  &= \sigma^{(-1)} (t^\sharp)
\end{align*}
Then, we obtain $\norm(t) = \norm(\sigma^{(-1)}(\norm(t^\sharp)))$.

\noindent "$\norm(\sigma^{(-1)}(\norm(t^\sharp))) = \norm(
  \sigma^{(-1)}(\sup(\norm(t^\sharp))))$": We first prove the
following fact:

\begin{quote}
If $m_1^\sharp, m_2^\sharp$ are reduced monomials consisting of
$\SVars^\sharp$ variables, such that $m_1^\sharp \lemon m_2^\sharp$
then, any reduced monomial $m$ in $\norm(\sigma^{(-1)}(m_1^\sharp))$ occurs
as reduced monomial in $\norm(\sigma^{(-1)}(m_2^\sharp))$.
\end{quote}
Without loss of generality, consider:
\begin{align*}
  \sigma^{(-1)}(m_1^\sharp)  = &~
  (1 + \svar_1)^{p_1} \cdot ... \cdot (1 + \svar_{n_1})^{p_{n_1}} \\
  \sigma^{(-1)}(m_2^\sharp) = &~
  (1 + \svar_1)^{q_1} \cdot ... \cdot (1 + \svar_{n_1})^{q_{n_1}} ~\cdot \\
  & ~(1 + \svar_{n_1+1})^{q_{n_1+1}} \cdot ... \cdot (1 + \svar_{n_2})^{q_{n_2}}
\end{align*}
for some pairwise distinct variables $\svar_1$, ..., $\svar_{n_2}$,
for positive $p_1$, ..., $p_{n_1}$, $q_1$, ..., $q_{n_2}$ such that
$q_i \ge p_i$ for all $i\in\interv{1}{n_1}$, and some $n_1 \le n_2$.
Then, any reduced monomial $m$ in $\norm(\sigma^{(-1)}(m_1^\sharp))$
is obtained from some product $\svar_1^{r_1} \cdot ... \cdot
\svar_{n_1}^{r_{n_1}}$ occurring in $\sigma^{(-1)}(m_1^\sharp)$.  Or,
exactly the same product exists in $\sigma^{(-1)}(m_2^\sharp)$ i.e.,
by making the same choices of $\svar_1$, ..., $\svar_{n_1}$ variables
in $\sigma^{(-1)}(m_1^\sharp)$ part, and choosing 1's from the other
additional factors in $\sigma^{(-1)}(m_2^\sharp)$.

Using the fact above, we obtain immediately that $\norm(\sigma^{(-1)}(
  \norm(t^\sharp))) = \norm(\sigma^{(-1)}(\sup(\norm(t^\sharp)))$
because every monomial derived from a non-maximal monomial of
$\sigma^{(-1)}(\norm(t^\sharp))$ (that is, $t_1$ in
Fig.~\ref{fig:1-sums}) can be also derived from some maximal monomial
in $\sigma^{(-1)}(\sup(\norm(t^\sharp)))$ (that is, $t_2$ in
Fig.~\ref{fig:1-sums}).
\end{proofE}

\begin{example}\label{ex:complexity:1-bounded-periodic}
  Let $\svar_1 \in \nopump{}{}$ and $\svar_2\in \pump{}$ be variables
  such that $\base{\svar_1}{} = 2$, $\period{\svar_2}{} = 3$ and $t
  \isdef (1 + \svar_1 + \svar_2) (1 + \svar_1) (1 + \svar_2) (1 +
  \svar_1) (1+ \svar_1) (1 + \svar_1 + \svar_2) (1 + \svar_2)$.  Then
  $\norm(t) = 1 + \svar_1 + \svar_2 + \svar_1 \svar_2 + \svar_2^2 +
  \svar_1 \svar_2^2$ and note that $\norm(t)$ contains all reduced
  monomials constructed from $\svar_1$ and $\svar_2$.  Let $t' \isdef
  (1 + \svar_1 + \svar_2) (1 + \svar_1) (1 + \svar_2)$, that is, $t'$
  contains only the first three factors of $t$.  Then, $\norm(t') =
  \norm(t)$ is an easy check. Moreover, $\norm(t'') = \norm(t)$ for
  all $t''\in\leinterv{t'}{t}$, because $\norm(t'')$ contains already
  all the monomials of $\norm(t')$, i.e., of $\norm(t)$.
\end{example}

%
%

For each $s_0 \in \SVars$ and $V \subseteq \SVars$, we denote by
$\XSums{V}{\svar_0} \isdef \set{ \svar_0 + \sum\nolimits_{\svar\in W}
  \svar ~\mid~ \emptyset \not= W \subseteq V}$ the set of linear terms
whose factors all contain a distinguished variable
$s_0$. Lemmas~\ref{lemma:complexity:s0-bounded-periodic-basic} and
\ref{lemma:complexity:s0-bounded-periodic} generalize
Lemma~\ref{lemma:complexity:1-bounded-periodic} by considering
products of linear terms containing a periodic variable $s_0$ instead
of $1$. In particular,
Lemma~\ref{lemma:complexity:s0-bounded-periodic-basic} considers the
situation where the period of $\svar_0$ divides the periods of all
other periodic variables occurring in the product. In this case, the
normal form of a linear product becomes downward closed
w.r.t. $\lemon$ when $s_0$ is removed from all monomials. We point out that for the products whose length exceeds the cut-off bound, Lemma~\ref{lemma:complexity:s0-bounded-periodic-basic}
only guarantees some strictly shorter product having the same normal form, whereas
Lemma~\ref{lemma:complexity:1-bounded-periodic} guarantees equal normal forms for all linear
products whose length exceeds the cut-off bound.

\begin{lemmaE}\label{lemma:complexity:s0-bounded-periodic-basic}
  For each periodic variable $\svar_0\in \pump{}$, whose period
  $\period{\svar_0}{}$ divides all periods $\set{\period{\svar}{} \mid
    \svar\in\pump{}}$, and each linear product $t\in
  \Prods{\XSums{\nopump{} \uplus \pump{}}{\svar_0}}$ of length
  $\lengthof{t} > \weightedcardof{\nopump{} \uplus \pump{}}$, there
  exists $t'\ltprod t$ such that $\norm(t') = \norm(t)$.
\end{lemmaE}
\begin{proofE}
  Assume wlog $t \isdef (\svar_0 + \ell_1) \cdot ... \cdot (\svar_0 +
  \ell_k)$ with $\ell_i \in \Sums{\nopump{} \uplus \pump{} \setminus
    \set{\svar_0}}$ for all $i\in \interv{1}{k}$, for some $k \ge 1$.
  Define $t_1 \isdef (1 + \ell_1) \cdot ... \cdot(1 + \ell_k) \in
  \Prods{\XSums{\nopump{} \uplus \pump{} \setminus \set{\svar_0}}{1}}$.
  Using these notations, we will first establish the following property:
  \begin{equation}\label{eqn:complexity:s0-bounded-periodic-basic}
    \norm(t) = \sum_{m_1 \in \norm(t_1)}
    \svar_0^{(k - \deg(m_1))\mod \period{\svar_0}{}} \cdot m_1
  \end{equation}
  for each $t\in\Prods{\XSums{\nopump{} \uplus \pump{}}{\svar_0}}$.
  We proceed by showing the double inclusion of reduced monomials:

  \noindent ''$\subseteq$'': Let $m \isdef \svar_0^{r_0} \cdot \svar_1^{r_1} \cdot
  ... \cdot \svar_n^{r_n} \isdef \norm(m_u) \in \norm(t)$ where $m_u \isdef
  \svar_0^{p_0} \cdot \svar_1^{p_1}... \cdot \svar_n^{p_n}$ is an
  unreduced monomial occurring in the expansion of $t$.  Then:
  \begin{itemize}
  \item $p_0 + p_1 + ... + p_n = k$ and because $\period{\svar_0}{}$
    divides the periods of all other periodic variables we have $r_0 +
    r_1 + ... + r_n \equiv (p_0 + p_1 + ... + p_n) \mod
    \period{\svar_0}{0} \equiv k \mod \period{\svar_0}{0}$ and
    consequently $r_0 = (k - r_1 - ... - r_n) \mod
    \period{\svar_0}{0}$.
  \item $m_{1u} \isdef 1^{p_0} \cdot \svar_1^{p_1} \cdot ... \cdot
    \svar_n^{p_n}$ is an unreduced monomial in the expansion of $t_1$,
    that is, which chooses 1 instead of $\svar_0$ and maintains the other
    choices for the other variables.  Let then $m_1 \isdef
    \norm(m_{1u}) \in \norm(t_1)$.  We can check easily that $m_1 =
    \svar_1^{r_1} \cdot ... \cdot \svar_n^{r_n}$ as the same axioms of
    multiplications apply for $\svar_1$, ..., $\svar_n$.  Then,
    $\deg(m_1) = r_1 + .... + r_n$ and consequently $r_0 = (k
    -\deg(m_1)) \mod \period{\svar_0}{0}$.
  \end{itemize}
  Hence, $m = \svar_0^{r_0} \cdot m_1$, where $r_0 = (k - \deg(m_1))
  \mod \period{\svar_0}{0}$ and $m_1 \in \norm(t_1)$.

  \noindent ''$\supseteq$'': Let $m_1 \isdef \svar_1^{r_1} \cdot ... \cdot
  \svar_n^{r_n} \isdef \norm(m_{1u}) \in \norm(t_1)$ where $m_{1u}
  \isdef \svar_1^{p_1} \cdot ... \cdot \svar_n^{p_n}$ is an unreduced
  monomial occurring in the expansion of $t_1$.  Then:
  \begin{itemize}
  \item $p_1 + p_2 + ... + p_n \le k$ and $m_u \isdef \svar_0^{p_0}
    \cdot \svar_1^{p_1} \cdot ... \cdot \svar_n^{p_n}$ where $p_0
    \isdef k-p_1 - ... - p_n$ is an unreduced monomial occurring in
    the expansion of $t$, that is, chooses $\svar_0$ instead of 1 and
    maintains the other choices for the other variables.  Let then $m
    = \norm(m_u)$.  We can easily check that $m = \svar_0^{ p_0 \mod
      \period{\svar_0}{} } \cdot \svar_1^{r_1} \cdot ... \cdot
    \svar_n^{r_n} = \svar_0^{ p_0 \mod \period{\svar_0}{} } \cdot m_1$
    given the same axioms of multiplications apply for $\svar_1$, ...,
    $\svar_n$ in $t$ and $t_1$.
  \item $p_0 \mod \period{\svar_0}{} = (k - p_1 - ... - p_n) \mod
    \period{\svar_0}{} = (k - r_1 - ...- r_n) \mod \period{\svar_0}{}
    = (k - \deg(m_1)) \mod \period{\svar_0}{}$ that is, because
    $\period{\svar_0}{}$ divides the periods of all other periodic
    variables.
  \end{itemize}
  Hence, $m = \svar_0^{r_0} \cdot m_1$ with $r_0
  \isdef (k - \deg(m_1)) \mod \period{\svar_0}{}$ belongs to
  $\norm(t)$, as required.

  Using property~(\ref{eqn:complexity:s0-bounded-periodic-basic}) we
  complete the proof of the lemma as follows:
  \begin{align*}
    \lengthof{t_1} = & ~\lengthof{t} = k \\
    > & ~\weightedcardof{\nopump{} \uplus \pump{}} \\
    = & ~\weightedcardof{ \nopump{} \uplus \pump{} \setminus \set{\svar_0}} +
    (\period{\svar_0}{} - 1)
  \end{align*}
  By (possibly repeated application of)
  Lemma~\ref{lemma:complexity:1-bounded-periodic} there exists $t_1'
  \ltprod t_1$ such that $\lengthof{t_1'} \le \weightedcardof{
    \nopump{} \uplus \pump{} \setminus \set{\svar_0}}$ and forall $u
  \in\interv{t_1'}{t_1}$ it holds $\norm(u) = \norm(t_1)$. Now, we
  proceed to the construction of $t'$ from $t_1$ and $t_1'$ as
  depicted in the Fig~\ref{fig:s0-bounded-periodic} in the following
  steps:

  \begin{figure}[htbp]
  \[ \begin{array}{lcr}
    t_1 = & \underbrace{
      \underbrace{\overbrace{\prod\nolimits_{j \in J_1}  (~1 + \ell_j)}^{J_1}}_{t_1'} \cdot
      \overbrace{\prod\nolimits_{ j \in J_2'} (~1 + \ell_j)}^{J_2'}}_{t_1''} &
    \cdot \underbrace{\overbrace{\prod\nolimits_{j \in J_2''}(~1 + \ell_j)}^{J_2''}}_{\cardof{J_2''} = \period{\svar_0}{}}\\
    t = & \underbrace{\prod\nolimits_{j \in J_1}  (\svar_0 + \ell_j) \cdot
      \prod\nolimits_{ j \in J_2'} (\svar_0 + \ell_j)}_{t'} &
    \cdot \prod\nolimits_{j \in J_2''}(\svar_0 + \ell_j)
  \end{array} \]
  \caption{\label{fig:s0-bounded-periodic} Construction of $t'$ from $t_1$ and $t_1'$}
  \end{figure}

  \begin{itemize}
  \item let $J_1 \subseteq \interv{1}{\lengthof{t_1}}$ such that $t_1'
    = \prod_{j \in J_1} (1 + \ell_j)$ and let $J_2 \isdef
    \interv{1}{\lengthof{t_1}} \setminus J_1$.  We can check that
    $\cardof{J_2} \ge \period{\svar_0}{}$, that is, because
    $\cardof{J_2} = \lengthof{t_1} - \lengthof{t_1'} >
    \weightedcardof{ \nopump{} \uplus \pump{} \setminus \set{\svar_0}}
    + (\period{\svar_0}{} - 1) - \weightedcardof{ \nopump{} \uplus
      \pump{} \setminus \set{\svar_0}} = (\period{\svar_0}{} - 1)$.
    That is, at least $\period{\svar_0}{0}$ more factors are added in
    $t_1$ with respect to $t_1'$.  Consider $J_2' \uplus J_2''$ a
    partition of $J_2$ such that $\cardof{J_2''} =
    \period{\svar_0}{0}$. Obviously, such a partitioning exists
    because $\cardof{J_2} \ge \period{\svar_0}{0}$.
  \item define $t' \isdef \prod_{j \in
      J_1} (\svar_0 + \ell_j) \cdot \prod_{j \in J_2'} (\svar_0 +
    \ell_j)$.  It is now an easy check that $t'$ satisfies the
    requirements as stated in the lemma.  First, $t' \ltprod t$, as
    precisely $\period{\svar_0}{}$ factors are removed in $t'$ with
    respect to $t$.  Second, let $t_1'' \isdef \prod_{j \in J_1} (1 +
    \ell_j) \cdot \prod_{j \in J_2'} (1 + \ell_j)$.  Observe that
    $t_1'' \in \interv{t_1'}{t_1}$ and hence $\norm(t_1'') =
    \norm(t_1)$.  Third, we check $\norm(t') = \norm(t)$ because
    \begin{align*}
      \norm(t')
      &= \sum_{m_1 \in \norm(t_1'')}
      \svar_0^{(k-\period{\svar_0}{} - \deg(m_1))\mod \period{\svar_0}{}} \cdot m_1 \\
      &= \sum_{m_1 \in \norm(t_1)}
      \svar_0^{(k - \deg(m_1)) \mod \period{\svar_0}{}} \cdot m_1
      = \norm(t)
    \end{align*}
  \end{itemize}
\end{proofE}

\begin{example}\label{ex:complexity:s0-bounded-periodic-basic}
  Consider a bounded variable $\svar_1 \in \nopump{}$ and periodic
  variables $\svar_0, \svar_2 \in \pump{}{}$ such that
  $\base{\svar_1}{} = 2$, $\period{\svar_0}{} = \period{\svar_2}{} =
  3$.  Let $t \isdef (\svar_0 + \svar_1 + \svar_2) (\svar_0 + \svar_1)
  (\svar_0 + \svar_2) (\svar_0 + \svar_1) (\svar_0 + \svar_1) (\svar_0
  + \svar_1 + \svar_2) (\svar_0 + \svar_2)$ and observe that
  $\lengthof{t} = 7$.  We compute $\norm(t) = \svar_0 + \svar_1 +
  \svar_2 + \svar_0^2 \svar_1 \svar_2 + \svar_0^2 \svar_2^2 + \svar_0
  \svar_1 \svar_2^2$.  Note that for all monomials $m\in\norm(t)$ we
  have $\deg(m) \equiv 7 \mod 3$, i.e., $\deg(m) \equiv \lengthof{t}
  \mod \period{\svar_0}{}$. Moreover, note that by erasing $\svar_0$
  from monomials of $\norm(t)$ we obtain all possible reduced
  monomials of $\svar_1$, $\svar_2$. Let then $t' \isdef (\svar_0 +
  \svar_1 + \svar_2) (\svar_0 + \svar_1) (\svar_0 + \svar_2) (\svar_0
  + \svar_1)$, that is, $t'$ contains the first four factors of $t$.
  Then $\norm(t') = \norm(t)$ is an easy check, in particular by
  taking into account the properties of $\norm(t)$ mentioned above.
\end{example}

%
%

Lemma~\ref{lemma:complexity:s0-bounded-periodic} generalizes
Lemma~\ref{lemma:complexity:s0-bounded-periodic-basic} by lifting the
restriction about the period of $\svar_0$ dividing all other
periods. This increases the cut-off bound above which factors become
redundant by a polynomial in the periods of the variables from
$\pump{}$, where $\lcm(a,b)$ denotes the least common multiple of
$a,b\in\nat^{\geq1}$:

\begin{lemmaE}\label{lemma:complexity:s0-bounded-periodic}
  For each periodic variable $\svar_0\in\pump{}$ and linear product
  $t\in \Prods{\XSums{\nopump{} \uplus \pump{}}{\svar_0}}$ of length
  $\lengthof{t} > \weightedcardof{\nopump{}} + \sum_{\svar \in
    \pump{}} (\lcm(\period{\svar}{},\period{\svar_0}{})-1)$, there
  exists $t'\ltprod t$ such that $\norm(t') = \norm(t)$.
\end{lemmaE}
\begin{proofE}Consequence of
  Lemma~\ref{lemma:complexity:s0-bounded-periodic-basic}.  We
  define a set of fresh \emph{periodic} variables $\SVars^\sharp
  \isdef \set{ \svar^\sharp ~\mid~ \svar \in \pump{}}$ where
  $\period{\svar^\sharp}{} \isdef
  \lcm(\period{\svar}{},\period{\svar_0}{})$ for every
  $\svar\in\pump{}$.  In particular, note that
  $\period{\svar_0^\sharp}{} = \period{\svar_0}{}$ and
  $\period{\svar_0^\sharp}{}$ divides $\period{\svar{}^\sharp}{}$ for
  every variable $\svar^\sharp \in \SVars^\sharp$.  Moreover, note
  that $\weightedcardof{\nopump{} \uplus \SVars^\sharp} =
  \weightedcardof{\nopump{}} + \sum_{\svar \in \pump{}}
  (\lcm(\period{\svar}{},\period{\svar_0}{})-1)$. Furthermore, define
  the substitutions $\sigma : \pump{} \rightarrow \SVars^\sharp$ and
  $\sigma^{(-1)} : \SVars^\sharp \rightarrow \pump{}$ by taking
  respectively $\sigma(\svar) \isdef \svar^\sharp$, $\sigma^{(-1)}
  (\svar^\sharp) \isdef \svar$, for every $\svar\in \pump{}$. Then,
  using the above notations we prove that computations depicted in
  Fig~\ref{fig:sums} commute, that is, formally:
  \begin{equation}
    \label{eqn:complexity:s0-bounded-periodic}
    \forall t\in\TermsOf{\nopump{} \uplus \pump{}}.~\norm(t) = \norm(
      \sigma^{(-1)}( \norm(\sigma(t))))
  \end{equation}

  \begin{figure}[ht]
  \begin{center}
    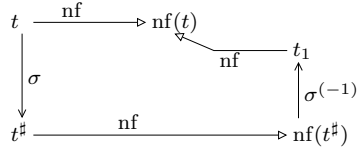
    \caption{\label{fig:sums} Linear products with enforced divisibilities}
  \end{center}
  \end{figure}

  The property~(\ref{eqn:complexity:s0-bounded-periodic}) is then
  sufficient to establish the conclusion of the lemma, as follows.
  Let $t$ satisfying the lemma hypothesis.  Then, $\sigma(t) \in
  \Prods{\XSums{\nopump{} \uplus \SVars^\sharp}{\svar_0^\sharp}}$ and
  $\lengthof{\sigma(t)} = \lengthof{t} > \weightedcardof{\nopump{}
    \uplus \SVars^\sharp}$.  Therefore, using
  Lemma~\ref{lemma:complexity:s0-bounded-periodic-basic} there
  exists $t'' \ltprod \sigma(t)$ such that $\norm(t'') =
  \norm(\sigma(t))$.  We construct $t' \isdef \sigma^{(-1)}(t'')$ and
  check that $t' \ltprod t$ (as $t'' \ltprod \sigma(t)$) and also:
  \begin{align*}
    \norm(t') = & ~\norm(\sigma^{(-1)}(t'')) \\
    = & ~\norm(\sigma^{(-1)}(\norm(t''))) \\
    = & ~\norm(\sigma^{(-1)}(\norm(\sigma(t)))) \\
    = & ~\norm(t)
  \end{align*}
  We are left with proving the
  property~(\ref{eqn:complexity:s0-bounded-periodic}). We proceed by a
  case split, depending on the form of $t$:
  \begin{itemize}
  \item if $t \isdef \svar^p$ for some $\svar\in\nopump{}$, $p \ge 1$ then
    \begin{align*}
      \norm(\sigma^{(-1)}( \norm(\sigma(t))))
      &= \norm(\sigma^{(-1)}(\norm(\sigma(\svar^p)))) \\
      &= \norm(\sigma^{(-1)}(\norm(\svar^p))) \\
      &= \norm(\norm(\svar^p)) = \norm(\svar^p) = \norm(t)
    \end{align*}
  \item if $t \isdef \svar^p$ for some $\svar\in\pump{}$, $p \ge 1$ then
    \begin{align*}
      \norm(\sigma^{(-1)}(\norm(\sigma(t)))) &= \norm(\sigma^{(-1)}(\norm(\sigma(\svar^p)))) \\
      &= \norm(\sigma^{(-1)}(\norm((\svar^\sharp)^p{}))) \\
      &= \norm(\sigma^{(-1)}((\svar^\sharp)^{p \mod \period{\svar^\sharp}{}})) \\
      &= \norm(\svar^{p \mod \period{\svar^\sharp}{}}) \\
      &= \svar^{(p \mod \period{\svar^\sharp}{}) \mod \period{\svar}{}}
       = \svar^{p \mod \period{\svar}{}} = \norm(t)
    \end{align*}
    We used that $((x \mod
    q_1) \mod q_2) = (x \mod q_2)$ for any non-negative integers $x$,
    $q_1$, $q_2$ such that $q_2$ divides $q_1$.
  \item if $t \isdef \svar_1^{p_1} \cdot ... \cdot \svar_n^{p_n}$ for
    pairwise distinct $\svar_1,..., \svar_n\in \nopump{} \uplus
    \pump{}$, and arbitrary $p_1 \geq 1,...,p_n \geq 1$ then
    \begin{align*}
      \norm(\sigma^{(-1)}(\norm(\sigma(t)))) = \\
      \norm(\sigma^{(-1)}(\norm(\sigma(\svar_1^{p_1} \cdot ... \cdot \svar_n^{p_n})))) = \\
      \norm(\sigma^{(-1)}(\norm(\sigma(\svar_1^{p_1}) \cdot ... \cdot \sigma(\svar_n^{p_n})))) = \\
      \norm(\sigma^{(-1)}(\norm(\sigma(\svar_1^{p_1})) \cdot ... \cdot \norm(\sigma(\svar_n^{p_n})))) = \\
      \norm(\sigma^{(-1)}(\norm(\sigma(\svar_1^{p_1}))) \cdot ... \cdot \sigma^{(-1)}(\norm(\sigma(\svar_n^{p_n})))) = \\
      \norm(\norm(\sigma^{(-1)}(\norm(\sigma(\svar_1^{p_1})))) \cdot ... \cdot \norm(\sigma^{(-1)}(\norm(\sigma(\svar_n^{p_n})))) = \\
      \norm(\norm(\svar_1^{p_1}) \cdot ... \cdot \norm(\svar_n^{p_n})) = \norm(\norm(t)) = \norm(t)
    \end{align*}
    We used that normalization and substitution operate independently
    on pairwise distinct variables.
  \item if $t$ is arbitrary then, let denote by respectively $M(t)$,
    $M(\sigma(t))$ the set of unreduced monomials occurring in the
    (monomial) expansion of $t$ and $\sigma(t)$.  Observe that
    $M(\sigma(t)) = \{ \sigma(m) ~\mid~ m \in M(t)\}$, that is, $t$
    and $\sigma(t)$ have the same set of unreduced monomials up to the
    renaming of periodic variables.  Then,
    \begin{align*}
      \norm(\sigma^{(-1)}(\norm(\sigma(t))))
      &= \norm(\sigma^{(-1)}(\norm(\sum\nolimits_{m \in M(\sigma(t))} m))) \\
      &= \norm(\sigma^{(-1)}(\norm(\sum\nolimits_{m \in M(t)} \sigma(m)))) \\
      &= \norm(\sigma^{(-1)}(\sum\nolimits_{m \in M(t)}^{!} \norm(\sigma(m)))) \\
      &= \norm(\sum\nolimits_{m \in M(t)}^{!} \sigma^{(-1)}(\norm(\sigma(m)))) \\
      &= \sum\nolimits_{m \in M(t)}^{!} \norm(\sigma^{(-1)}(\norm(\sigma(m)))) \\
      &= \sum\nolimits_{m \in M(t)}^{!} \norm(m) = \norm(t)
      \end{align*}
  \end{itemize}
  We denoted by $\sum^{!}$ a sum that implicitly removes replicated
  terms.
\end{proofE}

%
%

Lemma~\ref{lemma:complexity:summary} below handles the general case of
linear products over both bounded and periodic variables (threshold
variables are not needed any further). This provides an explicit
polynomial cut-off on the length of such products, guaranteeing that
each longer product is either $0$ or equivalent in normal form to a
strict sub-product. The proof uses a reordering of product factors
that enables the joint application of
Lemmas~\ref{lemma:complexity:bounded-only} and
\ref{lemma:complexity:s0-bounded-periodic}. Let $\pump{} \isdef \set{
  \svar_1, \ldots, \svar_n}$ be the set of periodic variables and
define the bound:
\begin{align}\label{eq:bound}
  b(\nopump{},\pump{}) \isdef & ~\weightedcardof{\nopump{}} \cdot (\cardof{\pump{}} + 1) +
  \sum_{1\le i \le j\le n} \lcm(\period{\svar_j}{},\period{\svar_i}{}) -
  \cardof{\pump{}} \cdot (\cardof{\pump{}} + 1)/2
\end{align}

\begin{lemmaE}\label{lemma:complexity:summary}
  For each linear product $t \in \Prods{\Sums{\nopump{} \uplus
      \pump{}}}$ of lenght $\lengthof{t} > b(\nopump{},\pump{})$,
  either $\norm(t)=0$ or there exists $t' \ltprod t$ such that
  $\norm(t') = \norm(t)$.
\end{lemmaE}
\begin{proofE}
  Consequence of Lemmas~\ref{lemma:complexity:bounded-only} and
  \ref{lemma:complexity:s0-bounded-periodic}.  First, the term $t$
  can be reorganized as a product of $n+1$ sub-products $t_0 \cdot t_1
  \cdot ... \cdot t_{n}$ where:
  \begin{itemize}
  \item for every sum $\ell$ occurring in $t_0$ it holds
    $\varsof{\ell} \subseteq \nopump{}$,
  \item for every sum $\ell$ occurring in $t_i$ it holds $\svar_i
    \in \varsof{\ell} \subseteq \nopump{} \uplus \set{\svar_i,...,\svar_n}$.
  \end{itemize}
  Consequently, we obtain $t_0 \in \Prods{\Sums{\nopump{}}}$ and \[t_i \in
  \Prods{\XSums{\nopump{} \uplus \set{ \svar_i, \svar_{i+1},
      ... \svar_n}}{\svar_i}}\] for every $i\in \interv{1}{n}$.
  Using Lemma~\ref{lemma:complexity:bounded-only}, if
  $\lengthof{t_0} > b_0 \isdef \weightedcardof{\nopump{}}$ then $\norm(t_0)
  = 0$.  Using Lemma~\ref{lemma:complexity:s0-bounded-periodic} if
  $\lengthof{t_i} > b_i \isdef \weightedcardof{\nopump{}} + \sum_{j=i}^n
  (\lcm(\period{\svar_j}{}, \period{\svar_i}{}) - 1)$ then there exists
  $t_i' \ltprod t_i$ such that $\norm(t_i') = \norm(t_i)$, for every
  $i\in\interv{1}{n}$. Observe that $b(\nopump{},\pump{}) = \sum_{i=0}^n b_i$.

  We now complete the proof as follows.  Using the notations introduced
  above, the condition $\lengthof{t} > b(\nopump{},\pump{})$ is
  equivalent to $\sum_{i=0}^n \lengthof{t_i} > \sum_{i=0} b_i$.  Then,
  either:
  \begin{itemize}
  \item $\lengthof{t_0} > b_0$: then, $\norm(t_0) = 0$ and also
    $\norm(t) = 0$ or,
  \item there exists $i \in \interv{1}{n}$ such that $\lengthof{t_i} >
    b_i$: then, we define $t' \isdef t_0 \cdot ... \cdot t_{i-1} \cdot
    t_i' \cdot t_{i+1} \cdot ... \cdot t_n$ for some $t_i' \ltprod
    t_i$ such that $\norm(t_i') = \norm(t_i)$, which exists according
    to Lemma~\ref{lemma:complexity:s0-bounded-periodic} as explained
    above. It is an easy check that $t' \ltprod t$ and $\norm(t') =
    \norm(t)$.
  \end{itemize}
\end{proofE}

%
%

We define \(\Func(a,b) \isdef \sum_{i=0}^b (2^a-1)^i\) and observe
that $\Func(a,b) \le (2^{a} - 1 + 1)^b = 2^{a \cdot b}$. Intuitively,
$\Func(a,b)$ represents the number of syntactically distinct products
of length at most $b$, built from non-empty sums of $a$ distinct
variables. For a set of terms $U \subseteq \termuniv{}$, we write
$\nfcardof{U} \isdef \cardof{\norm(t) \mid t \in U,~ \norm(t) \neq
  0}$. Lemma~\ref{lemma:complexity:counting} below provides the bounds
on the cardinalities on the sets of linear products over both bounded
and periodic variables. This lemma simply applies the bounds on
lengths of linear products from
Lemmas~\ref{lemma:complexity:bounded-only},
\ref{lemma:complexity:1-bounded-periodic} and
\ref{lemma:complexity:summary} to obtain bounds on the number of
pairwise distinct normal forms of the linear products thereof.

\begin{lemmaE}\label{lemma:complexity:counting}
  \begin{enumerate}[(i)]
  \item\label{it1:lemma:complexity:counting}
    $\nfcardof{\Prods{\Sums{\nopump{}}}} \leq \Func(\cardof{\nopump{}}, \weightedcardof{\nopump{}})$,
  \item\label{it2:lemma:complexity:counting}
    $\nfcardof{\Prods{\XSums{\nopump{} \uplus \pump{}}{1}}} \leq \Func(\cardof{\nopump{} \uplus \pump{}}, \weightedcardof{\nopump{} \uplus \pump{}})$,
  \item\label{it3:lemma:complexity:counting}
    $\nfcardof{\Prods{\Sums{\nopump{} \uplus \pump{}}}} \leq \Func(\cardof{\nopump{} \uplus \pump{}},b(\nopump{}, \pump{}))$.
  \end{enumerate}
\end{lemmaE}
\begin{proofE}
  \noindent(\ref{it1:lemma:complexity:counting}) Using
  Lemma~\ref{lemma:complexity:bounded-only} all non-zero terms
  $\norm(t)$ can be constructed using products of at most
  $\weightedcardof{\nopump{}}$ sums.  As there are
  $2^{\cardof{\nopump{}}} - 1$ such sums, the result follows.

  \noindent(\ref{it2:lemma:complexity:counting}) Using
  Lemma~\ref{lemma:complexity:1-bounded-periodic} all terms
  $\norm(t)$ can be constructed using products of at most
  $\weightedcardof{\nopump{} \uplus \pump{}}$ 1-sums.  As there are
  $2^{\cardof{\nopump{} \uplus \pump{}}} - 1$ such sums, the result
  follows.

  \noindent(\ref{it3:lemma:complexity:counting}) Using
  Lemma~\ref{lemma:complexity:summary} all non-zero terms $\norm(t)$
  can be constructed using products of at most $b(\nopump{},\pump{})$
  non-empty sums.  As there are $2^{\cardof{\nopump{} \uplus \pump{}}}
  - 1$ such sums, the result follows. \qed
\end{proofE}

\subsection{The Main Result}

The next lemma gives first an exponential bound on the cardinality of
a recognizer built from an alternative grammar in normal form. We
distinguish the cases of (i) grammars with no periodic variables, (ii)
with periodic variables of period $1$, and (iii) the general case.

\begin{lemmaE}\label{lemma:complexity:alternative-form}
  The recognizer $\algof{A}$ built from an alternative grammar
  $\grammar \isdef (\mathcal{P} \uplus \mathcal{S},\rules,\axioms)$ in
  normal form satisfies the following constraints:
  \begin{enumerate}[(1)]
  \item\label{it1:lemma:complexity:alternative-form}
    $\cardof{\universeOf{A}^P} \le \begin{cases}
    2^{ B_{max} \cdot \cardof{\mathcal{S}}^2 \cdot \cardof{\mathcal{P}}} &
    \mbox{if } \arules = \emptyset \\
    2^{ 2 \cdot B_{max} \cdot \cardof{\mathcal{S}}^2 \cdot \cardof{\mathcal{P}}} &
    \mbox{if } \arules \not= \emptyset,~ P_{max} = 1 \\
    2^{(B_{max} +
      P_{max}^2/2) \cdot \cardof{\mathcal{S}}^2 \cdot
      (\cardof{\mathcal{S}} + 2) \cdot{\cardof{\mathcal{P}}}} &
    \mbox{if } \arules \not= \emptyset,~ P_{max} > 1
    \end{cases}$
  \item\label{it2:lemma:complexity:alternative-form}
    $\cardof{\universeOf{A}^S} \le 2^{ \cardof{\mathcal{S}} \cdot
    (\cardof{\mathcal{S}} + \cardof{\mathcal{P}} + 1)}$
  \end{enumerate}
  where $B_{max} \isdef \max\limits_{p\in \mathcal{P},s\in \nopump{p}}
  \base{s}{p}$ and $P_{max} \isdef \max\limits_{p\in \mathcal{P},s \in
    \pump{p}} \period{s}{p}$.
\end{lemmaE}
\begin{proofE}
  \noindent (\ref{it1:lemma:complexity:alternative-form}) We define:
  \[X_p \isdef \set{ \norm_p(t) ~|~ t \in\Prods{\Sums{\nopump{p} \uplus
        \pump{p}}, \norm_p(t) \not= 0}}\] By Definition
  \ref{def:p-view}, each element $a \in \universeOf{A}^P$ is of the
  form $a = \tuple{t_p}_{p \in \mathcal{P}}$, where each $t_p \in
  X_p$. Consequently, $\cardof{\universeOf{A}^P} \le \prod_{p \in
    \mathcal{P}} \cardof{X_p}$. We need to obtain an upper bound on
  $\cardof{X_p}$ for an arbitrary nonterminal $p \in \mathcal{P}$.
  Let us consider the set $X_p$, for some $p \in \mathcal{P}$ fixed.
  We distinguish three cases: \begin{itemize}
  \item $\arules = \emptyset$, i.e., no periodic variables exist.
    Then $X_p = \{\norm_p(t) ~|~ t\in \Prods{\Sums{\nopump{p}}},
    \norm_p(t) \not= 0\}$.  From
    Lemma~\ref{lemma:complexity:counting}(\ref{it1:lemma:complexity:counting})
    we obtain $\cardof{X_p} \le \Func(\cardof{\nopump{p}},
    \weightedcardof{\nopump{p}}) = 2^{ \cardof{\nopump{p}} \cdot
      \weightedcardof{\nopump{p}}}$.  Moreover $\cardof{ \nopump{p}}
    \le \cardof{\mathcal{S}}$ and $\weightedcardof{\nopump{p}} \le
    B_{max} \cdot \cardof{\mathcal{S}}$. Therefore, $\cardof{X_p} \le
    2^{B_{max} \cdot \cardof{\mathcal{S}}^2}$ and the upper bound on
    $\cardof{\universeOf{A}^P}$ follows.
  \item $\arules \not=\emptyset$ and $P_{max} = 1$, i.e., all periodic
    variables have period $1$. Since the normal form of a $1$-periodic
    variable is $1$, we have $X_p = \{ \norm_p(t_0 \cdot t_1) ~|~ t_0
    \in \Prods{\Sums{\nopump{p}}}, t_1 \in
    \Prods{\XSums{\nopump{p}}{1}}, \norm_p(t_0 \cdot t_1) \not= 0\}$.
    Let $X_p^{(0)} \isdef \{\norm_p(t_0) ~|~ t_0 \in
    \Prods{\Sums{\nopump{p}}, \norm_p(t_0) \not=0}\} $ and $X_p^{(1)}
    \isdef \{\norm_p(t_1) ~|~ t_1 \in
    \Prods{\XSums{\nopump{p}}{1}}\}$.  We have $\cardof{X_p} \le
    \cardof{X_p^{(0)}} \cdot \cardof{X_p^{(1)}}$.  From
    Lemma~\ref{lemma:complexity:counting}(\ref{it1:lemma:complexity:counting})
    we obtain $\cardof{X_p^{(0)}} \le \Func(\cardof{\nopump{p}},
    \weightedcardof{\nopump{p}})$ and from
    Lemma~\ref{lemma:complexity:counting}(\ref{it2:lemma:complexity:counting})
    we obtain $\cardof{X_p^{(1)}} \le \Func(\cardof{\nopump{p}},
    \weightedcardof{\nopump{p}})$.  As before,
    $\Func(\cardof{\nopump{p}}, \weightedcardof{\nopump{p}}) =
    2^{\cardof{\nopump{p}} \cdot \weightedcardof{\nopump{}}} \le
    2^{B_{max} \cdot \cardof{\mathcal{S}}^2}$ therefore $\cardof{X_p}
    \le 2^{2 \cdot B_{max} \cdot \cardof{\mathcal{S}}^2}$ and hence
    the upper bound on $\cardof{\universeOf{A}^p}$ follows.
  \item $\arules \not=\emptyset$ and $P_{max} > 1$, i.e., some periodic
    variables have periods greater than $1$. We have $X_p = \{ \norm_p(t_1 \cdot t_2) ~|~
    t_1 \in \Prods{\XSums{\nopump{p} \uplus \pumpx{p}}{1}},~ t_2 \in
    \Prods{\Sums{\nopump{p} \uplus \pumpx{p}}},$ $\norm(t_1 \cdot t_2)
    \not= 0 \}$, where $\pumpx{p} \isdef \pump{p} \setminus \{\svar \in
    \pump{p} ~|~ \period{p}{\svar}=1\}$.  Let $X_p^{(1)} \isdef \{
    \norm_p(t_1) ~|~ t_1 \in \Prods{\XSums{\nopump{p} \uplus
        \pumpx{p}}{1}}\}$ and $X_p^{(2)} \isdef \{ \norm_p(t_2) ~|~
    t_2 \in \Prods{\Sums{\nopump{p} \uplus \pumpx{p}}}, \norm(t_2)
    \not= 0 \}$.  We have $\cardof{X_p} \le \cardof {X_p^{(1)}} \cdot
    \cardof{X_p^{(2)}}$.  From
    Lemma~\ref{lemma:complexity:counting} (\ref{it2:lemma:complexity:counting})
    we have $\cardof{X_p^{(1)}} \le \Func(\cardof{\nopump{p} \uplus
      \pumpx{p}}, \weightedcardof{\nopump{p} \uplus \pumpx{p}})$ and
    from
    Lemma~\ref{lemma:complexity:counting} (\ref{it3:lemma:complexity:counting})
    we have $\cardof{X_p^{(2)}} \le \Func(\cardof{\nopump{p} \uplus
      \pumpx{p}}, b(\nopump{p},\pumpx{p}))$.  Then, we obtain: \begin{align*}
    \cardof{\nopump{p} \uplus \pumpx{p}} \le & ~\cardof{\mathcal{S}} \\
    \weightedcardof{\nopump{p} \uplus \pumpx{p}} \le & ~\max(B_{max},P_{max}) \cdot \cardof{\mathcal{S}} \\
    \Func(\cardof{\nopump{p} \uplus \pumpx{p}}, \weightedcardof{\nopump{p} \uplus \pumpx{p}}) \le & ~2^{
      \max(B_{max},P_{max}) \cdot \cardof{\mathcal{S}}^2}
    \end{align*}
    Using
    that $\lcm(\period{\svar_1}{p},\period{\svar_2}{p}) - 1 \le
    P_{max}^2$ for every variables $\svar_1, \svar_2 \in \pumpx{p}$
    and the definition of $b(\nopump{p},\pumpx{p})$ we obtain
    $b(\nopump{p},\pumpx{p}) \le B_{max} \cdot \cardof{\mathcal{S}}
    \cdot (\cardof{ \mathcal{S}} + 1) + P_{max}^2 \cdot
    \cardof{\mathcal{S}} \cdot (\cardof{\mathcal{S}} + 1) / 2 =
    (B_{max} + P_{max}^2/2) \cdot \cardof{\mathcal{S}} \cdot
    (\cardof{\mathcal{S}} + 1)$.  Then $K(\cardof{\nopump{p} \uplus
      \pumpx{p}}, b(\nopump{p}, \pumpx{p})) \le 2^{(B_{max} +
      P_{max}^2/2) \cdot \cardof{\mathcal{S}}^2 \cdot
      (\cardof{\mathcal{S}} + 1)}$.  We compute:
    \begin{align*}
    \cardof{X_p} \le & ~2^{\max(B_{max},P_{max}) \cdot \cardof{\mathcal{S}}^2} \cdot
    2^{(B_{max} + P_{max}^2/2) \cdot \cardof{\mathcal{S}}^2 \cdot (\cardof{\mathcal{S}} + 1)} \\
    \le & ~2^{(B_{max}+P_{max}^2/2) \cdot \cardof{\mathcal{S}} \cdot (\cardof{\mathcal{S}} + 2)}
    \end{align*}
    because
    $\max(B_{max},P_{max}) \le B_{max} + P_{max} \le B_{max} +
    P_{max}^2/2$, and the upper bound on $\cardof{\universeOf{A}^P}$
    follows.
  \end{itemize}

  \noindent (\ref{it2:lemma:complexity:alternative-form}) Immediate,
  from the definition of $\universeOf{A}^S$ (Definition \ref{def:s-view}).
\end{proofE}

We denote by $\maxtermsizeof{\grammar}$ the maximum size of a term
that occurs as the right-hand size of a rule in $\grammar$. The result
of this section is that the cardinality of a recognizer built from a
regular grammar $\grammar$ is bounded by an exponential in the number
of nonterminals, rules and $\maxtermsizeof{\grammar}$:

\begin{theorem}\label{thm:main}
  The recognizer $\algof{A}$ built from the regular grammar $\grammar
  = (\mathcal{P} \uplus \mathcal{S},\rules,\axioms)$ has
  $2^{\bigO(\cardof{\mathcal{P}}
    \cdot \cardof{\mathcal{S}}^3 \cdot \cardof{\mathcal{R}}^3 \cdot \maxtermsizeof{\grammar}^2)}$
  elements.
\end{theorem}
\begin{proof}
  First, the regular grammar $\grammar$ is transformed to an
  equivalent regular grammar $\grammar' \isdef (\mathcal{P} \uplus
  \mathcal{S}',\rules',\axioms)$ in normal form, such that $\cardof{\mathcal{S}'} =
  \bigO(\cardof{\mathcal{S}} \cdot \cardof{\rules})$, by
  Proposition~\ref{prop:sp-normal-form}. Second, the regular grammar
  in normal form $\grammar'$ is transformed to an alternative grammar
  $\grammar'' \isdef (\mathcal{P} \uplus
  \mathcal{S}'',\rules'',\axioms)$ in normal form as explained in
  section~\ref{subsec:regular-languages-are-recognizable}.  The set of
  non-terminals $\mathcal{S''}$ is increased by the elimination of
  the rules $\erules$. This increase is bounded by the cardinality of
  $\erules$, i.e., $\cardof{\mathcal{S}''} \le
  \cardof{\mathcal{S'}} + \cardof{\rules}$. Also, the transformation
  guarantees that the values of maximal periods and/or bases are the
  same in $\grammar'$ and $\grammar''$. The recognizer $\algof{A}$ is
  built from $\grammar''$. By
  Lemma~\ref{lemma:complexity:alternative-form} applied to
  $\grammar''$ we obtain:
  \begin{align*}
    \cardof{\universeOf{A}^P} & \le u^P \isdef
    2^{(B_{max} + P_{max}^2/2) \cdot \cardof{\mathcal{S}''}^2 \cdot
      (\cardof{\mathcal{S}''} + 2) \cdot \cardof{\mathcal{P}}} \\
    \cardof{\universeOf{A}^S} & \le u^S \isdef
    2^{\cardof{\mathcal{S}''} \cdot (\cardof{\mathcal{S}''} + \cardof{\mathcal{P}})}
  \end{align*}
  Then, $\cardof{ \universeOf{A}} = \cardof{\universeOf{A}^P} +
  \cardof{\universeOf{A}^S} \le u^P + u^S \le 2 \cdot u^P$, since $u^S
  \le u^P$ always holds. Using moreover that $B_{max} \le
  \maxtermsizeof{\grammar}$, $P_{max} \le \maxtermsizeof{\grammar}$
  and $\cardof{\mathcal{S}''}  = \bigO(\cardof{\mathcal{S}} \cdot \cardof{\rules})$ we obtain
  $\cardof{\universeOf{A}} = 2^{\bigO(\cardof{\mathcal{P}} \cdot
    \cardof{\mathcal{S}}^3 \cdot \cardof{\mathcal{R}}^3 \cdot
    \maxtermsizeof{\grammar}^2)}$.
\end{proof}
We simplify the above bound to a more rough bound in the size of the
grammar:

\begin{corollary}\label{cor:main}
  The recognizer $\algof{A}$ built from the regular grammar $\grammar$
  has $2^{\bigO(\sizeof{\grammar}^9)}$ elements.
\end{corollary}
\begin{proof}
  Because $\sizeof{\grammar}$ is the length of the representation of
  $\grammar$ as a string, we have $\sizeof{\grammar} = 2n +
  \sum_{i=1}^n \sizeof{t_i}$, where the rules of $\grammar$ are $x_i
  \rightarrow t_i$, for $1 \leq i \leq n$. Moreover, by assuming that
  each nonterminal occurs as left-hand side in at least one rule, we
  obtain $\cardof{\mathcal{P}}+\cardof{\mathcal{S}} \leq n$. We compute:
  \begin{align*}
    \cardof{\mathcal{P}} \cdot \cardof{\mathcal{S}}^3 \cdot \cardof{\mathcal{R}}^3 \cdot \maxtermsizeof{\grammar}^2 \leq \sizeof{\grammar}^9
  \end{align*}
  \vspace*{-\baselineskip}
\end{proof}
As a byproduct, we obtain the following result: 
\begin{corollary}\label{cor:accepting}
  The accepting set of $\langof{}{\grammar}$ in $\algof{A}$ is
  computable in time $2^{\bigO(\cardof{\mathcal{P}} \cdot
    \cardof{\mathcal{S}}^3 \cdot \cardof{\mathcal{R}}^3 \cdot
    \maxtermsizeof{\grammar}^2)}$.
\end{corollary}
\begin{proof}
  Let $F \subseteq \universeOf{A}$ such that $\langof{}{\grammar} =
  h^{-1}(F)$. The elements of $F \cap \universeOf{A}^P$ can be
  enumerated in time $\cardof{\universeOf{A}} \cdot \cardof{\rules}
  \cdot \maxtermsizeof{\grammar}^{\bigO(\cardof{\mathcal{S}}})$, by
  Lemmas \ref{lemma:accepting} and \ref{lemma:match}. The elements of
  $F \cap \universeOf{A}^S$ can be enumerated in time
  $\cardof{\universeOf{A}} \cdot \cardof{\mathcal{S}} \cdot
  (\cardof{\mathcal{P}} + \cardof{\mathcal{S}} + 1)$, because
  $\cardof{\mathcal{S}} \cdot (\cardof{\mathcal{P}} +
  \cardof{\mathcal{S}} + 1)$ is the maximum size of an S-profile. By
  Theorem \ref{thm:main}, we obtain an upper bound of the order of
  $2^{\bigO(\maxtermsizeof{\grammar}^2 \cdot \cardof{\mathcal{P}}
    \cdot \cardof{\mathcal{S}}^3 \cdot \cardof{\mathcal{R}}^3)}$, in
  both cases.
\end{proof}

\section{Decision Problems}

In this section, we consider standard decision problems for
recognisable sets of SP-graphs, represented as regular grammars
(Definition \ref{def:sp-regular-grammar}). We prove that the problems
of (i) emptiness of the intersection of $n$ sets given as regular
grammars, and (ii) inclusions between a context-free and a regular
set, are both \exptime-complete.

We denote by $\maxarityof{\grammar}$ the maximum number of variables
occurring on the right-hand side of a rule from the grammar $\grammar$
and recall that $\maxtermsizeof{\grammar}$ denotes the maximum size of
a right-hand side of $\grammar$.  A first ingredient is the well-known
Filtering Theorem~\cite[Theorem 3.88]{courcelle_engelfriet_2012}:

\begin{theorem}\label{thm:filtering}
  For each grammar $\grammar=(\nonterm,\rules,\axioms)$ and set
  $\langu$ of SP-graphs recognized by a finite algebra $\algof{A}$,
  one can build a grammar $\grammar'$ such that $\langof{}{\grammar'}
  = \langof{}{\grammar} \cap \langu$. Moreover, $\grammar'$ can be
  built in time $\bigO(\cardof{\rules} \cdot
  \cardof{\universeOf{A}}^{\maxarityof{\grammar}} \cdot \tau \cdot
  \maxtermsizeof{\grammar})$, where $\tau$ is an upper bound on the
  time needed to compute $f^\algof{A}(a_1,\ldots,a_n)$, for each $f
  \in \signature_\algof{SP}$ and $a_1,\ldots,a_n \in \universeOf{A}$.
\end{theorem}
Intuitively, the grammar $\grammar'$ is obtained from $\grammar$ by
annotating each nonterminal $x \in \nonterm$ with an element $a \in
\universeOf{A}$ and introducing a rule $x^a \rightarrow t[x_1^{a_1},
  \ldots, x_n^{a_n}]$ for each rule $x \rightarrow t[x_1,\ldots,x_n]
\in \rules$ and tuple $a,a_1,\ldots,a_n \in \universeOf{A}$ such that
$a = t^\algof{A}(a_1,\ldots,a_n)$. The axioms of $\grammar'$ are
$\set{x^a \mid x \in \axioms, a \in h(\langu)}$, where $h$ is the
homomorphism recognizing $\langu$ (Definition \ref{def:rec}). The
upper bound on the time needed to compute $\grammar'$ follows
immediately from its construction.

To apply the Filtering Theorem to recognizable sets defined by regular
grammars, we use the previously computed upper bound on the
cardinality of their recognizers (Theorem \ref{thm:main}). An estimate
on the time $\tau$ needed to compute the operations from
$\signature_\algof{SP}$ in the recognizer of a regular grammar is
given below:

\begin{lemmaE}\label{lemma:omega}
  Let $\grammar=(\mathcal{P}\uplus\mathcal{S},\rules,\axioms)$ be a
  regular grammar and $\algof{A}$ be a recognizer for
  $\langof{}{\grammar}$. Then, $f^\algof{A}(a_1,\ldots,a_n)$ can be
  computed in time
  $\tau=(\cardof{\mathcal{P}}\cdot\cardof{\rules})^{\bigO(1)}
  \cdot \maxtermsizeof{\grammar}^{\bigO(\cardof{\mathcal{S}})}$, for
  each $f \in \signature_\algof{SP}$ and $a_1,\ldots,a_n \in
  \universeOf{A}$.
\end{lemmaE}
\begin{proofE}
  If $a = \tuple{t_p}_{p\in\mathcal{P}} \in \universeOf{A}^P$ is a
  P-profile, the size of each monomial from $t_p$ is at most
  $\cardof{\mathcal{S}} \cdot \maxtermsizeof{\grammar}$ and $t_p$ has
  at most $\maxtermsizeof{\grammar}^{\cardof{\mathcal{S}}}$
  monomials. Else, if $a \in \universeOf{A}^S$, there are at most
  $\cardof{\mathcal{S}} \cdot (\cardof{\mathcal{P}} +
  \cardof{\mathcal{S}} + 1)$ pairs in $a$. We distinguish the
  following cases, according to the type of $f \in
  \signature_\algof{SP}$: \begin{itemize}
  \item $a^\algof{A}$, for some $a \in \alphabet$: this case requires
    at most $\cardof{\rules}$ steps.
  \item $a_1 \pop^\algof{A} a_2$: in this case, computing each
    $\parmap(a_i)$ requires at most $\cardof{\mathcal{P}} \cdot
    \cardof{\mathcal{S}}^2 \cdot (\cardof{\mathcal{P}} +
    \cardof{\mathcal{S}} + 1)$ steps and, for each $p \in
    \mathcal{P}$, computing $\norm_p(\tuple{\parmap(a_1)}_p \cdot
    \tuple{\parmap(a_2)}_p)$ takes:
    \begin{align*}
    \max\{\cardof{\mathcal{P}} \cdot \cardof{\mathcal{S}}^2 \cdot
      (\cardof{\mathcal{P}} + \cardof{\mathcal{S}} + 1), \maxtermsizeof{\grammar}^{\cardof{\mathcal{S}}}\}^2 =
      \cardof{\mathcal{P}}^4 \cdot \maxtermsizeof{\grammar}^{\bigO(\cardof{\mathcal{S}})}
    \end{align*}
    steps, by Lemma \ref{lemma:match}. Hence, the entire computation
    of $a_1 \pop^\algof{A} a_2$ is of the order of
    $\cardof{\mathcal{P}}^5 \cdot
    \maxtermsizeof{\grammar}^{\bigO(\cardof{\mathcal{S}})}$.
  \item $a_1 \sop^\algof{A} a_2$: in this case, computing
    each $\seqmap(a_i)$ takes $\cardof{\mathcal{P}} \cdot
    \cardof{\mathcal{R}}^2 \cdot
    \maxtermsizeof{\grammar}^{\bigO(\cardof{\mathcal{S}})}$ steps, by
    Lemma \ref{lemma:match}. Hence, the entire computation of $a_1
    \sop^\algof{A} a_2$ is of the order of $\cardof{\mathcal{P}}^3
    \cdot \cardof{\rules}^6 \cdot
    \maxtermsizeof{\grammar}^{\bigO(\cardof{\mathcal{S}})}$.
  \end{itemize}
  We obtain the uniform upper bound of
  $(\cardof{\mathcal{P}}\cdot\cardof{\rules})^{\bigO(1)}
  \cdot \maxtermsizeof{\grammar}^{\bigO(\cardof{\mathcal{S}})}$, for
  each $f \in\signature_\algof{SP}$.
\end{proofE}

We establish the complexity of two decision problems concerning
recognizable sets, described by regular grammars. This tightens an
existing \twoexptime{} upper bound, for the problem of inclusion
between a context-free and a regular language~\cite[Theorem
  4]{Lics25}. The \exptime{} upper bound follows from the fact that
the language of a regular grammars is recognized by an algebra of
cardinality exponentially bounded in the size of the grammar (Theorems
\ref{thm:reg-rec} and \ref{thm:main}). The \exptime-hard lower bound
for the problems considered in Theorem \ref{thm:reg-grammars-problems}
is established by polynomial many-one reductions for the problems of
intersection and inclusion of tree automata~\cite[Theorems 1.7.5 and
  1.7.7]{comon:hal-03367725}.

\begin{theorem}\label{thm:reg-grammars-problems}
  The following problems are \exptime-complete: \begin{enumerate}
  \item\label{it1:thm:reg-grammars-problems} Given regular grammars
    $\grammar_1, \ldots, \grammar_n$, does $\langof{}{\grammar_1} \cap
    \ldots \cap \langof{}{\grammar_n} = \emptyset$ hold?
  \item\label{it2:thm:reg-grammars-problems} Given grammars
    $\grammar_1$ and $\grammar_2$ such that $\grammar_2$ is regular,
    does $\langof{}{\grammar_1} \subseteq \langof{}{\grammar_2}$ hold?
  \end{enumerate}
\end{theorem}
\begin{proof}
  We prove the \exptime{} upper bound for each of the considered
  problems:

  \vspace*{.5\baselineskip}\noindent(\ref{it1:thm:reg-grammars-problems})
  Let $\grammar_1,\ldots,\grammar_n$ be an instance of the
  intersection problem, where $\grammar_i = (\mathcal{P}_i \uplus
  \mathcal{S}_i, \rules_i, \axioms_i)$, for each $i \in
  \interv{1}{n}$. We denote $P \isdef \sum_{i=1}^n
  \cardof{\mathcal{P}_i}$, $S \isdef \sum_{i=1}^n
  \cardof{\mathcal{S}_i}$, $R \isdef \sum_{i=1}^n \cardof{\rules_i}$
  and $\Theta \isdef \sum_{i=1}^n \maxtermsizeof{\grammar_i}$. The
  size of the input is then $\sum_{i=1}^n \sizeof{\grammar_i} \leq N
  \isdef P + S + R \cdot \Theta$. Let $\algof{A}_i$ be a recognizer
  for $\langof{}{\grammar_i}$, for $i \in \interv{1}{n}$. By
  Proposition \ref{prop:rec-bool-closure}, the set
  $\langof{}{\grammar_1} \cap \ldots \cap \langof{}{\grammar_n}$ is
  recognized by an algebra $\algof{A}$ having the domain
  $\universeOf{A} = \universeOf{A}_1 \times \ldots \times
  \universeOf{A}_n$. By Theorem \ref{thm:main}, we obtain
  $\cardof{\universeOf{A}} = 2^{\bigO(P \cdot S^3 \cdot R^3 \cdot
    \Theta^2)}$. By Corollary \ref{cor:accepting}, the accepting set
  for the language $\langof{}{\grammar_1} \cap \ldots \cap
  \langof{}{\grammar_n}$ can be computed in time that is of the same
  order of magnitude as $\cardof{\universeOf{A}}$. Since $P \cdot S^3
  \cdot R^3 \cdot \Theta \leq N^7$, the emptiness of the intersection
  can be decided in time $2^{\bigO(N^7)}$.

  \vspace*{.5\baselineskip}\noindent(\ref{it2:thm:reg-grammars-problems})
  Let $\grammar_i = (\mathcal{P}_i \uplus \mathcal{S}_i, \rules_i,
  \axioms_i)$, for $i=1,2$, be an instance of the inclusion
  problem. Let $\algof{A}_2$ be a recognizer for
  $\langof{}{\grammar_2}$. The size of the input is then $N \isdef
  \sum_{i=1}^2 \sizeof{\grammar_i}$. By Proposition
  \ref{prop:rec-bool-closure}, the set $\universeOf{SP} \setminus
  \langof{}{\grammar_2}$ is also recognized by $\algof{A}_2$. Let
  $\grammar'_1$ be the result of applying the Filtering Theorem
  \ref{thm:filtering} to $\grammar_1$ and $\universeOf{SP} \setminus
  \langof{}{\grammar_2}$, i.e., $\langof{}{\grammar'_1} =
  \langof{}{\grammar_1} \cap (\universeOf{SP} \setminus
  \langof{}{\grammar_2})$. By Theorem \ref{thm:filtering} and Lemma
  \ref{lemma:omega}, $\grammar'_1$ can be built in time:
  \begin{align*}
    & ~\cardof{\rules_1} \cdot \cardof{\universeOf{A}_2}^{\maxarityof{\grammar_1}} \cdot
    \underbrace{(\cardof{\mathcal{P}_2} \cdot  \cardof{\rules_2})^{\bigO(1)} \cdot \maxtermsizeof{\grammar_2}^{\bigO(\mathcal{S}_2)}}_\tau ~\cdot~
    \maxtermsizeof{\grammar_1} \\
    = & ~\cardof{\rules_1} \cdot
    \underbrace{2^{\maxarityof{\grammar_1}\cdot\bigO(\cardof{\mathcal{P}_2} \cdot \cardof{\mathcal{S}_2}^3 \cdot \cardof{\rules_2}^3 \cdot \maxtermsizeof{\grammar_2}^2)}}_{\cardof{\universeOf{A}_2}^{\maxarityof{\grammar_1}}} ~\cdot \\
    & ~\underbrace{(\cardof{\mathcal{P}_2} \cdot \cardof{\rules_2})^{\bigO(1)} \cdot \maxtermsizeof{\grammar_2}^{\bigO(\mathcal{S}_2)}}_\tau ~\cdot~
    \maxtermsizeof{\grammar_1} \\
    = & ~2^{(\cardof{\rules_1} \cdot \maxarityof{\grammar_1} \cdot \cardof{\mathcal{P}_2} \cdot \cardof{\mathcal{S}_2} \cdot \cardof{\rules_2} \cdot \maxtermsizeof{\grammar_2})^{\bigO(1)}} = 2^{N^{\bigO(1)}}
  \end{align*}
  Since the emptiness of $\langof{}{\grammar'_1}$ can be decided in
  time $\bigO(\sizeof{\grammar'_1}^2)$, the inclusion problem has an
  \exptime{} upper bound. The \exptime-hard lower bound follows from
  the \exptime-hardness of the same problems for sets of terms
  represented by finite tree automata~\cite[Theorems 1.7.5 and
    1.7.7]{comon:hal-03367725}. The encoding of terms as SP-graphs and
  of bottom-up tree automata as regular grammars are detailed in the
  proof of \cite[Theorem 4]{Lics25}.
\end{proof}

\enlargethispage{4mm}

\section{Conclusions}

We prove an exponential bound on the cardinality of the minimal
algebra recognizing the language of regular grammar, for
series-parallel graphs. Based on this result, we establish the
\exptime-completeness of the intersection and inclusion problems.

\bibliographystyle{plainurl}
\bibliography{bib2doi}

\end{document}